\pdfoutput=1
\documentclass[sigconf,screen,nonacm]{acmart}

\usepackage[utf8]{inputenc}
\usepackage{tikz}

\usepackage[bibliography=common]{apxproof}
\usepackage{array,booktabs}
\usepackage{relsize}
\usetikzlibrary{decorations.markings}
\usetikzlibrary{decorations.pathreplacing,calc}
\usepackage[shortlabels]{enumitem}
\usepackage{stackengine}
\usepackage{nicefrac}
\stackMath

\AtBeginDocument{%
  \providecommand\BibTeX{{%
    \normalfont B\kern-0.5em{\scshape i\kern-0.25em b}\kern-0.8em\TeX}}}

\copyrightyear{2020}
\acmYear{2020}
\setcopyright{acmcopyright}
\acmConference[PODS '20] {39th ACM SIGMOD-SIGACT-SIGAI Symposium on Principles of Database Systems}{June 30--July 5, 2019}{Amsterdam, Netherlands}
\acmBooktitle{39th ACM SIGMOD-SIGACT-SIGAI Symposium on Principles of Database Systems (PODS '19), June 30--July 5, 2019, Amsterdam, Netherlands}
\acmPrice{15.00}
\acmDOI{10.1145/3294052.3319680}
\acmISBN{978-1-4503-6227-6/19/06}

\newtheoremrep{theorem}{Theorem}[section]
\newtheoremrep{proposition}[theorem]{Proposition}
\newtheoremrep{lemma}[theorem]{Lemma}
\newtheoremrep{example}[theorem]{Example}
\newtheoremrep{claim}[theorem]{Claim}
\newtheoremrep{observation}[theorem]{Observation}

\definecolor{mygreen}{rgb}{0.09, 0.45, 0.27}

\renewcommand{\phi}{\varphi}
\renewcommand{\epsilon}{\varepsilon}
\renewcommand{\leq}{\leqslant}
\renewcommand{\geq}{\geqslant}
\newcommand\defeq{\stackrel{\mathrm{def}}{=}}

\newcommand{\dom}{\mathsf{dom}}

\newcommand{\consts}{\mathsf{Consts}}
\newcommand{\nulls}{\mathsf{Nulls}}
\newcommand{\N}{\mathbb{N}}
\newcommand{\arity}{\mathsf{arity}}
\renewcommand{\Pr}{\mathsf{Pr}}
\newcommand{\surj}{\mathsf{surj}}
\newcommand{\sig}{\mathsf{sig}}

\newcommand{\countcompls}{\text{$\#$}\text{\sf Comp}}
\newcommand{\countvals}{\text{$\#$}\text{\sf Val}}
\newcommand{\ucountcompls}{\text{$\#$}\text{{\sf Comp}$^{\textit{u}}$}}
\newcommand{\ucountvals}{\text{$\#$}\text{{\sf Val}$^{\textit{u}}$}}
\newcommand{\ccountcompls}{\text{$\#$}\text{{\sf Comp}$_{\textit{Cd}}$}}
\newcommand{\ccountvals}{\text{$\#$}\text{{\sf Val}$_{\textit{Cd}}$}}
\newcommand{\cucountcompls}{\text{$\#$}\text{{\sf Comp}$^{\textit{u}}_{\textit{Cd}}$}}
\newcommand{\cucountvals}{\text{$\#$}\text{{\sf Val}$^{\textit{u}}_{\textit{Cd}}$}}

\newcommand{\avoidance}{\mathrm{\#Avoidance}}

\newcommand{\hamsubgraphs}{\mathrm{\#HamSubgraphs}}
\newcommand{\shpf}{\mathrm{\#PF}}
\newcommand{\ksat}{\mathrm{\#k3SAT}}

\newcommand{\mc}{\mathrm{MC}}

\newcommand{\IS}{\text{\rm IS}}
\newcommand{\VC}{\mathrm{VC}}
\newcommand{\checkp}{\mathrm{check}}
\newcommand{\I}{\mathcal{I}}
\newcommand{\T}{\mathrm{T}}
\newcommand{\s}{\mathrm{s}}
\newcommand\mA{\mathbf{A}}
\newcommand\mZ{\mathbf{Z}}
\newcommand\mC{\mathbf{C}}

\newcommand{\sjfbcq}{\text{\rm sjfBCQ}}
\newcommand{\sjfbcqs}{\text{\rm sjfBCQs}}
\newcommand{\pr}{\leq_{\mathrm{par}}^{\mathrm{p}}}

\newcommand{\ptime}{\text{\rm P}}
\newcommand{\fp}{\text{\rm FP}}
\newcommand{\rk}{\mathrm{rk}}
\newcommand{\shp}{\text{$\#$}\text{\rm P}}
\newcommand{\spanl}{\text{\rm SpanL}}
\newcommand{\spanp}{\text{\rm SpanP}}
\newcommand{\nl}{\text{\rm NL}}
\newcommand{\np}{\text{\rm NP}}

\newcommand{\up}{\text{\rm UP}}
\newcommand{\rp}{\text{\rm RP}}
\newcommand{\bpp}{\text{\rm BPP}}
\newcommand{\spp}{\text{\rm SPP}}
\newcommand{\CC}{\mathcal{C}}
\newcommand{\sIS}{\text{$\#$}\text{\rm IS}}
\newcommand{\sVC}{\text{$\#$}\text{\rm VC}}
\newcommand{\tr}{\leq_{\mathrm{T}}^{\mathrm{p}}}
\newcommand{\acc}{\text{\rm accept}}
\newcommand{\rej}{\text{\rm reject}}
\newcommand{\gap}{\text{\rm gap}}
\newcommand{\gapp}{\text{\rm GapP}}
\newcommand{\gapspanp}{\text{\rm GapSpanP}}

\newcommand{\boundellipse}[3]
{(#1) ellipse (#2 and #3)
}

\begin{document}

\title{Counting Problems over Incomplete Databases}


\author{Marcelo Arenas}
\affiliation{%
  \institution{PUC \& IMFD Chile}
  }
\email{marenas@ing.puc.cl}

\author{Pablo Barceló}
\affiliation{%
  \institution{PUC \& IMFD Chile}
  }
\email{pbarcelo@ing.puc.cl}

\author{Mikaël Monet}
\affiliation{%
  \institution{IMFD Chile}
  }
\email{mikael.monet@imfd.cl}

\renewcommand{\shortauthors}{Arenas, Barcel\'o and Monet}

\begin{abstract}

We study the complexity of various fundamental counting problems that arise in the context
of incomplete databases, i.e., relational databases that can contain unknown
values in the form of labeled nulls. Specifically, we assume that the domains of these unknown values are
finite and, for a Boolean query~$q$, we consider the following two problems:
given as input an incomplete database~$D$, (a) return the number of completions
of~$D$
that satisfy~$q$; or (b) return or the number of valuations of the nulls of~$D$
yielding a completion that satisfies~$q$.  We obtain dichotomies
between~\#P-hardness and polynomial-time computability for these problems when~$q$ is
a self-join--free conjunctive query, and study the impact on the complexity of
the following two restrictions: (1) every null occurs at most once in~$D$ (what is
called \emph{Codd tables});
and (2) the domain of each null is the same.
Roughly speaking, we show that counting
completions is much harder than counting valuations (for instance, while the
latter is always in~\#P, we prove that the former is not in~\#P under some
widely believed theoretical complexity assumption). Moreover, we find that both (1) and (2) reduce the
complexity of our problems. We also
study the approximability of these problems and show that, while counting
valuations always has a fully polynomial randomized approximation scheme, in most cases counting completions does not.
Finally, we consider more expressive query languages and situate our problems with
respect to known complexity classes.

\end{abstract}

%

\keywords{Incomplete databases, closed-world assumption, counting complexity, FPRAS}


\settopmatter{printfolios=true} 
\maketitle

\section{Introduction}
\label{sec:introduction}

\noindent 
\paragraph{{\bf Context}}
In the database literature, {\em incomplete databases} are often used to
represent missing information in the data; see,
e.g.,~\cite{abiteboul1995foundations,van1998logical,libkin2014incomplete}.   These are
traditional relational databases whose active domain can contain both constants
and \emph{nulls}, the latter representing unknown
values~\cite{imielinski1984incomplete}.  There are many ways in which one can
define the semantics of such a database, each being equally meaningful
depending on the intended application.  Under the so called \emph{closed-world
assumption}~\cite{abiteboul1995foundations,reiter1978closed}, a
standard, complete database~$\nu(D)$ is obtained from an incomplete database~$D$ by applying a {\em
valuation}~$\nu$ which replaces each null~$\bot$ in~$D$ with a
constant~$\nu(\bot)$.
The goal is then
to reason about the space formed by all valuations~$\nu$ and completions~$\nu(D)$ of~$D$.

Decision problems related to querying incomplete databases have been well
studied already.  Consider for instance the problem {\sf
Certainty}$(q(\bar{x}))$ 
which, for a fixed query~$q(\bar{x})$, takes as input an incomplete
database~$D$ and a tuple~$\bar{a}$ and asks whether~$\bar{a}$ is an answer to~$q$
for every 
possible completion/valuation of~$D$. 
 By now, we have a deep understanding of the complexity of these kind of
 decision problems for different choices of query languages, including
 conjunctive queries (CQs) and FO
 queries~\cite{imielinski1984incomplete,AbiteboulKG91}.  However, having the
 answer to this question is sometimes of little help: what if it is not the case
 that~$q$ is certain on~$D$? Can we still infer some useful information? This
 calls for new notions that could be used to measure the certainty with
 which~$q$ holds, notions which should be finer than those previously
 considered. This is for instance what the recent work in \cite{libkin2018certain} does by
 introducing a notion of~\emph{best answer}, which are those tuples~$\bar a$ for which the set of completions of~$D$ 
 over which~$q(\bar a)$ holds is maximal with respect to set inclusion.

A fundamental complementary approach to address this issue can be obtained by considering 
some {\em counting problems}
  related to incomplete databases; 
  more specifically,  
  determining the number of
  completions/valuations of an incomplete database that satisfy a given
  Boolean query~$q$.  
  These problems are relevant as they tell us,
  intuitively, how close is~$q$ from being certain over~$D$, i.e., what is the
  level of {\em support} that~$q$ has over the set of completions/valuations
  of~$D$.  Surprisingly,
  such counting problems do not seem to have been studied for
  incomplete databases. A reason for this omission in the literature
  might be that, in general, it is assumed that the domain over which nulls can
  be interpreted is infinite, and thus incomplete databases might have an
  infinite number of completions/valuations.  However, 
  in many scenarios it is natural to assume that the domain over which nulls are
  interpreted is finite, in particular when dealing with uncertainty in practice~\cite{andritsos2006clean,BSHW06,FGJK08,AntovaKO09,SBHNW09,ASUW10}.
  By assuming this we can ensure that the number of
  completions and valuations
  are always finite, and thus
  that they can be counted.  This is the setting that we study.

In this paper, we will focus only on Boolean queries.
  We consider this as a necessary first step in understanding the complexity of some relevant counting problems over incomplete databases. Besides, as shown in the paper, the case of Boolean queries is rich and complex enough to deserve its own investigation, providing very valuable information for the general case where queries can have arbitrary tuples as answers. 

\paragraph{{\bf Problems studied}}
We focus on the problems \countcompls$(q)$ and \countvals$(q)$ for a Boolean query~$q$, which take as input an incomplete 
database~$D$ together with a finite set~$\dom(\bot)$ of constants for every null~$\bot$ occurring in~$D$,
and ask the following: How many completions, resp., valuations, 
of~$D$ satisfy~$q$? More formally,
a \emph{valuation}~$\nu$ of~$D$ is a mapping that associates to every null~$\bot$ of~$D$ a constant~$\nu(\bot)$ in~$\dom(\bot)$.
Then, given a valuation~$\nu$ of~$D$, we 
denote by~$\nu(D)$ the database that is obtained from~$D$ after replacing each null~$\bot$ with~$\nu(\bot)$.
Besides, in this paper we consider set semantics, so repeated tuples have to be removed from~$\nu(D)$.
For~\countcompls$(q)$ we count all databases of the form~$\nu(D)$ such that 
$q$ holds in~$\nu(D)$. Instead, for \countvals$(q)$ we count the number of valuations~$\nu$ such that~$q$ holds in 
$\nu(D)$. It is easy to see that these two values can differ, as there might be different valuations~$\nu,\nu'$ 
of~$D$ leading to the same completion, i.e.,~$\nu(D) = \nu'(D)$. We 
think
that both problems are 
meaningful: while \countcompls$(q)$ determines the support for~$q$ over the databases represented by~$D$, 
we have that \countvals$(q)$ further refines this 
by incorporating the support for a particular completion 
that satisfies~$q$ over the set of valuations for~$D$.  

\begin{table*}[t]
  \centering
  \begin{tabular}{|l|m{0.185\textwidth}|m{0.185\textwidth}|m{0.185\textwidth}|m{0.185\textwidth}|}
				\hline
				&  \multicolumn{2}{ c| }{{\bf Counting valuations}} & \multicolumn{2}{ c| }{{\bf Counting completions}} \\
				\hline
				\hline
				& {\bf Non-uniform} & {\bf Uniform} & {\bf Non-uniform} & {\bf Uniform} \\
				\hline
				{\bf Naïve} &
				\begin{tabular}{l} $R(x,x)$ \\ $R(x) \land  S(x)$ \end{tabular}
& \begin{tabular}{l} \vspace{-4pt} \\$R(x,x)$ \\ $R(x) \land S(x,y) \land T(y)$ \\ $R(x,y) \land S(x,y)$ \\ \vspace{-4pt} \end{tabular} 
& \begin{tabular}{l} $R(x)$ \end{tabular}
& \begin{tabular}{l} $R(x,x)$ \\ $R(x,y)$ \end{tabular}\\
\hline
				{\bf Codd} &
				\begin{tabular}{l} $R(x) \land S(x)$ \end{tabular} 
				&
				\begin{tabular}{l} \vspace{-4pt} \\ $R(x) \land S(x,y) \land T(y)$ \\ ? \\ \vspace{-4pt} \\ \end{tabular} 
				& \begin{tabular}{l} $R(x)$ \end{tabular} 
				&
				\begin{tabular}{l}  $R(x,x)$ \\ $R(x,y)$ \end{tabular}\\
				\hline
			\end{tabular}
		\caption{Our dichotomies for counting valuations and completions of \sjfbcqs. For each one of the eight cases, if an \sjfbcq~$q$ has as a pattern a query mentioned in that case, then the problem is \shp-hard (and~\shp-complete for counting valuations, as well as for counting completions over Codd tables). In turn, for each case except for counting valuations under uniform Codd tables, if an \sjfbcq\ $q$ does not have as a pattern any of the queries mentioned in that case, then the problem is in \fp.}
	\label{tab:dichos-count-intro}
\end{table*}

The problems we study are analogous to the ones studied in other uncertainty scenarios; in particular, 
in {\em probabilistic} {\em databases}~\cite{suciu2011probabilistic,dalvi2012dichotomy} 
and {\em inconsistent databases}~\cite{ABC99,2011Bertossi,Bertossi19}. 
More specifically, we deal with the problems \countcompls$(q)$ and \countvals$(q)$ focusing on obtaining dichotomy results for them  
in terms of counting complexity classes, as well as 
studying the existence of randomized algorithms that approximate their results under probabilistic guarantees. 
For the dichotomies, we 
concentrate on self-join-free Boolean conjunctive queries (\sjfbcqs). This assumption simplifies the 
mathematical 
analysis, while at the same 
defines a setting which is rich enough for many of the theoretical concepts behind these problems to appear in full force. 
Notice that a similar assumption is used in several works that study counting problems over probabilistic and inconsistent 
databases; see, e.g.,~\cite{dalvi2011queries,maslowski2013dichotomy}. 

To refine our analysis, we study two restrictions of the problems \countcompls$(q)$ and \countvals$(q)$ based on 
natural modifications of the semantics, and analyze to what extent these restrictions simplify our problems.
For the first restriction we consider incomplete databases in which each null occurs exactly once, which corresponds to the 
well-studied setting of {\em Codd tables} -- 
as opposed to {\em naïve tables} where nulls are allowed to have multiple occurrences. We denote the corresponding problems by~$\ccountvals(q)$ and~$\ccountcompls(q)$.
For the second restriction, we consider {\em uniform} incomplete databases in which all the nulls share the same domain -- as opposed to the basic {\em non-uniform} setting in which all nulls come equipped with their own domain. 
We denote the corresponding problems by~$\ucountvals(q)$ and~$\ucountcompls(q)$.
When both restrictions are in place, we denote the problems by~$\cucountvals(q)$ and~$\cucountcompls(q)$.

\paragraph{{\bf Our dichotomies for exact counting}} We provide almost complete characterizations of the complexity of 
counting valuations and completions satisfying a given~\sjfbcq\ $q$,
when the input is 
a Codd table or a naïve table, and is a non-uniform or a uniform incomplete database (hence we have~eight cases in total). 
The only case we have not yet completely solved is that of counting valuations in the uniform setting 
over Codd tables, i.e., the problem~$\cucountvals(q)$.
Our seven dichotomies express that these problems are  either tractable 
or~\shp-hard, and that the tractable cases can be fully characterized by the absence of certain forbidden {\em patterns} 
in~$q$. A pattern is simply an \sjfbcq\ which can be obtained from~$q$  
essentially by deleting atoms, 
renaming relation symbols, deleting occurrences of variables and reordering the variables in atoms (the exact definition of this notion is given in Section \ref{sec:countvals-sjfcqs}).
Our characterizations are presented in Table~\ref{tab:dichos-count-intro}. 
By analyzing this table we can draw some important conclusions as 
explained next.

\medskip
\noindent \underline{\countcompls$(q)$ and \countvals$(q)$ are computationally difficult:} For very few \sjfbcqs~$q$  
the aforementioned problems can be solved in polynomial time. Take as an example 
the uniform setting over naïve tables. 
Then 
\ucountvals$(q)$ is~\shp-hard as long as~$q$ contains the pattern~$R(x,x)$, or~$R(x)\land S(x,y)\land T(y)$, 
or $R(x,y) \wedge S(x,y)$. That is, as long as there is an atom in~$q$ that contains a repeated variable~$x$,  
or a pair~$(x,y)$ of variables that appear in an atom and both~$x$ and~$y$ appear in some other atoms in~$q$.
By contrast, for this same setting, \ucountcompls$(q)$ is~\shp-hard as long as~$q$ contains the pattern~$R(x,y)$ or~$R(x,x)$, that is, as long as 
there is an atom in~$q$ that is not of arity one. 

\smallskip
\noindent \underline{\countvals$(q)$ is always easier than \countcompls$(q)$:} In all
		of the possible versions of our problem, 
 the tractable cases for \countvals$(q)$
		are a strict superset of the ones for \countcompls$(q)$.
		For instance, we have that $\cucountcompls(\exists x \exists y\, R(x,y))$ is hard, while $\cucountvals(\exists x \exists y\, R(x,y))$ is\linebreak tractable  (because $\ucountvals(\exists x \exists y\, R(x,y))$ is).

\smallskip
\noindent
\underline{Even counting completions is hard:}
While counting the total number of valuations for an 
incomplete database can always be done in polynomial time, observe from 
Table \ref{tab:dichos-count-intro}
that $\cucountcompls(\exists x \exists y\, R(x,y))$ is~\shp-hard,
and thus that simply counting the completions of a uniform Codd table with a single binary relation~$R$ 
is~\shp-hard. 
Moreover, we show that in the non-uniform case a single unary relation suffices to obtain~\shp-hardness. 

\smallskip
\noindent \underline{Codd tables help but not much:} We show that counting valuations
is easier 
for Codd tables than for
		naïve tables. In particular, there is always an \sjfbcq~$q$ 
such that counting  the valuations that satisfy~$q$
is~\shp-hard, yet it becomes tractable when restricted to the case of Codd tables.
		However, for counting completions, both in the uniform and non-uniform setting, 
the sole restriction to Codd tables presents no benefits: 
		for every \sjfbcq~$q$, we have that 
		$\countcompls(q)$ (resp., $\ucountcompls(q)$) is \shp-hard if and only if  $\ccountcompls(q)$ (resp., $\cucountcompls(q)$) is \shp-hard.

\smallskip
\noindent \underline{Non-uniformity complicates things:} All versions of our problems become harder in the non-uniform setting. 
This means that in all cases there is an \sjfbcq~$q$ for which countings valuations is tractable on uniform incomplete databases, but becomes~\shp-hard assuming non-uniformity, 
and an analogous result holds for counting completions. 

\paragraph{{\bf Our dichotomies for approximate counting}} 
Although\linebreak
\countvals$(q)$ can be~\shp-hard, we prove in the paper that
 good randomized approximate algorithms can be designed for this problem. 
 More precisely, we give a general condition under which~$\countvals(q)$ admits a  {\em fully polynomial-time randomized approximation scheme}~\cite{jerrum1986random} (FPRAS).
 This condition applies in particular to all \emph{unions of Boolean conjunctive queries}.
		Remarkably, we show that this no longer holds for \countcompls$(q)$; more precisely, 
there exists an \sjfbcq~$q$ such that \countcompls$(q)$ does not admit an FPRAS under a widely believed complexity theoretical assumption. More surprisingly, even counting the completions of uniform incomplete database containing a single
binary relation does not admit an FPRAS under such an assumption (and in the non-uniform case, a single unary relation suffices). 
Generally, for \sjfbcqs, we obtain seven dichotomies for our problems between polynomial-time computability of exact counting and non admissibility of an FPRAS. The only case
that we did not solve yet is that of~$\cucountcompls(q)$.

\paragraph{\bf Beyond \shp} 
It is easy to see that the problem of counting valuations is always in \shp. 
This is no longer the case for counting completions, 
and in fact we show that, under a complexity theoretical assumption, there is an \sjfbcq\ $q$ 
for which \ucountcompls$(q)$ is not in~\shp.
This does not hold if restricted to Codd 
tables, however, as we prove that \ccountcompls$(q)$ is always in \shp. 

For reasons that we explain in the paper, a suitable complexity class for the problem \countcompls$(q)$ is
\spanp, 
which is defined as the class of counting 
problems that can be expressed as the number of different accepting outputs 
of a nondeterministic Turing machine running in polynomial time. 
While we have not managed to prove that there is an \sjfbcq\ $q$ for which \countcompls$(q)$ is \spanp-complete, 
we show in the paper that this is the case for the problem of counting completions for the negation of an \sjfbcq, 
even in the uniform setting; that is, we show that
\ucountcompls$(\neg q)$ is \spanp-complete for some \sjfbcq\ $q$.

\paragraph{{\bf Organization of the paper}}

		We start with 
		the main terminology used in the paper in Section~\ref{sec:preliminaries}, and then present in Section~\ref{sec:countvals-sjfcqs} our three dichotomies on~\countvals$(q)$ when~$q$ is an \sjfbcq, and the input incomplete database can be Codd or not, and the domain can be uniform or not. We then establish the 
		four dichotomies on~\countcompls$(q)$ in Section~\ref{sec:countcompls-sjfcqs}. In Section~\ref{sec:approx}, we study the approximability complexity of our problems. We then give in Section~\ref{sec:misc} some general considerations 
		about the exact complexity of the problem~\countcompls$(q)$ going beyond \shp. 
		 In Section~\ref{sec:related}, we discuss related work and explain the differences with the problems considered in this paper. Last, we provide some conclusions and mention possible directions for future work in Section~\ref{sec:conclusion}.
All missing proofs can be found in the appendix.

\section{Preliminaries}
\label{sec:preliminaries}

\paragraph{{\bf Relational databases and conjunctive queries}}
A \emph{relational schema}~$\sigma$ is a finite non-empty set of relation
symbols written $R$, $S$,
$T$, \dots, each with its associated \emph{arity}, which is denoted by 
$\arity(R)$.
Let~$\consts$ be a countably infinite set of constants. A {\em database}~$D$ over~$\sigma$ is a set of {\em facts} 
of the form $R(a_1, \ldots, a_{\arity(R)})$ with~$R \in
\sigma$, and where each element~$a_i \in \consts$. For~$R\in \sigma$, we denote by~$D(R)$ the subset of~$D$ consisting of facts over~$R$. Such a set is usually called a \emph{relation} of~$D$.

A Boolean query~$q$ is a query that a database~$D$ can \emph{satisfy} (written~$D \models q$) or not
(written~$D \not\models q$).
A {\em Boolean conjunctive query} (BCQ) over~$\sigma$ is an FO formula of the form 
\begin{equation} \label{eq:cq}  
\exists \bar x \, \big(R_1(\bar x_1) \land   \, \dots \, \land R_m(\bar x_m)\big),
\end{equation} 
where all variables are existentially quantified, and where for each $i \in [1,m]$, we have that 
$R_i$ is a relation symbol in~$\sigma$ and~$\bar x_i$ is a tuple of variables with $|\bar x_i| = \arity(R_i)$.  
To avoid trivialities, we will always assume that~$m \geq 1$, i.e., the query has at least one atom, and also that~$\arity(R_i) \geq 1$ for all atoms.
For simplicity, we 
typically write a BCQ~$q$ of the form \eqref{eq:cq} as 
$$R_1(\bar x_1)  \land   \, \dots \, \land   R_m(\bar x_m),$$
and it will be implicitly understood that all variables in~$q$ are existentially quantified.  
As usual, we define the semantics of a BCQ in terms of {\em homomorphisms}.
A homomorphism from~$q$ to a database~$D$ is a mapping from the variables in~$q$ to the constants 
used in~$D$ such that
$\{R_1(h(\bar x_1)),\dots,R_m(h(\bar x_m))\} \subseteq D$.
Then, we have~$D \models q$
if there is a homomorphism from~$q$ to~$D$.   
A \emph{self-join--free BCQ} (\sjfbcq) is a BCQ such that no two atoms use the same relation symbol.

\begin{figure*}[t]
	\begin{tabular}{c|cccccc}
		\toprule
		$(\nu(\bot_1),\nu(\bot_2))$ & $(a,a)$ & $(a,b)$ & $(b,a)$ & $(b,b)$ & $(c,a)$ & $(c,b)$ \\
		\hline
		$\nu(D)$

		&
		\begin{tabular}[t]{cc}
			\toprule
			\multicolumn{2}{c}{$S$} \\
			\midrule
			$a$ & $b$ \\
			$a$ & $a$ \\
			\bottomrule
		\end{tabular}

		& 
		\begin{tabular}[t]{cc}
			\toprule
			\multicolumn{2}{c}{$S$} \\
			\midrule
			$a$ & $b$ \\
			$a$ & $a$ \\
			\bottomrule
		\end{tabular}

		& 
		\begin{tabular}[t]{cc}
			\toprule
			\multicolumn{2}{c}{$S$} \\
			\midrule
			$a$ & $b$ \\
			$b$ & $a$ \\
			$a$ & $a$ \\
			\bottomrule
		\end{tabular}

		& 
		\begin{tabular}[t]{cc}
			\toprule
			\multicolumn{2}{c}{$S$} \\
			\midrule
			$a$ & $b$ \\
			$b$ & $a$ \\
			\bottomrule
		\end{tabular}

		& 
		\begin{tabular}[t]{cc}
			\toprule
			\multicolumn{2}{c}{$S$} \\
			\midrule
			$a$ & $b$ \\
			$c$ & $a$ \\
			$a$ & $a$ \\
			\bottomrule
		\end{tabular}

		&
		\begin{tabular}[t]{cc}
			\toprule
			\multicolumn{2}{c}{$S$} \\
			\midrule
			$a$ & $b$ \\
			$c$ & $a$ \\
			\bottomrule
		\end{tabular} \\[2.2cm]
		\hline
		$\nu(D) \models q$? & Yes & Yes & Yes & No & Yes & No \\
		\bottomrule

	\end{tabular}
\caption{The~six valuations of the (non-uniform) incomplete database~$D = (T,\dom)$ with $T=\{S(a,b),S(\bot_1,a),S(a,\bot_2)\}$ from Example~\ref{expl:problems}, and their corresponding completions.
	The Boolean conjunctive query~$q$ is $\exists x \, S(x,x)$.}
\label{fig:example-completions}
\end{figure*}

\paragraph{{\bf Incomplete databases}}
Let~$\nulls$ be a countably infinite set of nulls (also called \emph{labeled nulls} in the literature), which is 
disjoint with~$\consts$. An {\em incomplete database} over schema~$\sigma$ is a 
pair~$D = (T,\dom)$, where~$T$ is a database over~$\sigma$ 
whose facts contain elements in~$\consts \cup \nulls$, and where $\dom$ is a function that associates to every null~$\bot$ occurring in~$D$ a subset~$\dom(\bot)$ of~$\consts$.
Intuitively,~$T$ is a database that can mention both constants 
and nulls, while~$\dom$ tells us where nulls are to be interpreted.
Following the literature, we call~$T$ a {\em na\"ive table} \cite{imielinski1984incomplete}.

An incomplete database~$D = (T,\dom)$ can represent potentially many complete databases, via what are called \emph{valuations}.
A valuation of~$D$ is simply a function~$\nu$ that maps each null~$\bot$ occurring in~$T$ to a constant~$\nu(\bot) \in \dom(\bot)$.
Such a valuation naturally defines a \emph{completion of~$D$}, denoted by~$\nu(T)$, which is the complete database obtained from~$T$ by substituting each null~$\bot$ appearing in~$T$ by~$\nu(\bot)$.
It is understood, since a database is a \emph{set} of facts, that~$\nu(T)$ does not contain duplicate facts.
By paying attention to completions of incomplete databases that are generated exclusively by applying valuations to them, we are sticking to the so called {\em closed-world} semantics of 
incompleteness~\cite{abiteboul1995foundations,reiter1978closed}. This means that the databases represented by an incomplete database~$D = (T,\dom)$ are not open to adding facts that are not ``justified'' 
by the facts in~$T$. 

\begin{example}
\label{expl:completion}
	{\em Let~$D = (T,\dom)$ be the incomplete database
	consisting of the na\"ive table $T = \{S(\bot_1,\bot_1),S(a,\bot_2)\}$, and where $\dom(\bot_1) = \{a,b\}$ and~$\dom(\bot_2) = \{a,c\}$.
	Let~$\nu_1$ be the valuation mapping~$\bot_1$ to~$b$ and~$\bot_2$ to~$c$.
	Then~$\nu_1(T)$ is $\{S(b,b),S(a,c)\}$.
	Let~$\nu_2$ be the valuation mapping both~$\bot_1$ and~$\bot_2$ to~$a$.
	Then~$\nu_2(T)$ is~$\{S(a,a)\}$. On the other hand, the function~$\nu$ mapping~$\bot_1$ and~$\bot_2$ to~$b$ is not a valuation of~$D$, because~$b\notin \dom(\bot_2)$.} \qed 
\end{example}

When every null occurs at most once in~$T$, then~$D$ is what is called a \emph{Codd table}~\cite{codd1975understanding}; for instance,
the incomplete database in Example~\ref{expl:completion} is not a Codd table because~$\bot_1$ occurs twice. 
We also consider {\em uniform} incomplete databases in which the domain of every null is the same. 
Formally, a uniform incomplete database is a pair~$D = (T,\dom)$, where~$T$ is a database over~$\sigma$ and~$\dom$ is a subset of~$\consts$. The difference now is  
that a valuation~$\nu$ of~$D$ must simply satisfy~$\nu(\bot) \in \dom$ for every null of~$D$.

We will often abuse notation and use~$D$ instead of~$T$;
for instance, we write~$\nu(D)$ instead of~$\nu(T)$, or~$R(a,a) \in D$ instead of~$R(a,a) \in T$, or again~$D(R)$ instead of~$T(R)$.

\paragraph*{\bf Counting problems on incomplete databases.}
We will study two kinds of counting problems for incomplete databases: problems of the form $\countvals(q)$, that count the number of \emph{valuations}~$\nu$
that yield a completion~$\nu(D)$ satisfying a given BCQ~$q$, and problems of the form~$\countcompls(q)$, that count the number of \emph{completions} that satisfy~$q$.
The query~$q$ is assumed to be fixed, so that each query gives rise to different counting problems,
and we are considering the \emph{data complexity}~\cite{vardi1982complexity} of these problems.

Before formally introducing our problems, let us observe that they are well defined if we
assume that the set of constants to which a null can be mapped to is finite.
Hence, for the (default) case of an 
incomplete database~$D = (T, \dom)$, we assume that~$\dom(\bot)$ is always a finite subset of~$\consts$.
Similarly, for the case of a uniform incomplete database~$D = (T, \dom)$, we assume that~$\dom$ is a finite subset of~$\consts$.
Finally, given a Boolean query~$q$, we use notation $\sig(q)$ for the set of relation symbols occurring in~$q$.
With these ingredients, we can define our problems for the (default) case of 
incomplete naïve tables and a Boolean query~$q$. 

\medskip
\begin{center}
\fbox{\begin{tabular}{lp{6.1cm}}
\small{PROBLEM} : & $\countvals(q)$ 
\\
{\small INPUT} : & An incomplete database $D$ over $\sig(q)$ 
\\
{\small OUTPUT} : & Number of valuations $\nu$ of $D$ with $\nu(D) \models q$ 
\end{tabular}}

\medskip
\fbox{\begin{tabular}{lp{6.1cm}}
\small{PROBLEM} : & $\countcompls(q)$
\\
{\small INPUT} : & An incomplete database $D$ over $\sig(q)$ 
\\
{\small OUTPUT} : & Number of completions $\nu(D)$ of $D$ with $\nu(D) \models q$
\end{tabular}}
\end{center}
\medskip

We also consider the uniform variants of these problems, in which the input~$D$ is a uniform incomplete database over $\sig(q)$, and the restriction of these problems where the input is a Codd table instead of a naïve table.
We then use the terms $\ucountvals(q)$, $\ucountcompls(q)$ when restricted to the uniform case, $\ccountvals(q)$, $\ccountcompls(q)$ when restricted to Codd tables, and $\cucountvals(q)$, $\cucountcompls(q)$ when both restrictions are applied.

As we will see, even though the problems 
\countvals$(q)$ and \countcompls$(q)$ 
look similar, they are of a different computational nature;
this is because two distinct valuations can produce the same completion of an incomplete database.
In the following example, we illustrate this phenomenon.

\begin{example}
\label{expl:problems}
	{\em Let~$q$ be the Boolean conjunctive query~$\exists x \, S(x,x)$, and~$D$ be the (non-uniform) incomplete database~$D = (T,\dom)$, 
	with $T \, = \, \{S(a,b)$, $S(\bot_1,a),S(a,\bot_2)\}$,  
	$\dom(\bot_1)=\{a,b,c\}$ and $\dom(\bot_2)=\{a,b\}$.
	We have depicted in Figure~\ref{fig:example-completions} the~six valuations of~$D$ together with the completions that they define.
	Out of these six valuations~$\nu$, only four are such that~$\nu(D) \models q$, so that
	we have	$\countvals(q)(D) = 4$.
	Moreover, there are only~$3$ distinct completions of~$D$ that satisfy~$q$, so $\countcompls(q)(D) = 3$.} \qed
\end{example}

\paragraph*{\bf Counting complexity classes.}

Given two problems~$A,B$, we write 
$A \leq_{\mathrm{T}}^{\mathrm{p}} B$ when~$A$ reduces to~$B$ under polynomial-time Turing reductions.
When both~$A$ and~$B$ are counting problems, we write~$A \pr B$ when~$A$ can be reduced to~$B$ under polynomial-time \emph{parsimonious} reductions, i.e., there exists a polynomial-time computable function~$f$ that transforms an input~$x$ of~$A$ to an input~$f(x)$ of~$B$ such that~$A(x)=B(f(x))$.
We say that a counting problem is in~\fp\ when it can be solved in polynomial time.
We will consider the counting complexity class \shp~\cite{valiant1979complexity} of problems that can be expressed as the number of accepting paths of a 
nondeterministic Turing machine running in polynomial time.
Following~\cite{valiant1979complexity,V79}, we define \#P-hardness using Turing reductions. 
It is clear that $\fp \subseteq \shp$. This inclusion, on the other hand, is widely believed to be strict. Therefore, proving that 
a counting problem is \shp-hard implies that it cannot be solved in polynomial time under such an assumption. 

\paragraph*{\bf Graphs.}

In our reductions, we will often depart from hard problems that are defined over graphs. Unless mentioned otherwise, by~\emph{graph} we mean a pair~$G=(V,E)$, where~$V$ is a finite set of \emph{nodes}, and~$E$ is a set whose elements are of the form~$\{u,v\}$ for~$u,v\in V$ and~$u \neq v$. 
Notice then that such graphs are \emph{undirected}, cannot contain~\emph{self-loops}, and cannot contain multiple edges between any two nodes.

\section{Dichotomies for counting valuations}
\label{sec:countvals-sjfcqs}
\begin{toappendix}
\label{apx:countvals-sjfcqs}
\end{toappendix}

In this section, given a fixed \sjfbcq\ q, we study the 
complexity of the problem of computing, given an incomplete database $D$, the number of valuations $\nu$ of $D$ such that $\nu(D)$ satisfies $q$.
Recall that we have four cases to consider for this problem depending on whether we focus on na\"ive or on Codd tables, where nulls are restricted to appear at most once, and whether we focus on non-uniform or uniform incomplete databases, where nulls are restricted to have the same domain. 
Our specific goal then is to understand whether the problem is tractable (in \fp) or \shp-hard in these scenarios, depending on the shape of $q$.

To this end, the shape of an \sjfbcq\ $q$ will be characterized by the presence or absence of certain specific patterns.
In the following definition, we introduce the necessary terminology to formally talk about the presence of a pattern in a query. 

\begin{definition}
\label{def:pattern}
	Let~$q,q'$ be 
	\sjfbcqs. 
		We say that~$q'$ is a \emph{pattern} of~$q$ if $q'$ can be obtained from~$q$ by using an arbitrary number of times and in any order the following operations: deleting an atom,  deleting an occurrence of a variable, renaming a relation to a fresh one, renaming a variable to a fresh one and reordering the variables in an atom.\footnote{We remind the reader that we assume all sjfBCQs to contain at least one atom and that all atoms must contain at least one variable.}
\end{definition}

\begin{example}
\label{expl:pattern}
 {\em Recall that we always omit existential quantifiers in Boolean queries. Then we have that 
	$q' = R'(u,u,y) \land  S'(z)$ is a pattern of~$ q = R(u,x,u) \land  S'(y,y) \land  T(x,s,z,s)$. Indeed, $q'$ can be obtained from $q$ by deleting atom $T(x,s,z,s)$, renaming $R(u,x,u)$ as $R'(u,x,u)$ to obtain~$R'(u,x,u) \land S'(y,y)$, reordering the variables in~$R'(u,x,u)$ to obtain~$R'(u,u,x) \land S'(y,y)$, renaming variable~$y$ into~$z$ to obtain~$R'(u,u,x) \land S'(z,z)$, deleting the second variable occurrence in~$S'(z,z)$ to obtain~$R'(u,u,x) \land S'(z)$, and finally renaming variable~$x$ into~$y$ to obtain~$q'$.} \qed
\end{example}

In the following general lemma, we show that if $q'$ is a pattern of $q$, then each one of the problems considered in this section is as hard for $q$ as it is for $q'$. Recall in this result that unless stated otherwise, our problems are defined for na\"ive tables.

\begin{toappendix}
\subsection{Proof of Lemma~\ref{lem:pattern-parsimonious}}
\end{toappendix}

\begin{lemmarep}
\label{lem:pattern-parsimonious}
	Let~$q,q'$ be \sjfbcqs\ such that~$q'$ is a pattern of~$q$.
	Then we have~$\countvals(q') \pr \countvals(q)$ and $\ucountvals(q') \pr \ucountvals(q)$. Moreover, the same results hold if we restrict to Codd tables.
\end{lemmarep}
\begin{proof}
	We only present the proof of $\countvals(q') \pr \countvals(q)$, as the uniform case is identical.
	Let~$q$ be $R_1(\overline{x}_1) \land  \ldots \land  R_m(\overline{x}_m)$, and $q'$ be $R'_{j_1}(\overline{x}'_{j_1}) \land  \ldots \land  R'_{j_p}(\overline{x}'_{j_p})$,
	where $1 \leq j_1 < \ldots < j_p \leq m$ and each $R_{j_k}$ has become $R'_{j_k}$ (i.e., we either have~$R_{j_k} = R'_{j_k}$, or~$R_{j_k}$ has been renamed into~$R'_{j_k}$) and $\overline{x}'_{j_k}$ is obtained from 
	$\overline{x}_{j_k}$ by deleting some variable occurrences but not all, and the other atoms have been deleted (we will assume without loss of generality that
	we did not reorder the variables in the atoms nor renamed variables by fresh ones, because this obviously does not change the complexity of the problems).
	Let~$D'$ be an incomplete database input of~$\countvals(q')$.
	Let~$A$ be set of constants that are appearing in~$D'$ or are in a domain of some null occurring in~$D'$.
	For~$1 \leq k \leq p$, we construct the relation~$D(R_{j_k})$ from the relation~$D'(R'_{j_k})$.
	Let us assume that~$\overline{x}_{j_k}$ is the tuple
	$(x_1, \ldots, x_r)$ (with some variables possibly being equal).
	We initialize~$D(R_{j_k})$ to be empty, and then for every tuple~$\overline{t}'$ in $D'(R'_{j_k})$
	we add to~$D(R_{j_k})$ all the tuples~$\overline{t}$ that can be obtained from~$\overline{t}'$ in the following way for $1 \leq i \leq r$:
	\begin{itemize}
		\item[\bf a)] If $x_i$ is a variable occurrence that has not been deleted from~$\overline{x}_{j_k}$,
			then copy the element (constant or null) of $\overline{t}'$ corresponding to that variable occurrence to the $i$-th position of~$\overline{t}$;
		\item[\bf b)] Otherwise, if $x_i$ is a variable occurrence that has been deleted from~$\overline{x}_{j_k}$,
				then fill the~$i$-th position of~$\overline{t}$ 
				with every possible constant from~$A$.
		\end{itemize}
		Then we construct the relations~$D(R_i)$ where~$R'_i$ does not appear in~$q'$ (this can happen if we have deleted the atom $R_i(\overline{x}_i)$) by filling it with every possible~$R_i$-fact over~$A$.
		We leave the domains of all nulls unchanged.
		The whole construction can be performed in polynomial time (this uses the fact that $q$ is assumed to be fixed, so that the arities of the relations mentioned in $q$ are fixed).
				Since~$D$ and~$D'$ contain exactly the same set of nulls,
		the construction preserves the property of being a Codd table.
		Hence, it only remains to be checked that~$\countvals(q')(D') = \countvals(q)(D)$, that is, that the reduction works and is indeed parsimonious.
		It is clear that the valuations of~$D'$ are exactly the same as the valuations of~$D$ (because they have the same sets of nulls).
		Hence it is enough to verify that for every valuation~$\nu$, we have~$\nu(D') \models q'$ if and only if~$\nu(D) \models q$.
		Let~$h'$ be a homomorphism from~$q'$ to $\nu(D')$ witnessing that~$\nu(D') \models q'$ (i.e., we have $h'(q) \subseteq \nu(D')$).
		Then~$h'$ can clearly be extended in the expected way into a homomorphism~$h$ from~$q$ to~$\nu(D)$: this is in particular thanks to the fact that
		we filled the missing columns with every possible constant.
		Conversely, let $h$ be a homomorphism from~$q$ to $\nu(D)$ witnessing that~$\nu(D) \models q$.
		Then the restriction~$h'$ of~$h$ to the variables occurring in~$q'$ is such that~$h(q') \subseteq \nu(D')$, hence we have~$\nu(D') \models q'$.
		This concludes the proof.
\end{proof}

	The idea is then to show the \shp-hardness of our problems for some simple patterns, which then we combine with Lemma~\ref{lem:pattern-parsimonious} and with some tractability proofs to
	obtain the desired dichotomies. 
	Our findings are summarized in the first two columns of 
	Table \ref{tab:dichos-count-intro} in the introduction. 
	We first focus on the two dichotomies for the non-uniform setting in Section~\ref{subsec:countvals-non-uniform}, and then we move to the case of uniform incomplete databases 
	in Section~\ref{subsec:countvals-uniform}. We explicitly state when 
	a \shp-hardness result holds even in the restricted setting in which there is a fixed 
	domain over which nulls are interpreted. In other words, when there is a fixed domain $A$ such that the incomplete databases 
	used in the reductions are of the form $D = (T,\dom)$ and $\dom(\bot) \subseteq A$, for each null $\bot$ of $T$.

	\subsection{The complexity on the non-uniform case}
	\label{subsec:countvals-non-uniform}

	In this section, we study the complexity of the problems $\countvals(q)$ and~$\ccountvals(q)$, providing dichotomy results in both cases. We start by proving the \shp-hardness results needed for these dichotomies. We first show that $\countvals(R(x,x))$ is \shp-hard by actually proving that hardness holds already in the uniform case.

	\begin{proposition}
	\label{prp:Rxx-hard}
	\begin{sloppypar}
		$\ucountvals(R(x,x))$ is \shp-hard and, hence, $\countvals(R(x,x))$ is also \shp-hard. This holds even in the restricted setting 
		in which all nulls are interpreted over the same fixed domain~$\{1,2,3\}$.
	\end{sloppypar}
	\end{proposition}
	\begin{proof}
		We reduce from the problem of counting the number of $3$-colorings of a graph~$G=(V,E)$, which is \shp-hard~\cite{jaeger1990computational}.
		For every node~$v\in V$ we have a null~$\bot_v$, and for every edge~$\{u,v\} \in E$ we have the facts~$R(\bot_v,\bot_u)$ and $R(\bot_u,\bot_v)$.
		The domain  of the nulls is~$\{1,2,3\}$.
		It is then clear that the number of valuations of the constructed database that do not satisfy~$R(x,x)$ is exactly the number of
		$3$-colorings of~$G$. Since the total number of valuations can be computed in PTIME, this concludes the reduction.
	\end{proof}
	The next pattern that we consider is~$R(x) \land S(x)$.	This time, we can show \shp-hardness of the problem even for Codd databases.

	\begin{toappendix}
		\subsection{Proof of Proposition~\ref{prp:RxSx-hard}}
		In this section we prove the following.
	\end{toappendix}

	\begin{propositionrep}
	\label{prp:RxSx-hard}
		$\ccountvals(R(x) \land S(x))$ is \shp-hard.
	\end{propositionrep}

	\begin{toappendix}
In the main text of this article, all graphs considered were undirected graphs with no self-loops an did not contain multiple edges between any two nodes.
By contrast, the proofs in this section will rely on hardness results that use a more general notion of graph.
We introduce here the notation that we will use in these proofs.

\paragraph*{\bf Multigraphs.}

By \emph{multigraph}, we mean a finite undirected graph without self-loops, where two nodes can have multiple edges between them.
Formally, a multigraph $G=(V,E,\lambda)$ consists of a finite set~$V$ of \emph{nodes}, a finite set~$E$ of \emph{edges}, and a function~$\lambda$
that assigns to every edge~$e\in E$ a set~$\lambda(e)=\{u,v\}$ of two distinct nodes~$u,v\in V$.
We say that $e$ is \emph{incident} to~$u$, and to~$v$, and that~$u$ and $v$ are \emph{adjacent} (or are \emph{neighbors}).
Two edges~$e\neq e'$ are \emph{parallel} when~$\lambda(e)=\lambda(e')$.
For a node~$u\in V$, we write~$E(u)$ the set of edges that are incident to~$u$, and the \emph{degree} of~$u$ is~$|E(u)|$.
We say that~$G$ is~\emph{$d$-regular} (where~$d \in \N_{\geq 1}$) when every node of~$G$ has degree~$d$.
A multigraph~$G$ is \emph{bipartite} when its nodes can be partitioned in two sets~$A,B$ such that for every edge~$e$ of~$G$ we have that
$\lambda(e)$ contains one node in~$A$ and one node in~$B$.
We will often write such a bipartite multigraph as~$G=(A \sqcup B,E, \lambda)$.
A bipartite multigraph~$G$ is~\emph{$a$-$b$--regular} when every node in~$A$ has degree~$a$ and every node in~$B$ has degree~$b$.
Observe that with these definitions, we can see a graph, as defined in Section~\ref{sec:preliminaries}, as a multigraph that does not contain parallel edges.
We can then indeed write a graph simply as~$G=(V,E)$, and an edge~$e\in E$ as~$e=\{u,v\}$ (with~$u\neq v$).\\

		We now explain how we prove Proposition~\ref{prp:RxSx-hard}. We execute the reduction in two steps, and we start from the problem that we call~\emph{$\avoidance$}.

	\begin{definition}
		\label{def:avoidance}
		Let~$G=(V,E,\lambda)$ be a multigraph. 
		An \emph{assignment} of~$G$ is a mapping~$\mu: V \to E$ such that for every node~$u$ we have~$\mu(u) \in E(u)$.
		We say that~$\mu$ is \emph{avoiding} when there does not exist two (adjacent) nodes~$u,v$ such that~$\mu(u)=\mu(v)$.
		The problem $\avoidance$ takes as input a multigraph~$G$ and outputs the number of avoiding assignments of~$G$.
	\end{definition}

	\begin{example}
		Figure~\ref{fig:avoiding} represents a multigraph together with an avoiding assignment.
	\end{example}

	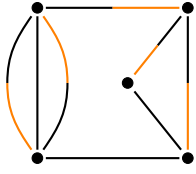
\begin{figure}
		\centering
		\begin{tikzpicture}
			\tikzset{nodestyle/.style={circle,fill=black,radius=2pt}}

\tikzset{draw half paths/.style 2 args={%
  decoration={show path construction,
    lineto code={
      \draw [#1] (\tikzinputsegmentfirst) -- 
         ($(\tikzinputsegmentfirst)!0.5!(\tikzinputsegmentlast)$);
      \draw [#2] ($(\tikzinputsegmentfirst)!0.5!(\tikzinputsegmentlast)$)
        -- (\tikzinputsegmentlast);
    }
  }, decorate
}}

\node (a) at (0,0) {};
\node (b) at (0,2) {};
\node (c) at (2,2) {};
\node (d) at (1.2,1) {};
\node (e) at (2,0) {};

\draw[nodestyle] (a) circle; 
\draw[nodestyle] (b) circle; 
\draw[nodestyle] (c) circle; 
\draw[nodestyle] (d) circle; 
\draw[nodestyle] (e) circle; 

\draw[orange,thick] [out=125,in=270] (a) to (-0.4,1);
\draw[black,thick] [out=90,in=235] (-0.4,1) to (b);
\draw[black,thick] (a) to [out=55,in=270] (0.4,1);
\draw[orange,thick] (0.4,1) to [out=90,in=305] (b);
\draw[draw half paths={black, thick}{orange, thick}] (b) to (c);
\draw[draw half paths={black, thick}{orange, thick}] (c) to (d);
\draw[black,thick] (d) to (e);
\draw[black,thick] (a) to (e);
\draw[black,thick] (a) to (b);
\draw[draw half paths={black, thick}{orange, thick}] (c) to (e);

		\end{tikzpicture}
		\caption{A multigraph with an assignment~$\mu$. The edge~$\mu(v)$ assigned to each node is indicated in orange next to each node. This assignment is avoiding because no edge is fully orange.}
		\label{fig:avoiding}
	\end{figure}

	We explain how to obtain the following:

	\begin{proposition}[{Implied by~\cite{cai2012holographic}, see also~\cite{cstheory_a_variant}}]
	\label{prp:avoidance-hard-3-regular}
		The problem $\avoidance$ is \shp-complete, and hardness holds even when restricted to $3$-regular multigraphs.
	\end{proposition}

		Membership in \shp\ is clear.
		We will show how hardness derives from the
		results of~\cite{cai2012holographic}.
		First, let us introduce what are called \emph{Holant problems}:

	\begin{definition}
	\label{def:hamming}
		Let~$G=(V,E,\lambda)$ be a multigraph.
		For a valuation of the edges $\nu: E \to \{0,1\}$ and a node $t\in V$, we write $\nu(E(t))$ the multiset $\{\!\!\{ \nu(e) \mid e \in E(t) \}\!\!\}$.
		Given a multiset of bits $B$, the \emph{Hamming weight} of~$B$ is the number of $1$ bits in~$B$.
		For each $x_0, \ldots, x_n \in \{0, 1\}$, let $[x_0, \ldots, x_n]$ denote
		the function that takes a multiset of~$n$ bits as input and outputs $x_i$ if the Hamming weight
		of those $n$ bits is $i$.
	\end{definition}

	\begin{definition}
	\label{def:holant}
		For every $x_0, x_1, x_2, y_0, y_1, y_2, y_3 \in \{0,1\}$, the problem
		$\mathrm{Holant}([x_0, x_1, x_2] | [y_0, y_1, y_2, y_3])$ is the following:
		given a~\mbox{$2$-$3$}--regular bipartite multigraph $G = (U \sqcup V, E, \lambda)$, 
		compute the quantity 
		\[
		  \sum_{\nu: E \to \{0,1\}} \prod_{u \in U} [x_0, x_1, x_2] (\nu(E(u))
		\times\prod_{v \in V} [y_0, y_1, y_2, y_3] ( \nu(E(v))  ).
		\]
	\end{definition}

		Holant problems provide a rich framework to express a lot of natural problems on $2$-$3$--regular graphs
		(this was actually the motivation for introducing this framework), as the following examples illustrate:

		\begin{example}
			On $2$-$3$--regular multigraphs, observe that:
			\begin{itemize}
				\item Counting perfect matchings is exactly the problem~$\mathrm{Holant}([0,1,0]|[0,1,0,0])$;
				\item Counting matchings is exactly the problem~$\mathrm{Holant}([1,1,0]|[1,1,0,0])$;
				\item Counting edge covers is exactly the problem~$\mathrm{Holant}([0,1,1]|[0,1,1,1])$.
			\end{itemize}
		\end{example}

		We will then reduce $\avoidance$ to the problem~$\mathrm{Holant}([1,1,0]|[0,1,0,0])$, which is shown to be \shp-hard in~\cite{cai2012holographic} using the
		tools of \emph{holographic reduction} and \emph{interpolation}:

		\begin{proposition}[{\cite{cai2012holographic}}]
			\label{prp:holant-hard}
			$\mathrm{Holant}([1,1,0]|[0,1,0,0])$ is \shp-hard, even when restricted to $2$-$3$--regular graphs.
		\end{proposition}
		\begin{proof}
		The problem appears as hard for multigraphs in the table on page 10 of~\cite{cai2012holographic}.
			A careful inspection of the paper reveals that it is hard for graphs.
		\end{proof}

		Given a \mbox{$2$-$3$}--regular bipartite graph~$G$, we define the \emph{merging of~$G$} to be the multigraph obtained from~$G$
		by merging the incident edges of every node of degree~$2$.
		Note that, because~$G$ is a graph and not a multigraph, the merging of~$G$ is indeed a multigraph, i.e., it does not contain self-loops.
		Furthermore, it is easy to see that this multigraph is~$3$-regular.
		We can now show Proposition~\ref{prp:avoidance-hard-3-regular}:

		\begin{proof}[Proof of Proposition~\ref{prp:avoidance-hard-3-regular}]
			Let~$G$ be a \mbox{$2$-$3$}--regular bipartite graph~$G$ input of
			$\mathrm{Holant}([1,1,0]|[0,1,0,0])$.
			Construct in polynomial time from~$G$ its merging~$G'$, which is a~$3$-regular multigraph.
			Observe that the assignments of~$G'$ are in bijection with the edge subsets~$S$ of~$G$ such that every node of degree~$3$ in~$G$ is adjacent to exactly one edge in~$S$.
			One can then easily see that the number of avoiding assignments of~$G'$ corresponds to the value of $\mathrm{Holant}([1,1,0]|[0,1,0,0])$ on~$G$.
		\end{proof}

	However, in order to show hardness of~$\ccountvals(R(x)\land  S(x))$, we will actually need the hardness of \#Avoidance on \emph{bipartite} graphs.
	To the best of our knowledge this does not follow from related work, so we need to prove it here:

	\begin{proposition}
		\label{prp:avoidance-hard-bipartite}
		The problem~$\avoidance$ is \shp -hard when restricted to $2$-$3$--regular bipartite graphs.	
	\end{proposition}
	\begin{proof}
		We reduce from \#Avoidance on~$3$-regular multigraphs, which is hard according to Proposition~\ref{prp:avoidance-hard-3-regular}.
		Let~$G=(V,E,\lambda)$ be a $3$-regular multigraph.
		Let~$G'$ be the graph obtained from~$G$ by adding a node in the middle of every edge of~$G$.
		Formally, the vertices of~$G'$ are $V \sqcup \{n_e \mid e \in E\}$ and its edges are
		$\bigcup_{e \in E, \lambda(e)=\{u,v\}} \{\{u,n_e\},\{n_e,v\}\}$.
		It is clear that~$G'$ is a $2$-$3$--regular bipartite graph.
		We claim that \#Avoidance$(G') = 2^{|E|-|V|} \times$\#Avoidance$(G)$, which would complete the reduction.
		To prove this, we will use the following definition: letting~$\mu$ be an assignment of~$G$ and~$\mu'$ be an assignment of~$G'$,
		we say that~\emph{$\mu$ and $\mu'$ agree} if for every~$v \in V$, if~$\mu'(v) = \{v,n_e\}$ then we have~$\nu(v) = e$.
		We then show the following, which directly implies our claim:
		\begin{enumerate}
			\item For every avoiding assignment~$\mu$ of~$G$, there are exactly $2^{|E|-|V|}$ avoiding assignments~$\mu'$ of~$G'$ that agree with~$\mu$;
			\item If~$\mu'$ is an avoiding assignment of~$G'$, then~$\mu'_{|V}$ is an avoiding assignment of~$G$.
		\end{enumerate}
			We first prove item 1).
			We say that an edge~$e$ of~$G$ is \emph{chosen} if it is in the image of~$\mu$.
			Observe that, because~$\mu$ is avoiding, there are exactly~$|V|$ edges of~$G$ that are chosen; for instance, considering the graph of Figure~\ref{fig:avoiding}, the assignment is avoiding and there are~$5=|V|$ chosen edges. 
			Let us now look at the number of avoiding assignments~$\mu'$ of~$G'$ that agree with~$\mu$.
			It is easy to see that for every edge~$e$ of~$G$ that is chosen,
			the value of~$\mu'(n_e)$ is forced: we have to set~$\mu'(n_e)$ to be the (unique) edge~$\{n_e,v\}$ such that~$\mu(v) \neq e$.
			Moreover when~$e$ is not chosen, both values for~$n_e$ are possible.
			But then this indeed implies that there
			are~$2^{|E|-|V|}$ avoiding assignments~$\mu'$ of~$G'$ that agree with~$\mu$.
			To show item 2), assume by contradiction that~$\mu'_{|V}$ is not avoiding.
			This means that there is an edge~$e\in E$ with~$\lambda(e)=\{u,v\}$ such that we have~$\mu'_{|V}(u) = \mu'_{|V}(u) = e$.
			But then, looking at the possible value for~$\mu'(n_e)$, we see that~$\mu'$ cannot be avoiding in~$G'$, a contradiction.
	\end{proof}

	We are now ready to treat the problem~$\ccountvals(R(x)\land  S(x))$:

	\begin{proof}[Proof of Proposition~\ref{prp:RxSx-hard}]
		We reduce from \#Avoidance on bipartite graphs, which we have just shown to be \shp-hard.
		Let~$G=(U\sqcup V,E)$ be a bipartite graph.
		For every node~$t\in U \sqcup V$, we have a null~$\bot_t$ with domain~$\dom(\bot_t) \defeq E(t)$.
		For every node~$u\in U$ we have a fact~$R(\bot_u)$ and for every node~$v \in V$ a fact~$S(\bot_v)$.
		The constructed database is a Codd table.
		Moreover, it is clear that the value of $\ccountvals(R(x)\land  S(x))$ on that database is exactly the number of assignments of~$G$ that are not avoiding, thus establishing hardness.
	\end{proof}
\end{toappendix}

	Already with Propositions~\ref{prp:RxSx-hard} and~\ref{prp:Rxx-hard}, we have all the relevant hard patterns for the non-uniform setting.  
	We start by proving our dichotomy result for naïve tables, which is our default case.

	\begin{theorem}[dichotomy]
	\label{thm:countvals-sjfcqs}
		Let~$q$ be an \sjfbcq.
		If $R(x,x)$ or $R(x) \land S(x)$ is a pattern of $q$, then $\countvals(q)$ is \shp-complete. Otherwise, $\countvals(q)$ is in \fp.
	\end{theorem}
	\begin{proof}
		The \shp-hardness part of the claim follows from the last two propositions and from Lemma~\ref{lem:pattern-parsimonious}.
		We explain why the problems are in \shp\ right after this proof.
		When~$q$ does not have any of these two patterns then all variables have exactly one occurrence in~$q$. This implies that every valuation~$\nu$ of~$D$ is such that~$\nu(D)$ satisfies~$q$ (except when one relation is empty, in which case the result is simply zero). We can obviously compute the total number of valuations in FP by simply multiplying the sizes of the domains of every null in~$D$.
	\end{proof}
Notice that in this theorem, the membership of $\countvals(q)$ in \shp\ can be established by considering a nondeterministic Turing Machine~$M$ that, with input a non-uniform incomplete database~$D$, guesses a valuation~$\nu$ of~$D$ and verifies whether~$\nu(D)$ satisfies~$q$. This machine works in polynomial time as we can verify whether~$\nu(D)$ satisfies~$q$ in polynomial time (since~$q$ is a fixed FO query). Then given that $\countvals(q)(D)$ is equal to the number of accepting runs of~$M$ with input~$D$, we conclude that $\countvals(q)$ is in \shp. Obviously, the same idea shows that $\ccountvals(q)$ is in \shp. But with this restriction we obtain more tratable cases, as shown by the following dichotomy result.

	\begin{theorem}[dichotomy]
		\label{thm:countvals-sjfcqs-codd}
		Let~$q$ be an \sjfbcq.
		If $R(x) \land S(x)$ is a pattern of $q$, then $\ccountvals(q)$ is \shp-complete. Otherwise, $\ccountvals(q)$ is in \fp.
	\end{theorem}
	\begin{proof}
		We only need to prove the tractability claim, since hardness follows from Proposition~\ref{prp:RxSx-hard} and Lemma~\ref{lem:pattern-parsimonious}.
		We will assume without loss of generality that~$D$ contains no constants, as we can introduce a fresh null
		with domain~$\{c\}$ for every constant~$c$ appearing in~$D$, and the result is again a Codd table, and this does not change the output of the problem.
		Let~$q$ be~$R_1({\bar x_1}) \land  \ldots \land  R_m({\bar x_m})$.
		Observe that since~$q$ does not have~$R(x) \land  S(x)$ as a pattern then any two atoms cannot have a variable in common.
		But then, since~$D$ is a Codd table we have
		$$\ccountvals(q)(D) \, = \, \prod_{i=1}^m \ccountvals(R_i({\bar x_i}))(D(R_i)).$$
		Hence it is enough to show how to compute
		$\ccountvals(R_i({\bar x_i}))(D(R_i))$ for every~$1\leq i \leq m$.
		Let~${\bar t_1},\ldots,{\bar t_n}$ be the tuples of~$D(R_i)$.
		Let us write~$\rho({\bar t_j})$ for the number of valuations of the nulls appearing in~${\bar t_j}$ that do not match~${\bar x_i}$.
		Clearly, $\ccountvals(R_i({\bar x_i}))(D(R_i)) = \prod_{\bot \text{ appears in }D(R_i)} |\dom(\bot)| - \prod_{j=1}^n \rho({\bar t_j})$, so
		we only have to show how to compute~$\rho({\bar t_j})$ for~$1\leq j \leq n$.
		Since we can easily compute the total number of valuations of~${\bar t_j}$, it is enough to show how to compute the number of valuations of~${\bar t_j}$
		that match~${\bar x_i}$.
		For every variable~$x$ that appears in~${\bar x_i}$, compute the size of the intersection of the corresponding nulls in~${\bar t_j}$, and denote it~$s_x$.
		Then the number of valuations of~${\bar t_j}$ that match~${\bar x_i}$ is simply~$\prod_{x\text{ appears in }{\bar x_i}} s_x$.
		This concludes the proof.
	\end{proof}

	At this stage, we have completed the first 
	column of Table~\ref{tab:dichos-count-intro}, 
	and we also know that~$R(x,x)$
	is a hard pattern in the uniform setting for naïve tables (but not for Codd tables, by Theorem~\ref{thm:countvals-sjfcqs-codd}).
	In the next section, we treat the uniform setting.

	\subsection{The complexity on the uniform case}
	\label{subsec:countvals-uniform}
	We start our investigation with the case of naïve tables. In Proposition~\ref{prp:Rxx-hard}, we already showed that $\ucountvals(R(x,x))$ is \shp-hard. In the following proposition, we identify two other simple queries for which this problem is still intractable.

	\begin{proposition}
	\label{prp:RxSxyTy-RxySxy-hard}
		$\ucountvals(R(x) \land S(x,y) \land T(y))$ and $\ucountvals(R(x,y) \land S(x,y))$ are both \shp-hard. This holds even in the restricted setting 
		in which all nulls are interpreted over the same fixed domain~$\{0,1\}$.
	\end{proposition}
	\begin{proof}
		We reduce both problems from the problem of counting the number of independent sets in a graph (denoted by \sIS), which is \shp-complete~\cite{provan1983complexity}. 
		We start with~$\ucountvals(R(x)\land  S(x,y)\land  T(y))$. 
		Let $q = R(x)\land  S(x,y)\land  T(y)$ and~$G=(V,E)$ be a graph. 
		Then we define an incomplete database $D$ as follows.
		For every node~$v \in V$, we have a null~$\bot_v$, and the uniform domain is~$\{0,1\}$.
		For every edge~$\{u,v\} \in E$, we have facts~$S(\bot_u,\bot_v)$ and~$S(\bot_v,\bot_u)$ in $D$.
		Finally, we have facts~$R(1)$ and~$T(1)$ in $D$.
		For a valuation~$\nu$ of the nulls, consider the corresponding subset~$S_\nu$ of nodes of~$G$, given by $S_\nu = \{t \in V \mid \nu(\bot_t) = 1 \}$.
		This is a bijection between the valuations of the database and the node subsets of~$G$.
		Moreover, we have that~$\nu(D) \not\models q$ if and only if~$S_\nu$ is an independent set of~$G$.
		Since the total number of valuations of $D$ is $2^{|V|}$, we have that the number of independent sets of $G$ is equal to~$2^{|V|} - \ucountvals(q)(D)$.
		Hence, we conclude that $\sIS \tr \ucountvals(q)$. 
		The idea is similar for~$\ucountvals(R(x,y)\land  S(x,y))$: 
		we encode the graph with the relation~$S$ in the same way, and this time we add the fact~$R(1,1)$.
		\end{proof}
	As shown in the following result, it turns out that the three aforementioned patterns are enough to fully characterize the complexity of counting valuations for naïve tables in the uniform setting.
\begin{toappendix}
\subsection{Proof of Theorem~\ref{thm:countvals-naive-uniform}}
	\label{apx:countvals-naive-uniform}
	In this section we prove the tractability claim of the following dichotomy theorem.
\end{toappendix}

	\begin{theoremrep}[dichotomy]
	\label{thm:countvals-naive-uniform}
		Let~$q$ be an \sjfbcq. If $R(x,x)$ or $R(x) \land S(x,y) \land T(y)$ or~$R(x,y) \land S(x,y)$ is a pattern of $q$, then $\ucountvals(q)$ is \shp-complete.
		Otherwise, $\ucountvals(q)$ is in \fp.
	\end{theoremrep}

	The \shp -completeness part of the claim follows directly from what we have proved already. 
	Here, the most challenging part of the proof is actually the tractability part. We only present a simple example to give an idea of the proof technique, and defer the full proof to Appendix~\ref{apx:countvals-naive-uniform}.
	We will use the following definition.
		Given~$n,m \in \mathbb{N}$, let us write~$\surj_{n\rightarrow m}$ for the number of surjective functions from~$\{1,\ldots,n\}$ to~$\{1,\ldots,m\}$.
		By an inclusion--exclusion argument, one can show that~$\surj_{n\rightarrow m} = \sum_{i=0}^{m-1} (-1)^i \binom{m}{i} (m-i)^n$
		(for instance, see~\cite{stackexchange_surjective}).
		It is clear that this can be computed in FP, when~$n$ and~$m$ are given in unary.

	\begin{example}
	\label{expl:RxSx}
		{\em Let~$q$ be the \sjfbcq\ $R(x)\land  S(x)$, and~$D$ be an incomplete database over relations~$R,S$.
		Notice that~$q$ does not have any of the patterns mentioned in Theorem~\ref{thm:countvals-naive-uniform}. We
		will show that~$\ucountvals(q)$ is in \fp.
		Since~$q$ contains only two unary atoms we can also assume without loss of generality 
		that the input~$D$ is a Codd table (otherwise all valuations are satisfying).
		
		Since we can compute in FP the total number of valuations, it is enough to show how to compute the number of valuations of~$D$ that do not satisfy~$q$.
		Let $\dom$ be the uniform domain,~$d$ be its size,
		$n_R$ (resp., $n_S$) be the number of nulls in~$D(R)$ (resp., in $D(S)$) and~$C_R$ (resp., $C_S$) be the set of constants occurring
		in~$D(R)$ (resp., in $D(S)$), with~$c_R$ (resp., $c_S$) its size.
		We can assume without loss of generality that~$C_R \cap C_S = \emptyset$, as otherwise all the valuations are satisfying, and this is computable in PTIME.
		Furthermore, we can also assume that~$C_R \cup C_S \subseteq \dom$, since we can remove the constants that are not in~$\dom$, as these can never match.
		
		Let~$M := \dom \setminus(C_R \cup C_S)$, and~$m$ its size (i.e., with our assumptions we have $m = d-c_R-c_S$).
		Fix some subsets~$M' \subseteq M$ and~$R' \subseteq C_R$. The quantity
		$\surj_{n_R \to |M'|+|R'|}$ then counts the number of valuations of the nulls of~$D(R)$ that span exactly~$M'\cup R'$.
		Moreover, letting~$\nu_R$ be a valuation of the nulls of~$D(R)$ that spans exactly~$M'\cup R'$, the quantity
		$(d-c_R-|M'|)^{n_S}$ is the number of ways to extend~$\nu_R$ into a valuation~$\nu$ of all the nulls of~$D$ so that~$\nu(D) \not\models q$: indeed,
		every null of~$D(S)$ can take any value in~$\dom \setminus(C_R \cup M')$.
		The number of valuations of~$D$ that do not satisfy~$q$ is then (keeping in mind that a null in~$D(R)$ cannot take a value in~$C_S$):
		\[ \sum_{\substack{M' \subseteq M \\ R' \subseteq C_R}} \surj_{n_R \to |M'|+|R'|} \times (d-c_R-|M'|)^{n_S}\]
		and since the summands only depends on the sizes of~$M'$ and~$R'$, this is equal to
		\[ \sum_{\substack{0 \leq m' \leq m \\ 0 \leq r' \leq c_R}} \binom{m}{m'} \binom{c_R}{r'} \surj_{n_R \to m'+r'} \times (d-c_R-m')^{n_S}\]
		This last expression can clearly be computed in PTIME.\footnote{Note that in the sum we do not need to specify that~$m'+r' \leq n_R$, as when~$a <b$ we have~$\surj_{a \to b} = 0$.} } \qed
	\end{example}

	\begin{toappendix}
	First, to characterize the queries that do not have these patterns, we will use the notion of \emph{connectivity graph} of a sj-free CQ~$q$:

	\begin{definition}
	\label{def:connectivity-graph}
		Let~$q$ be a sj-free CQ.
		The \emph{connectivity graph of~$q$} is the graph~$G_q = (V,E)$ with labeled edges, where~$V$ is the set of atoms of~$q$, and for every
		two atoms~$R(\bar{x_i}),S(\bar{y_i})$ of~$q$, if they share a variable then we have an edge between the corresponding nodes of~$G_q$, that edge being labeled
		with the variables in~$\bar{x_i} \cap \bar{y_i}$.
	\end{definition}

	\begin{example}
		\label{expl:connectivity-graph}
		Figure~\ref{fig:connectivity-graph} shows the connectivity graph of the query
			\[R_1(x_1,x_1,y_1,t_1),R_2(x_1,y_1,t_2),S_1(x_2,t_3),S_2(x_2,t_4),\\S_3(x_2),T_1(x_3),T_2(x_3),T_3(x_3),T_4(x_3,t_5).\]
	\end{example}

	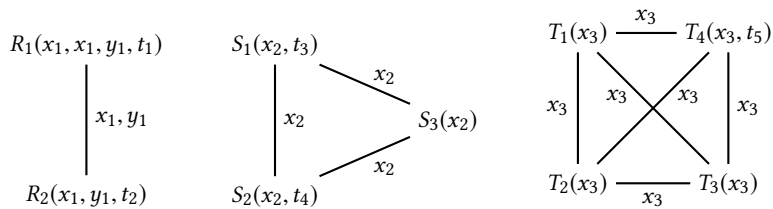
\begin{figure}
		\centering
			\begin{tikzpicture}

\node (r1) at (0,2) {$R_1(x_1,x_1,y_1,t_1)$};
\node (r2) at (0,0) {$R_2(x_1,y_1,t_2)$};

\draw[black,thick] (r1) -- (r2) node[midway,right] {$x_1,y_1$};

\end{tikzpicture}
\qquad
\begin{tikzpicture}
\node (s1) at (0,2) {$S_1(x_2,t_3)$};
\node (s2) at (0,0) {$S_2(x_2,t_4)$};
\node (s3) at (2.3,1) {$S_3(x_2)$};

\draw[black,thick] (s1) -- (s2) node[midway,right] {$x_2$};
\draw[black,thick] (s1) -- (s3) node[midway,right,yshift=.3em] {$x_2$};;
\draw[black,thick] (s2) -- (s3) node[midway,right,yshift=-.5em] {$x_2$};;

\end{tikzpicture}
\qquad
\begin{tikzpicture}
\node (t1) at (0,2) {$T_1(x_3)$};
\node (t2) at (0,0) {$T_2(x_3)$};
\node (t3) at (2,0) {$T_3(x_3)$};
\node (t4) at (2,2) {$T_4(x_3,t_5)$};

\draw[black,thick] (t1) -- (t2) node[midway,left] {$x_3$};
\draw[black,thick] (t2) -- (t3) node[midway,below] {$x_3$};
\draw[black,thick] (t3) -- (t4) node[midway,right] {$x_3$};
\draw[black,thick] (t4) -- (t1) node[midway,above] {$x_3$};

\draw[black,thick] (t1) -- (t3) node[midway,xshift=-1.5em,yshift=.5em] {$x_3$};
\draw[black,thick] (t2) -- (t4) node[midway,xshift=1.5em,yshift=.5em] {$x_3$};

\end{tikzpicture}

		\caption{The connectivity graph~$G_q$ of the sj-free CQ~$q$ from Example~\ref{expl:connectivity-graph}.}
		\label{fig:connectivity-graph}
	\end{figure}

	The following is then readily observed:

	\begin{lemma}
	\label{lem:connectivity-graph-clique-naive}
		Let~$q$ be a sj-free CQ that does not contain any of the patterns mentioned in Theorem~\ref{thm:countvals-naive-uniform}.
		Then for every connected component~$C$ of~$G_q$,~$C$ is a clique and there exists a variable such that all edges of~$C$ are labeled by exactly that variable. 
	\end{lemma}
	\begin{proof}
		First, observe that every edge of~$G_q$ must be labeled by exactly one variable, as otherwise the query~$q$ would contain the pattern~$R(x,y)\land  S(x,y)$.
		Let~$C$ be a connected component of~$G_q$. Then we have:
		\begin{itemize}
			\item $C$ is a clique. Indeed, assume by contradiction that~$C$ is not a clique. Then, since~$C$ is connected and is not a clique, we can find
				3 nodes~$A_1(\overline{x}),A_2(\overline{x'}),A_3(\overline{x''})$ such that~$A_1(\overline{x})$ is adjacent to~$A_2(\overline{x'})$, $A_2(\overline{x'})$ is adjacent to~$A_3(\overline{x''})$, and~$A_1(\overline{x})$ is not adjacent to~$A_3(\overline{x''})$. Let~$X$ be $\bar{x} \cap \bar{x'}$ and~$Y$ be~$\bar{x'} \cap \bar{x''}$, i.e., the labels on the two corresponding edges of~$C$. By definition of~$G_q$ and since~$A_1(\overline{x})$ is not adjacent to~$A_3(\overline{x''})$, we must have~$X \cap Y = \emptyset$.
				But~$X$ and~$Y$ are not empty (again by definition of~$G_q$), so by picking~$x$ in~$X$ and~$y$ in~$Y$ we see that~$q$ contains the pattern~$R(x)\land  S(x,y)\land  T(y)$, a contradiction.
			\item There exists a variable that labels every edge of~$C$. Indeed, since every edge of~$G_q$ is labeled by exactly one variable, and since~$C$ is a clique, if it was not the case then again we could find the pattern~$R(x)\land  S(x,y)\land  T(y)$ in~$q$.
		\end{itemize}

	This concludes the proof.
	\end{proof}
	
	For instance, the query from Example~\ref{expl:connectivity-graph} does not satisfy this criterion, since the edge in the first connected component of~$G_q$ is labeled by two variables. However if we consider the query $S_1(x_2,t_3),S_2(x_2,t_4),S_3(x_2),T_1(x_3),T_2(x_3),T_3(x_3),T_4(x_3,t_5)$ (i.e., we remove the first connected component), then it satisfies the criterion.

	We will also use the general fact that for a sj-free CQ~$q$, we can assume wlog that~$q$ does not contain variables that occur only once:

	\begin{lemma}
	\label{lem:remove-ear-variables}
		Let~$q$ be a sj-free CQ,
		and let~$q'$
		be the sj-free CQ obtained from~$q$ by deleting all the variables that have only one occurrence in~$q$.
		Then~$\ucountvals(q) \leq_{\mathrm{T}}^{\mathrm{p}} \ucountvals(q')$.
	\end{lemma}
	\begin{proof}
		Let~$D$ be an incomplete database input of~$\ucountvals(q)$. 
		Let~$S$ be set of nulls~$\bot$ such that:
		\begin{itemize}
			\item $\bot$ occurs in a column corresponding to a variable that has been deleted; and
			\item $\bot$ does not occur in a column corresponding to a variable that has not been deleted.
		\end{itemize}
		Then, letting~$D'$ be the database obtained from~$D$ by projecting out the columns corresponding to the deleted variables,
		it is clear that we have~$\ucountvals(q)(D) = \ucountvals(q')(D') \times \prod_{\bot \in S}|\dom(\bot)| $, where~$\dom$ is the uniform domain of the nulls.
		We note here that this lemma is also true in the non-uniform setting.
	\end{proof}

	By Lemma~\ref{lem:connectivity-graph-clique-naive} and Lemma~\ref{lem:remove-ear-variables}, it is enough to show the tractability of~$\ucountvals(q)$ when~$q$ is of the form~$C_1(x_1) \land  \ldots \land  C_m(x_m)$,
	where each~$C_i(x_i)$ is what we call a \emph{basic singleton query}, i.e., is a conjunction of unary atoms over the same variable~$x_i$.
	We call such a sj-free CQ a \emph{conjunction of basic singletons}. For instance, \[S_1(x_2),S_2(x_2),S_3(x_2),T_1(x_3),T_2(x_3),T_3(x_3),T_4(x_3)\] is such a query, with~$m=2$.
	We will use the following: 

	\begin{lemma}
		\label{lem:IE}
		Let~$q=C_1(x_1) \land  \ldots \land  C_m(x_m)$ be a conjunction of basic singletons sj-free query, and let~$D$ be an incomplete database.
		For~$S \subseteq [m]$, we define~$N_S(D) \defeq |\{\nu \text{ valuation of } D \mid \nu(D) \not\models \bigvee_{i \in S} C_i(x_i)\}|$.
		Then we have~$\ucountvals(q)(D) = \sum_{S \subseteq [m]} (-1)^{|S|} N_S(D)$.
	\end{lemma}
	\begin{proof}
		Direct, by inclusion--exclusion.
	\end{proof}

	Hence, and remembering that we consider data complexity, it is enough to show how to compute~$N_S(D)$ for every~$S \subseteq [m]$.
	The main difficulties in computing~$N_S(D)$ is that the relations can have nulls in common (since we consider naïve tables), and that they may also have constants; this makes it technically painful to express a closed-form expression for~$N_S(D)$.
	We explain how to do it next, thus finishing the proof of Theorem~\ref{thm:countvals-naive-uniform}.

	\begin{proposition}
		\label{prp:compute-NSD}
		Let~$q=C_1(x_1) \land  \ldots \land  C_m(x_m)$ be a conjunction of basic singletons sj-free query and~$S \subseteq [m]$.
		Then, given an incomplete database~$D$ as input, we can compute~$N_S(D)$ in polynomial time.
	\end{proposition}
	\begin{proof}
		We strongly advise the reader to have well understood Example~\ref{expl:RxSx} before reading this proof.
		First, observe that to compute~$N_S(D)$ we can assume without loss of generality that the input database~$D$ only contains facts over relation names that occur in
		some~$C_i(x_i)$, for~$i\in S$.
		Indeed, $N_S(D)$ counts the valuations~$\nu$ of~$D$ that do not satisfy any of the~$C_i(x_i)$ for~$i\in S$,
		so that for any~$j \notin S$ we do not care if~$\nu$ satisfies~$C_j(x_j)$ or not; hence, we could simply multiply the result by the appropriate factor.
		Therefore, we can assume that~$S$ is~$[m]$.
		We now need to fix some notation.
		Let us write the conjunction of basic singleton sj-free query~$q$ as
		\[R_1(x_1) \land  \ldots \land  R_{m_1}(x_1) \land  R_{m_1 + 1}(x_2) \land  \ldots \land  R_{m_1 + m_2}(x_2) \land  \ldots \land  R_{\sum_{i=1}^{m-1} m_i}(x_m) \land  \ldots \land  R_{\sum_{i=1}^m m_i}(x_m)\]
		and let~$K$ be the number of atoms in~$q$, that is, $K \defeq \sum_{i=1}^m m_i$.
		Let~$\dom$ be the uniform domain of the nulls occurring in~$D$ and~$d$ its size.
		For~$\mathbf{s} \subseteq [K]$, we write~$C_\mathbf{s}$ the set of constants that occur in
		each of the relations~$D(R_i)$ for~$i\in \mathbf{s}$ but in none of the others, and write~$c_\mathbf{s}$ the size of that set. We call such a set a
		\emph{block} of constants.
		Similarly for the nulls, we write~$N_\mathbf{s}$ the set of nulls that occur in each of the relations~$D(R_i)$ for~$i\in \mathbf{s}$ but in none of the others (and we call this a block of nulls), and~$n_\mathbf{s}$ for its size. We can assume wlog that:
		\begin{description}
			\item[(a)] For every~$1 \leq i \leq m$, there is no constant that occurs in every~$D(R)$ for~$R$ a relation name in~$C_i(x_i)$.
				Indeed otherwise any valuation would satisfy~$C_i(x_i)$, thus~$N_{[m]}(D)$ would simply be~$0$.
			\item[(b)] Every constant~$c$ appearing in~$D$ is in~$\dom$. Indeed otherwise, with the last item, this constant would have no chance to
				be part of a match, so we could simply remove it (i.e., remove all tuples of the form~$R(c)$ from~$D$).
		\end{description}
		For a subset~$A \subseteq \dom$, let us write~$A^\complement \defeq \dom \setminus A$.
		Finally, for a set~$Z=\{A_1,\ldots,A_l\}$ of subsets of~$\dom$, we denote by~$\I(Z)$ the set
		\[\I(Z) \defeq \{ \bigcap_{i=1}^l B_i \mid (B_1,\ldots,B_l) \in \{A_1,A_1^\complement\} \times \ldots \times \{A_l,A_l^\complement\} \} \]

		We now explain informally how we can compute~$N_{[m]}(D)$.
		Let~$L=\mathbf{s}_1,\ldots,\mathbf{s}_{2^K}$ be an arbitrary linear order of the set of subsets of~$[K]$.
		We will define by induction on~$i\in [2^K]$ an expression computing~$N_{[m]}(D)$, which will be a nested sum of the form

		\begin{equation}
		\label{eq:huge-sum}
			\sum_{\text{something}_{\mathbf{s}_1}} f_{\mathbf{s}_1} \times \bigg( \sum_{\text{something}_{\mathbf{s}_2}} f_{\mathbf{s}_2} \times 
			\big( \ldots (\sum_{\text{something}_{\mathbf{s}_{2^K}}} f_{\mathbf{s}_{2^K}}) \ldots \big) \bigg)
		\end{equation}

		where each~$\text{something}_{\mathbf{s}_i}$ sums over the possible images~$A_{\mathbf{s}_i}$ of the nulls in~$N_{\mathbf{s}_i}$ by a valuation, and~$f_{\mathbf{s}_i}$ 
		will simply be~$\surj_{n_{\mathbf{s}_i} \to a_{\mathbf{s}_i}}$, where~$a_{\mathbf{s}_i} \defeq |A_{\mathbf{s}_i}|$, i.e., the number of valuations~$\nu$
		of~$N_{\mathbf{s}_i}$ with image exactly~$A_{\mathbf{s}_i}$. But there are two technicalities:
		\begin{itemize}
			\item First, we need to ensure that each basic singleton query~$C_i(x_i)$ of~$q$ will not be satisfied. In order to 
				do that, $\text{something}_{\mathbf{s}_i}$ will actually sum over all
				the possible partitions~$(B_{\mathbf{s_i}}^1,\ldots,B_{\mathbf{s_i}}^{|\I(Z_{i-1}|)})$ of~$A_{\mathbf{s}_i}$,
				where each of the~$B_{\mathbf{s}_i}^j$ 
				is included in one of the sets in~$\I(Z_{i-1})$, where~$Z_{i-1}$ contains all the blocs of constants and all the other~$B_{\mathbf{s_j}}^r$ for~$j < i$.
				We iteratively build that sum from the outside to the inside, starting with~$Z_0 \defeq \{\dom\} \cup \{C_{\mathbf{s}} \mid \mathbf{s} \subseteq [K]\}$.
				This will allow us to avoid summing over the~$B_{\mathbf{s_i}}^j$ that would render a basic singleton query true.
			\item Second, as is, such a sum is obviously not going to be computable in PTIME, as we are summing over subsets of~$\dom$.
				To fix this, observe that the value of the subsum for~$\mathbf{s}_i$ actually only depends on the \emph{sizes} of the sets in~$Z_{i-1}$.
				Hence, iterating from the outside to the inside, whenever~$\text{something}_{\mathbf{s}_i}$ contains a
				sum of the form, say, $B_{\mathbf{s_i}}^k \subseteq B_{\mathbf{s_j}}^{k'}$ for~$j<i$,
				we can replace this with a sum over~$0 \leq b_{\mathbf{s_i}}^k \leq b_{\mathbf{s_j}}^{k'}$, and add to~$f_{\mathbf{s}_i}$ a factor
				of~$\binom{b_{\mathbf{s_j}}^{k'}}{b_{\mathbf{s_i}}^k}$.
				Now, because of how~$Z_0$ is defined, and because of how~$\I$ works, all the initial numbers in the first sum are either 
				$|\dom \setminus \bigcup_{i=1}^K C_{\{i\}}|$ or	one of the numbers~$c_{\mathbf{s}}$
				for~$\mathbf{s} \subseteq [K]$. These can all be computed in polynomial time.
		\end{itemize}

		The resulting expression then indeed evaluates to~$N_{[m]}(D)$, and is in a form
		that allows us to directly compute it in polynomial time (but non-elementary in the query).
		This concludes the proof of Proposition~\ref{prp:compute-NSD}.
	\end{proof}

\end{toappendix}

	We conclude this section by turning our attention to the 
	case of Codd tables.
	Notice that none of the results proved so far 
	provides a hard pattern in this case.
	We identify in the following proposition a simple query for which the problem is intractable.

	\begin{toappendix}
		\subsection{Proof of Proposition~\ref{prp:RxSxyTy-hard-codd}}
	\end{toappendix}

	\begin{propositionrep}
	\label{prp:RxSxyTy-hard-codd}
		$\cucountvals(R(x) \land S(x,y) \land T(y))$ is \shp-hard.
	\end{propositionrep}
	\begin{proof}
		We reduce from the problem of counting the number of independent sets of a bipartite graph, written~\#BIS, which is \shp-hard~\cite{provan1983complexity}.
		Let~$G = (X \sqcup Y, E)$ be a bipartite graph.
		Without loss of generality,
		we can assume that~$|X|=|Y|=n$; indeed, if~$|X| < |Y|$ then we could simply
		add~$|Y|-|X|$ isolated nodes to complete the graph, which simply multiplies
		the number of independent sets
		by~$2^{|Y|-|X|}$.
		Also, observe that counting the number of independent sets of~$G$ is the same as counting the number of
		pairs~$(S_1,S_2)$ with~$S_1 \subseteq X, S_2 \subseteq Y$, such
		that~$(S_1 \times S_2) \cap E = \emptyset$. We will call such a pair an \emph{independent pair}.
		For~$0 \leq i,j \leq n$, let~$Z_{i,j}$ be the number of independent pairs~$(S_1,S_2)$ such
		that~$|S_1| = i$ and~$|S_2| = j$.
		It is clear that ($\star$) the number of independent sets of~$G$ is then~$\#\mathrm{BIS}(G) = \sum_{0 \leq i,j \leq n} Z_{i,j}$.
		The idea of the reduction is to construct in polynomial time $(n+1)^2$ incomplete databases~$D_{a,b}$
		for~$0 \leq a,b \leq n$ such that, letting~$C_{a,b}$ be the number of valuations~$\nu$
		of~$D_{a,b}$ with~$\nu(D_{a,b}) \not \models R(x)\land  S(x,y)\land  T(y)$,
		the values of the variables~$Z_{i,j}$ and~$C_{i,j}$ form a linear system of
		equations~$\mA \mZ = \mC$, with~$\mA$ an invertible matrix.
		This will allow us, using $(n+1)^2$ calls to an
		oracle for $\cucountvals(R(x)\land  S(x,y)\land  T(y))$, to recover the~$Z_{i,j}$ values,
		and then to compute~$\#\mathrm{BIS}(G)$ using ($\star$).
		We now explain how we construct~$D_{a,b}$ from~$G$
		for~$0 \leq a,b \leq n$, and define~$\mA$.
		First, we fix an arbitrary linear order~$x_1,\ldots,x_n$ of~$X$, and similarly
		$y_1,\ldots,y_n$ for~$Y$.
		The database~$D_{a,b}$ has constants~$a_i$ for~$1 \leq i \leq n$, and has a
		fact~$S(a_i,a_j)$ whenever~$(x_i,y_j) \in E$. It has nulls~$\bot_1,\ldots,\bot_a$ and
		facts~$R(\bot_i)$ for~$1 \leq i \leq a$ (if~$a=0$ there are no such nulls and facts),
		and nulls~$\bot_1',\ldots,\bot'_b$ and
		facts~$T(\bot'_i)$ for~$1 \leq i \leq b$; in particular, this is a Codd table.
		The uniform domain of the nulls is~$\{a_i \mid 1 \leq i \leq n\}$.
		Given a valuation~$\nu$ of~$D_{a,b}$, let $P(\nu)$ be the pair of subsets of~$V$ defined by
			\[P(\nu) \defeq (\{x_i \mid \exists 1 \leq k \leq a \text{ s.t. } \nu(\bot_k)=a_i\}, \\
			\{y_i \mid \exists 1 \leq k \leq b \text{ s.t. } \nu(\bot'_k)=a_i \})\]
		One can then easily check that the following two claims hold:
		\begin{itemize}
			\item For every valuation~$\nu$ of~$D_{a,b}$, we have that~$\nu(D_{a,b}) \not\models R(x)\land  S(x,y)\land  T(y)$ iff~$P(\nu)$ is an independent
				pair of~$G$;\footnote{This observation, and in fact the idea of reducing from \#BIS, is due to Antoine Amarilli.}
			\item For every independent pair~$(S_1,S_2)$ of~$G$, there are
				exactly $\surj_{a \to |S_1|} \times \surj_{b \to |S_2|}$ valuations~$\nu$ such that~$P(\nu) = (S_1,S_2)$.
		\end{itemize}

		But then, we have $C_{a,b} = \sum_{0 \leq i,j \leq n} (\surj_{a \to i} \times \surj_{b \to j}) Z_{i,j}$.
		In other words, we have the linear system of equations~$\mA \mZ = \mC$, where~$\mA$ is
		the~$(n+1)^2 \times (n+1)^2$ matrix
			defined by $\mA_{(a,b),(i,j)} \defeq \surj_{a \to i} \times \surj_{b \to j}$.
			This matrix is the Kronecker product~$\mA' \otimes \mA' $ of the~$(n+1) \times (n+1)$
			matrix with entries~$\mA'_{a,i} \defeq \surj_{a \to i}$.
			Since~$\mA'$ is a triangular matrix with non-zero coefficients on the diagonal, it is invertible, hence so is~$\mA$, which concludes the proof.  
\end{proof}
In Proposition~\ref{prp:RxSxyTy-RxySxy-hard}, we proved that $\ucountvals(R(x) \land S(x,y) \land T(y))$ is \shp-hard in the general case where na\"ive tables are allowed. Hence, Proposition~\ref{prp:RxSxyTy-RxySxy-hard} was in fact a consequence of Proposition~\ref{prp:RxSxyTy-hard-codd}, where only Codd tables are allowed. However, we decided to provide a separate proof for Proposition~\ref{prp:RxSxyTy-RxySxy-hard}, because this result includes another intractable case, and both cases in Proposition~\ref{prp:RxSxyTy-RxySxy-hard} can be established via a simple reduction from counting independent sets. By contrast, Proposition~\ref{prp:RxSxyTy-hard-codd} requires of a more complicated proof (we reduce from \sIS\ on bipartite graphs and use a Turing reduction with~$(\nicefrac{n}{2}+1)^2$ calls to the oracle, where~$n$ is the number of nodes of the input graph, to form a system of linear equations which we then invert to recover the number of independent sets).

  Up to this point, we have not been able to prove that when an \sjfbcq~$q$ does not contain~$R(x) \land S(x,y) \land T(y)$ as a pattern, it holds that~$\cucountvals(q)$ is in \fp.
  Thus, the possibility of having a dichotomy in this case is left as a problem for future research. 
  Nevertheless, we can still observe that restricting to Codd tables simplifies the problem of counting valuations in the non-uniform setting. Indeed, considering the query~$R(x,x)$, counting valuations is \shp-hard for naïve tables, while it is in FP for Codd tables by Theorem~\ref{thm:countvals-sjfcqs-codd}.

\section{Dichotomies for counting completions}
\label{sec:countcompls-sjfcqs}

In this section, we study the complexity of the problems of counting completions
satisfying an~\sjfbcq~$q$, 
in the four cases that can be obtained by considering naïve or Codd tables and non-uniform or uniform domains.
We will again use the notion of pattern as introduced in Definition~\ref{def:pattern}. Our first step is to show that Lemma~\ref{lem:pattern-parsimonious}, which we used in the last section for the problems or counting valuations, extends to the problems of counting completions. 
\begin{toappendix}
\subsection{Proof of Lemma~\ref{lem:pattern-parsimonious-compls}}
\end{toappendix}

\begin{lemmarep}
\label{lem:pattern-parsimonious-compls}
	Let~$q,q'$ be \sjfbcqs\ such that~$q'$ is a pattern of~$q$.
	Then we have that~$\countcompls(q') \pr \countcompls(q)$ and $\ucountcompls(q') \pr \ucountcompls(q)$. Moreover, the same results hold if we restrict to the case of Codd tables.
\end{lemmarep}
\begin{proof}
	The reduction is exactly the same as the one of Lemma~\ref{lem:pattern-parsimonious}.
		To show that this reduction works properly for counting completions, it is enough to observe
		that for every valuations~$\nu_1,\nu_2$ of~$D'$, we have that~$\nu_1(D')=\nu_2(D')$
		iff~$\nu_1(D) = \nu_2(D)$.
\end{proof}

We will then follow the same general strategy as in the last section, i.e., prove hardness for some simple patterns and combine these with Lemma~\ref{lem:pattern-parsimonious-compls} and tractability proofs to obtain dichotomies. 
Our findings are summarized in the last two columns of Table \ref{tab:dichos-count-intro} in the introduction. 
We start in Section~\ref{subsec:countcompls-non-uniform} with the non-uniform cases and continue in Section~\ref{subsec:countcompls-uniform} with the uniform cases. Again, we explicitly state when 
	a \shp-hardness result holds even in the restricted setting in which there is a fixed 
	domain over which nulls are interpreted.

\subsection{The complexity on the non-uniform case}
\label{subsec:countcompls-non-uniform}
	Here, we study the complexity of the problems $\countcompls(q)$ and $\ccountcompls(q)$,
	providing dichotomy results in both cases.
	In fact, it turns out that these problems are~\shp-hard for all \sjfbcqs.
	 To prove this, it is enough to show that the problem~$\ccountcompls(R(x))$ is hard, 
	 that is, even counting the completions of a single unary table is \shp-hard in the non-uniform setting.

\begin{proposition}
\label{prp:countcompls}
	$\ccountcompls(R(x))$ is \shp-hard.
\end{proposition}
\begin{proof}
We provide a polynomial-time parsimonious reduction from 
	the problem of counting the vertex
	covers of a graph,
	which we denote by \sVC.
	Let~$G=(V,E)$
	be a graph. We construct a Codd table~$D$ using a single
	unary relation~$R$ such that the number of completions of~$D$ equals
	the number of vertex covers of~$G$.  For every edge~$e=\{u,v\}$ of~$G$,
	we have one null~$\bot_e$ with $\dom(\bot_e) = \{u,v\}$ and the
	fact~$R(\bot_e)$.  Let~$a$ be a fresh constant.  For every
	node~$u\in V$ we have a null~$\bot_u$ with~$\dom(\bot_u) = \{u,a\}$ and
	the fact~$R(\bot_u)$.  Last, we add the fact~$R(a)$.  We now show
	that the number of completions of~$D$ equals the number of vertex
	covers of~$G$. 
	
	Let~$\VC(G)$ be the set of vertex covers of~$G$.  For a
	valuation~$\nu$ of~$D$, define the set~$S_\nu := \{u \in V \mid
	R(u) \in D\}$.  Since the fact~$R(a)$ is in every completion
	of~$D$, it is clear that the number of completions of~$D$ is equal
	to~$|\{S_\nu \mid \nu \text{ is a valuation of } D\}|$.  We claim
	that~$\VC(G)=\{S_\nu \mid \nu \text{ is a valuation of } D\}$, which
	shows that the reduction works. ($\subseteq$) Let~$C \in \VC(G)$, and let us show that
			there exists a valuation~$\nu$ of~$D$ such that~$S_\nu
			= C$.  For a null of the form~$\bot_e$
			with~$e=\{u,v\}\in E$, assuming wlog that~$u \in C$, we
			define~$\nu(\bot_e)$ to be~$u$.  For a null of the
		form~$\bot_u$ with~$u\in V$, we define~$\nu(\bot_u)$ to be~$u$
	if~$u\in C$ and~$a$ otherwise.  It is then clear that~$S_\nu =
C$. ($\supseteq$) Let~$\nu$ be a valuation of~$D$, and let us show
	that~$S_\nu$ is a vertex cover.  Assume by contradiction that there is
		an edge~$e=\{u,v\}$ such that~$e \cap S_\nu = \emptyset$.  By
		definition of~$D$, we must have~$\nu(\bot_e) \in \{u,v\}$, so
		that one of~$u$ or~$v$ must be in~$S_\nu$, hence a
contradiction. Therefore, we conclude that $\sVC \pr \ccountcompls(R(x))$.
\end{proof}

Recall from Section~\ref{sec:preliminaries} that, to avoid trivialities, we assume all \sjfbcqs\ to contain at least one atom and that the arity of every atom is~$\geq 1$ (that is, all atoms have at least one variable).
Using Lemma~\ref{lem:pattern-parsimonious}, this allows us to obtain the following dichotomy result.

\begin{theorem}[Dichotomy]
\label{thm:countcompls-sjfcq-hard} 	
For every \sjfbcq~$q$, it holds that $\countcompls(q)$ and $\ccountcompls(q)$ are \shp-hard. 	
\end{theorem}

Notice here that we do not claim membership in \shp; in fact, we will come back to this issue in Section~\ref{sec:misc} to show that this is unlikely to be true for naïve tables.
However, we can still show that 
membership in \shp\ holds for Codd tables.
We then obtain:

	\begin{toappendix}
		\subsection{Proof of Proposition~\ref{thm:countcompls-sjfcqs-complete-codd}}
		\label{subsec:countcompls-sjfcqs-complete-codd}
	\end{toappendix}

\begin{theoremrep}[Dichotomy]
	\label{thm:countcompls-sjfcqs-complete-codd}
For every \sjfbcq~$q$, the problem $\ccountcompls(q)$ is \shp-complete.
\end{theoremrep} 
\begin{proof}
	Hardness is from Theorem~\ref{thm:countcompls-sjfcq-hard}. To show membership in \shp\, we will actually prove a more general result, which we prove in
	Appendix~\ref{apx:countcompls-codd-sharp-p}. There, we show that for every Boolean query~$q$ such that~$q$ has model checking in \ptime\, the problem~$\ccountcompls(q)$ is in \shp. This in particular applies to all \sjfbcqs.
\end{proof}

\begin{toappendix}
	\subsection{Proof of membership in \shp\ of~$\ccountcompls(q)$}
	\label{apx:countcompls-codd-sharp-p}
	We recall that the \emph{model checking} of a Boolean query~$q$ is the problem of deciding, given a (complete) database~$D$, whether $D \models q$.
	In this section, we call a fact that contains only constants a \emph{ground fact}.
	Our goal here is to prove the following.
\begin{proposition}
	\label{prp:countcompls-codd-sharp-p}
	If a Boolean query~$q$ has model checking is in \ptime, then we have that $\ccountcompls(q)$ is in~\shp.
\end{proposition}

	To show this, we first prove that we can check in polynomial time if a given set of ground facts is a possible completion of an incomplete database:
	\begin{lemma}
		\label{lem:matchings}
		Given as input an incomplete Codd table~$D$ and a set~$S$ of ground facts, we can decide in polynomial time whether there exists a valuation~$\nu$ of~$D$ such that~$\nu(D) = S$.
	\end{lemma}
	\begin{proof}
		For every fact~$f$ of~$D$, let us denote by~$P(f)$ the
		set of ground facts that can be obtained from~$f$ via a valuation ($P(f)$ can
		be~$\{f\}$ if~$f$ is already a ground fact).
		The first step is to check that for every fact~$f$ of~$D$, it holds
		that ($\star$)~$P(f) \cap S \neq \emptyset$. If this is not the case, then we know for sure that for every
		valuation~$\nu$ of~$D$ we will have~$\nu(D) \not\subseteq S$, so that we can safely reject.
		Next, we build the bipartite graph~$G_{D,S}$ defined as follows:
		the nodes in the left partition of~$G_{D,S}$ are the facts of~$D$, the nodes in the right partition
		are the facts in~$S$, and we connect a fact~$f$ of~$D$ with all the ground facts in the right
		partition that are in~$S \cap P(f)$.
		It is clear that we can construct~$G_{D,S}$ in polynomial time.
		We then compute in polynomial time the size~$m$ of a maximum-cardinality matching of~$G_{D,S}$,
		for instance using~\cite{edmonds1965paths}. It is clear that we have~$m \leq |S|$.
		At this stage, we claim that there exists a valuation~$\nu$ of~$D$ such that~$\nu(D) = S$
		if and only if~$m= |S|$.
		We prove this by analysing the two possible cases:
		\begin{itemize}
			\item If~$m < |S|$, then let us show that there is no such valuation. Indeed,
				assume by way of contradiction that such a valuation~$\nu$ exists.
				Let~$B$ be a subset of~$D$ of minimal size such that~$\nu(B) = S$.
				It is clear that such a subset exists, and moreover that its size is
				exactly~$|S|$. But then, consider the set~$M$ of edges of~$G_{D,S}$
				defined by~$M \defeq \{(f,\nu(f)) \mid f \in B\}$.
				Then~$M$ is a matching of~$G_{D,S}$ of size~$|S| > m$,
				contradicting the fact that~$m$ is the size of a maximum-cardinality matching.
			\item If~$m = |S|$, let us show that such a valuation exists.
				Let~$M$ be a matching of~$G_{D,S}$
				of size~$|S|$.
				By the pigeonhole principle, it is clear that every node corresponding to a ground fact~$f \in S$ is incident to (exactly)
				one edge of~$M$; let us denote that edge by~$e_f$.
				Moreover, since~$M$ is a matching, the mapping that associates to a ground fact~$f \in S$ the
				fact~$f'_f$ at the other end of~$e_f$ is injective.
				Hence, we can define~$\nu(\bot)$ of every null~$\bot$ occurring in such a fact $f'_f \in D$
				to be the unique constant such that~$\nu(f'_f) = f$ holds, and
				for every other fact~$f'$ in~$D$ not incident to an edge in~$M$, we chose a value for its nulls
				so that~$\nu(f') \in S$,
				which we can do thanks to~($\star$).
				It is then clear that we have~$\nu(D) = S$.
		\end{itemize}
		But then, we can simply accept if~$m = |S|$ and reject otherwise.
	\end{proof}

	We can now prove Proposition~\ref{prp:countcompls-codd-sharp-p}:

	\begin{proof}[Proof of Proposition~\ref{prp:countcompls-codd-sharp-p}.]
	We define a non-deterministic turing machine~$M_q$ such that, on input incomplete Codd table~$D$,
	its number of accepting computation paths is exactly the number of completions of~$D$
	that satisfy~$q$. 
	First, compute in polynomial time the set~$A \defeq \bigcup_{f \in D} P(f)$, where~$P(f)$ is defined just as in Lemma~\ref{lem:matchings}.
	Then, the machine~$M_q$ guesses a subset~$S$ of~$A$.
	It then checks in polynomial time if~$S$, when seen as a database, satisfies~$q$, and rejects if it is not the case.
		Then, using Lemma~\ref{lem:matchings}, it checks in polynomial time
		whether there exists a valuation~$\nu$ of~$D$ such that~$\nu(D) = S$, and accepts iff this is the case.
		It is then clear that~$M_q$ satisfies the conditions, which shows that~$\ccountcompls(q)$
		is in~\#P.
	\end{proof}

\end{toappendix}

\subsection{The complexity on the uniform case}
\label{subsec:countcompls-uniform}
We now investigate the complexity of~$\ucountcompls(q)$ and $\cucountcompls(q)$.
Recall that in the non-uniform case, even counting the completions of a single unary table is 
a \shp-hard problem. This no longer holds in the uniform case, as 
we will 
show that $\ucountcompls(q)$ is in \fp\ for every \sjfbcq\ that is defined 
over a schema consisting exclusively of unary 
relation symbols.

Such a positive result, however, cannot be extended much further. 
In fact, we show next that~$R(x,x)$ and~$R(x,y)$ are hard patterns  
(and, thus, 
we also conclude that the problem of 
counting the completions of a single binary table is 
a \shp-hard problem). Moreover, \shp-hardness holds even if restricted to one of the following settings:   
(a) Na\"ive  tables where nulls are interpreted over a fixed domain, and (b) Codd tables. 

\begin{toappendix}
	\subsection{Proof of Proposition~\ref{prp:Rxx-Rxy-hard-compls}, item (b)}
	\label{apx:Rxx-Rxy-hard-compls}
	In this section we prove point (b) of the following claim (we recall that (a) was proved in Section~\ref{subsec:countcompls-uniform}).
\end{toappendix}

	\begin{propositionrep}
	\label{prp:Rxx-Rxy-hard-compls}
	We have that:
	\begin{enumerate}
	\item[(a)] $\ucountcompls(R(x,x))$ and $\ucountcompls(R(x,y))$ are both \shp-hard, even
		when nulls are interpreted over the same fixed domain~$\{0,1\}$.
	\item[(b)] $\cucountcompls(R(x,x))$ and $\cucountcompls(R(x,y))$ are \shp-hard.
		\end{enumerate} 
	\end{propositionrep}
	\begin{inlineproof}
		We only present the proof of (a) here. The proof of (b) requires more work and can be found in Appendix~\ref{apx:Rxx-Rxy-hard-compls}. 
	We reduce from $\sIS$, the problem of counting the number of independent sets of a graph. 
	
	Let~$G=(V,E)$ be a 
		graph. We will construct an incomplete database~$D$ containing
		a single binary predicate~$R$ such that each completion of~$D$
		satisfies~$R(x,x)$ and 
		the number of completions
		of~$D$ is~$2^{|V|} + \#\IS(G)$, thus establishing hardness for
		the two queries.
		For every node~$u \in V$, we have a
		null~$\bot_u$ with~$\dom(\bot_u) = \{0,1\}$. We then construct the na\"ive 
		table~$D$ as follows:
		\begin{itemize}
			\item for every node~$u \in V$ we add to $D$ the 
			fact~$R(u,\bot_u)$; 
			\item then for every edge~$\{u,v\}
			\in E$, we add the facts~$R(\bot_u,\bot_v)$
		and~$R(\bot_v,\bot_u)$ to $D$; and
			\item last, we add the facts~$R(0,0)$, $R(0,1)$, $R(1,0)$,
		and~$R(\bot,\bot)$, where~$\bot$ is a fresh null. 
		\end{itemize}
		It is clear that every completion of~$D$ satisfies~$R(x,x)$.
		
		Let us now count the number of completions of~$D$.  First, we
		observe that, thanks to the facts of the form $R(u,\bot_u)$, for $u \in V$, for every two
		valuations~$\nu,\nu'$ that do not assign the same value to the
		nulls of the form~$\bot_u$, it is the case that $\nu(D) \neq \nu(D')$.  We then partition the
		completions of~$D$ into those that contain the fact~$R(1,1)$,
		and those that do not contain~$R(1,1)$. 
		Because of the 
		facts of the form $R(u,\bot_u)$, for $u \in V$, and thanks to the fact~$R(\bot,\bot)$ which
		becomes~$R(1,1)$ when we assign~$1$ to~$\bot$, there are
		exactly~$2^{|V|}$ completions of~$D$ that contain~$R(1,1)$.  Moreover, it is
		easy to see that there are~$\#\IS(G)$ valuations~$\nu$ of~$D$
		that assign~$0$ to~$\bot$ and that yield a completion not containing~$R(1,1)$.
		Indeed, one can check that a valuation of~$D$ that
		assigns~$0$ to~$\bot$ yields a completion not containing~$R(1,1)$
		if
		and only if the set~$\{u \in V \mid \nu(\bot_u)=1 \}$ is an
		independent set of~$G$. Therefore, we conclude that~$\sIS \tr \ucountcompls(q)$, where~$q$ can be~$R(x,x)$ or~$R(x,y)$.
	\end{inlineproof} 

\begin{toappendix}
	That is, we deal with the problems~$\cucountcompls(R(x,x))$ and~$\cucountcompls(R(x,y))$.
	We will use the problem of counting the number of induced \emph{pseudoforests} of a graph, as defined next.

\begin{definition}
	\label{def:pseudoforest}
	 A graph~$G$ is a \emph{pseudoforest} if every connected component of~$G$ contains at most one cycle.
	Let~$G=(V,E)$ be a graph. For~$S \subseteq E$, let us denote by~$G[S]$ the graph~$(V',S)$, where~$V'$ is the set of nodes of~$G$ that appear in some edge of~$S$. The problem~$\shpf$ is the problem that takes as input a graph~$G=(V,E)$ and outputs the number of edge sets~$S \subseteq E$ such that~$G[S]$ is a pseudoforest.
\end{definition}

	Using techniques from matroid theory, the authors of~\cite{gimenez2006complexity} have shown that~$\shpf$ is \shp-hard on 
	graphs. We explain in Appendix~\ref{apx:shpf-hard-bipartite} how their proof actually shows hardness of this problem for \emph{bipartite graphs}.

To prove that the reduction that we will present is correct, we will also need the following folklore lemma about pseudoforests.
We recall that an \emph{orientation} of an undirected graph~$G=(V,E)$ is a directed graph that can be obtained from~$G$ by orienting every edge of~$G$. Equivalently, one can see such an orientation as a function~$f:E \to V$ that assigns to every edge in~$G$ a node to which it is incident.
We then have: 

	\begin{lemma}
		\label{lem:orientation}
		A graph~$G$ is a pseudoforest if and only if there exists an orientation of~$G$ such that every node has outdegree at most~$1$.
	\end{lemma}
	\begin{proof}
	Folklore, see, e.g., \cite{kowalik2006approximation,fan2015decomposing,grout2019decomposing}.
	\end{proof}

Using hardness of~$\shpf$ on bipartite graphs, are able show hardness of~$\cucountcompls(R(x,y))$ and~$\cucountcompls(R(x,y))$ for Codd tables, as follows.

	\begin{proof}[Proof of Proposition~\ref{prp:Rxx-Rxy-hard-compls}, item (b)]
	We reduce both problems from $\shpf$ on bipartite graphs. Let~$G=(U \sqcup V,E)$ be a bipartite  graph. We will construct a uniform Codd table~$D$ over binary relation~$R$ such that (1) all the completions of~$D$ satisfy both queries; and (2) the number of completions of~$D$ is equal to~$\shpf(G)$, thus establishing hardness.
	For every~$(t,t') \in (U \cup V)^2 \setminus E$, we add to~$D$ the fact~$R(t,t')$; we call these the \emph{complementary facts}.
	For every~$u \in U$ we add to~$D$ the fact~$R(u,\bot_u)$ and for every~$v \in V$ the fact~$R(\bot_v,v)$. 
	Finally, we add to~$D$ a fact~$R(f,f)$ where~$f$ is a fresh constant.
	The uniform domain of the nulls if~$\dom = U \cup V$. It is clear that~$D$ is a Codd table and that every completion of~$D$ satisfies both queries (thanks to the fact~$R(f,f))$, so (1) holds.
	We now prove that (2) holds.
	First of all, observe that a completion~$\nu(D)$ of~$D$ is uniquely determined by the set of edges~$\{(u,v)\in E \mid R(u,v) \in \nu(D)\}$: this is because~$\nu(D)$ already contains all the complementary facts. For a set~$S\subseteq E$ of edges, let us define~$D_S$ to be the complete database that contains all the complementary facts and all the facts~$R(u,v)$ for~$(u,v) \in S$ (note that $D_S$ is not necessarily a completion of~$D$). We now argue that for every set~$S \subseteq E$, we have that~$D_S$ is a completion of~$D$ if and only if~$G[S]$ is a pseudoforest, which would conclude the proof.
		By lemma~\ref{lem:orientation} we only need to show that~$D_S$ is a completion of~$D$ if and only if~$G[S]$ admits an orientation with maximum outdegee~$1$.
		We show each direction in turn. $(\Rightarrow)$ Assume~$D_S$ is a completion of~$D$, and let~$\nu$ be a valuation witnessing this fact, i.e., such that~$\nu(D)=D_S$. 
		First, observe that we can assume without loss of generality that~($\star$) for every~$e =(u,v) \in S$, we have either~$\nu(\bot_u)=v$ or~$\nu(\bot_v)=u$ but not both. Indeed, if we had both then we could modify~$\nu$ into~$\nu'$ by redefining, say,~$\nu'(\bot_u)$ to be~$u$, and we would still have that~$\nu'(D)=D_S$ (because~$R(u,u)$ is already present in~$D$: it is a complementary fact).
		We now define an orientation~$f_{\nu}:S \to U\cup V$ of~$G[S]$ from~$\nu$ as follows. Let~$e=(u,v) \in S$. Then: if we have~$\nu(\bot_u)=v$ we define~$f_\nu((u,v))$ to be~$v$, i.e., we orient the (undirected) edge~$(u,v)$ from~$u$ to~$v$. Else, if we have~$\nu(\bot_v)=u$ we define~$f_\nu((u,v))$ to be~$u$, i.e., we orient the (undirected) edge~$(u,v)$ from~$v$ to~$u$. Observe that by~($\star$) $f_\nu$ is well defined. It is then easy to check that the maximal outdegree or the directed graph defined by~$f_\nu$ is~$1$: this is because for every~$u \in U$ (resp.,~$v \in V$), there is only one fact in~$D$ of the form~$R(u,\text{null})$ (resp.,~$R(\text{null},v)$), namely, the fact~$R(u,\bot_u)$ (resp., $R(\bot_v,v)$).
		$(\Leftarrow)$ let~$f:S \to U\cup V$ an orientation of~$G[S]$ with maximum outdegree~$1$. Let~$\nu_f$ be the valuation of~$D$ defined from~$f$ as follows:
		for every~$u\in U$ (resp., $v\in V$), if there is an edge~$(u,v)\in S$ such that~$f((u,v)) = v$ (resp., such that $f((u,v)) = u$), then define~$\nu_f(\bot_u)$ to be~$v$ (resp., define~$\nu_f(\bot_v)$ to be~$u$). Observe that there can be at most one such edge because~$f$ has maximum outdegree~$1$, so this is well defined. If there is no such edge, define~$\nu_f(\bot_u)$ to be~$u$ (resp., define $\nu_f(\bot_v)$ to be~$v$). Since all edges in~$S$ are given an orientation by~$f$, it is clear that for every~$(u,v)\in S$ we have~$R(u,v) \in \nu_f(D)$. Moreover, since~$\nu_f(D)$ contains all the complementary facts, we have that~$\nu_f(D)=D_S$, which shows that~$D_S$ is a completion of~$D$ and concludes this proof.
\end{proof}
\end{toappendix}

\begin{toappendix}
	\subsection{Proof for Proposition~\ref{prp:shpf-hard-bipartite}}
	\label{apx:shpf-hard-bipartite}
	In this section we explain how to obtain the following hardness result. 

\begin{proposition}[Implied by {\cite{gimenez2006complexity}}]
\label{prp:shpf-hard-bipartite}
	The problem~$\shpf$ restricted to bipartite graphs is \shp-hard.
\end{proposition}

	This result is proven for 
	(non-necessarily bipartite) graphs
	in~\cite{gimenez2006complexity} using techniques from matroid theory, in particular using the notions of \emph{bicircular matroid} of a graph and of \emph{Tutte polynomial} of a matroid.
	We did not find a way to show that the result holds on bipartite graphs without explaining their proof for general graphs, and we did not find a way to explain the proof for general graphs without introducing these concepts. Therefore, we need
	to define these concepts here.
	We have tried to keep this exposition as brief as possible, but more detailed introductions to matroid theory and to the Tutte polynomial can be found in~\cite{oxley2003matroid,welsh1999tutte}.
	First, we define what is a matroid.
	\begin{definition}
		\label{def:matroid}
		A matroid~$M=(E,\I)$ is a pair where~$E$ is a finite set (called the \emph{ground set}) and~$\I$ is a set of subsets of~$E$ whose elements are called \emph{independent sets} and that satisfies the following properties:
		\begin{description}
			\item[Non emptiness.] $\I \neq \emptyset$;
			\item[Heritage.] For every~$A' \subseteq A \subseteq E$, if~$A \in \I$ then~$A' \in \I$;
			\item[Independent set exchange.] For every~$A,B \in \I$, if~$|A| > |B|$ then there exists~$x\in A \setminus B$ such that~$B\cup \{x\} \in \I$.
		\end{description}
	\end{definition}
	
	In a matroid~$M=(E,\I)$, an independent set~$A \in \I$ is called a
	\emph{basis} if every strict superset~$A \subsetneq A' \subseteq E$ is
	not in~$\I$. Notice that, thanks to the independent set exchange
	property, all bases of~$M$ have the same number of elements. The \emph{rank} of~$M$ is defined as the number of elements in any basis of $M$. Given a matroid~$M=(E,\I)$
	and~$A \subseteq E$, we can define the \emph{submatroid} of~$M$ generated
	by~$A$ to be~$M_A = (A,\I')$, where for~$A' \subseteq A$ we
	have~$A' \in \I'$ iff~$A' \in \I$ (one should check that this is indeed
	a matroid).  The \emph{rank function}~$\rk_M:\{A \mid A \subseteq E\} \to \N$
	of~$M$ is then defined with~$\rk_M(A)$ being the rank of the
	matroid~$M_A$. We will now ommit the subscript in~$\rk_M$ as this will not cause confusion.
	We are ready to define the Tutte polynomial of a matroid.

	\begin{definition}
		\label{def:tutte}
		Let~$M=(E,\I)$ be a matroid. The \emph{Tutte polynomial} of~$M$, denoted~$\T(M;x,y)$, is the two-variables polynomial defined by
		\begin{eqnarray*}
		\label{eq:tutte}
			\T(M;x,y) & = & \sum_{A \subseteq E} (x-1)^{\rk(M)-\rk(A)} (y-1)^{|A|-\rk(A)}
		\end{eqnarray*}
	\end{definition}

	We will use the following observation:
	\begin{observation}
		\label{obs:2-1}
		Let~$M=(E,\I)$ be a matroid. Then~$\T(M;2,1) = |\I|$, i.e., evaluating the Tutte polynomial of a matroid at point~$(2,1)$ simply counts its number of independent sets.
	\end{observation}
	\begin{proof}
		We have~$\T(M;2,1) =\sum_{A \subseteq E} 0^{|A|-\rk(A)}$. We recall the convention that~$0^0 = 1$, and the fact that~$0^k=0$ for~$k>0$.
		Observe then that we always have~$\rk(A) \leq |A|$, and that we have~$\rk(A) = |A|$ if and only if~$A\in \I$, which proves the claim.
	\end{proof}
	
	Next, we define what is called the \emph{bicircular matroid} of a graph~$G=(V,E)$.
	Recall from Section~\ref{apx:Rxx-Rxy-hard-compls} the definition of the induced subgraph~$G[S]$ for~$S \subseteq E$.

	\begin{definition}
		\label{def:bicircular}
		Let~$G=(V,E)$ be a graph and~$\I = \{S\subseteq E \mid G[S] \text{ is a pseudoforest}\}$.
		Then one can check that~$(E,\I)$ is a matroid~\cite{zaslavsky1982bicircular}.
		This matroid is called the \emph{bicircular matroid of~$G$}, and is denoted by~$B(G)$.
	\end{definition}

	Notice then that the problem~$\shpf$ is exactly the same as the problem of computing, given as input a graph~$G$, the quantity~$\T(B(G);2,1)$.
	We now explain the steps used in~\cite{gimenez2006complexity} to prove that computing~$\T(B(G);2,1)$ is \shp-hard for graphs.
	The starting point of our explanation is that computing~$\T(B(G);1,1)$ is \shp-hard.

	\begin{proposition}[{\cite[Corollary 4.3]{gimenez2006complexity}}]
		\label{prp:1-1}
		The problem of computing, given a graph~$G$, the quantity~$\T(B(G);1,1)$ is \shp-hard.
	\end{proposition}
	Second, let us define the following univariate polynomial: for a graph~$G$, let~$P_G(x)$ be
	\[P_G(x) \ = \ \T(B(G);x,1).\]
	Notice that this is indeed a polynomial and that its degree is at most~$|E|$ (the degree is exactly~$|E|$ iff~$G$ is itself a pseudoforest).
	If we could compute efficiently the coefficients of~$P_G$, then we could in particular compute the value~$P_G(1) = \T(B(G);1,1)$, which is \shp-hard by the previous proposition.
	We recall that to compute the coefficients of a polynomial of degree~$n$, it is enough to know its value on~$n+1$ distinct points; in fact, given these values in $n+1$ distinct points, it is possible to efficiently compute the coefficients of the polynomial by using standard  interpolation techniques (for example, by using Lagrange polynomials).

	We need one last definition.
	\begin{definition}
		\label{def:stretch}
		Let~$G$ be a graph. For~$k\in \N$, let~$\s_k(G)$ be the graph obtained from~$G$ by replacing each edge of~$G$ by a path of lenght~$k$; this graph is called the \emph{$k$-stretch} of~$G$.
	\end{definition}

	Then, using a result attributed to Brylawski (see~\cite{jaeger1990computational}), the authors of~\cite{gimenez2006complexity} obtain that, “up to a trivial factor”, we have
	\[\T(B(\s_k(G));2,1) \ \simeq \ \T(B(G);2^k,1). \]
	A careful inspection of~\cite{jaeger1990computational} reveals\footnote{To be precise, we use Equations (7.1) and (7.2) of~\cite{jaeger1990computational} with~$x=1,y=0$, and Equation (2.2) with~$x=2,y=1$.} that, in fact, we have
	\[\T(B(\s_k(G));2,1) \ = \ (2^k -1)^{|E|-\rk_{B(G)}(E)} \times \T(B(G);2^k,1). \]
	Notice that~$\rk_{B(G)}(E)$ is the size (number of edges) of a pseudoforest of~$G$ that is maximal by inclusion of edges, which we can compute in polynomial time.\footnote{This is because, since~$B(G)$ is a matroid, any two such pseudoforests have the same number of edges. We can then simply start from the empty subgraph and iteratively add edges until it is not possible to add an edge such that the resulting graph is a pseudoforest. This also relies on the fact that we can check in polynomial time whether a graph is a pseudoforest.}
	
	With this, the authors of~\cite{gimenez2006complexity} can conclude the proof that computing~$\T(B(G);2,1)$ is hard for (non-necessarily bipartite) graphs, i.e., that~$\shpf$ is \shp-hard.
	Indeed, given as input~$G=(V,E)$, we can construct in polynomial time the graphs~$\s_k(G)$ for~$|E|+1$ distinct values of~$k$,
	then use oracle calls to obtain the numbers~$\T(B(\s_k(G));2,1)$, which gives us the value of~$P_G$ on~$|E|+1$ distinct points.
	With that we can recover the coefficients of~$P_G$ and compute~$P_G(1) = \T(B(G);1,1)$ as argued above, thus proving hardness for general graphs.
	To obtain hardness for bipartite graphs, it is enough to observe that when~$k$ is even then the~$k$-stretch of~$G$ is bipartite (even if $G$ is not bipartite).
	Hence, to obtain a proof of Proposition~\ref{prp:shpf-hard-bipartite} for bipartite graphs, we can simply change that proof and specify that we make~$|E|+1$ calls to the oracle~$\T(B(\s_k(G));2,1)$ for~$|E|+1$ disctinct \emph{even} values of~$k$.
\end{toappendix}

	The patterns in Proposition 
	\ref{prp:Rxx-Rxy-hard-compls}
	suffice to characterize the complexity of $\ucountcompls(q)$ and $\cucountcompls(q)$.

	\begin{toappendix}
		\subsection{Proof of Theorem~\ref{thm:ucc-naive-fp}}
		\label{apx:ucc-naive-fp}
	\end{toappendix}

	\begin{theoremrep}[Dichotomy]
	\label{thm:ucc-naive-fp}
	Let $q$ be an \sjfbcq. 
	 If $R(x,x)$ or $R(x,y)$ is a pattern of $q$, 
	 then~$\ucountcompls(q)$ and~$\cucountcompls(q)$ are \shp-hard. Otherwise, these problems are in~\fp.
	\end{theoremrep}

\begin{toappendix}
	We only need to prove the tractability part of that claim, and this only for uniform incomplete databases.
	Remember from Section~\ref{apx:countvals-naive-uniform} that what we call a \emph{conjunction of basic singleton \sjfbcq} is an \sjfbcq\ of the form~$C_1(x_1) \land \ldots \land C_m(x_m)$,
	where each~$C_i(x_i)$ is a conjunction of unary atoms over the same variable~$x_i$.
	Since~$q$ does not contain the pattern~$R(x,x)$ nor the pattern~$R(x,y)$,~$q$ is in fact a conjunction of basic singleton \sjfbcq.
	The main difficulty is to decompose the computation in such a way that we do not count the same completion twice. Moreover, the fact that the database is naïve and not Codd, and the fact that constants can appear everywhere, complicate a lot the description of the algorithm.
For these reasons, and to give the intuition of the general proof, we first present a few warm-up examples of increasing difficulty. We strongly advise the reader to read these before reading the proof.
	In what follows we will always denote by~$\dom$ the uniform domain of the nulls, and~$d \geq 1$ its size.

	\subsubsection{Warm-up example 1: $\ucountcompls(R(x))$ without constants}
	The database~$D$ consists of the facts~$\{R(\bot_1),\ldots,R(\bot_{n_R})\}$. If~$n_R=0$ then the result is~$1$ (the empty completion), so let us assume~$n_R \geq 1$.
	Notice then that for every subset~$I_R \subseteq \dom$, the database~$\{R(a) \mid a \in I_R\}$ is a completion of~$D$ if and only if we have~$1 \leq |I_R| \leq n_R$.
	But then the answer is simply~$\sum_{1 \leq i_R \leq n_R} \binom{d}{i_R}$, which we can compute in polynomial time.\footnote{With the convention that~$\binom{a}{b} = 0$ when~$b > a$.} Note that the expression $\sum_{0 \leq i_R \leq n_R} \binom{d}{i}$ would give the right answer only when~$n_R=0$ because there has to be at least one tuple in the completion if~$n_R > 1$ so the sum should start at~$1$.
	We can also compute the result using the following expression, which works for all~$n \in \mathbb{N}$:
	\begin{equation}\sum_{0 \leq i_R \leq d} \binom{d}{i_R} \times \checkp(i_R) \end{equation}
		where~$\checkp(i_R) \in \{0,1\}$ is defined as~$\checkp(i_R)\defeq \begin{cases}
		\text{ if } i_R > n_R \text{ then }0\\
		\text{ if } i_R = 0 \text{ and } n_R \geq 1 \text{ then } 0 \\
			    \text{ otherwise } 1 
	\end{cases}$.

	\subsubsection{Warm-up example 2: $\ucountcompls(R(x))$ with constants}

	Let~$D$ be the database~$\{R(a_1),\ldots,R(a_{c_R}),R(\bot_1),\ldots,R(\bot_{n_R})\}$, and let~$C_R = \{a_1,\ldots,a_{c_R}\}$ be the set of constants.
		First, observe that we can assume wlog that~$C_R \subseteq \dom$; indeed, letting~$D'$ be the incomplete database obtained from~$D$ by removing all the facts~$R(a)$ for which~$a \in C_R \setminus \dom$, then~$D$ and~$D'$ actually have the same number of completions.
		Moreover, we can assume that~$c_R \geq 1$, otherwise we are in the previous example.
		Then, observe that for every subset~$I_R \subseteq \dom \setminus C_R$, the database~$\{R(a) \mid a \in I_R \text{ or } a \in C_R\}$ is a completion of~$D$ if and only if we~$0 \leq |I_R| \leq n_R$.
		But then the answer is simply~$\sum_{0\leq i_R \leq n_R} \binom{d-c_R}{i_R}$.
		Note that this expression would not give the right answer in case we had~$c_R=0,n_R\geq 1$ because the nulls cannot be “absorbed” by~$C_R$, so in that case the sum should start at~$1$.
		We can also compute the result using the following expression, which works for all~$c_R,n_R \in \mathbb{N}$:
	\begin{equation}\sum_{0\leq i_R \leq d} \binom{d-c_R}{i_R} \times \checkp(i_R)\end{equation}
	where~$\checkp(i_R) \defeq \begin{cases}
		\text{ if } i_R > n_R\text{ then } 0\\
		\text{ if } i_R = 0 \text{ and } c_R = 0 \text{ and }n_R \geq 1 \text{ then } 0 \\
			    \text{ otherwise } 1 
	\end{cases}$.

	\subsubsection{Warm-up example 3: $\ucountcompls(R(x)\land S(y))$ without constants}
	Let~$D$ be an incomplete naïve table over~$R,S$ that do not have constants. Let~$n_{RS}$ be the number of nulls that occur in both~$R$ and~$S$, $n_R$ the number of nulls that occur only in~$R$ and~$n_S$ the number of nulls that occur only in~$S$. We further assume that~$n_{RS} \geq 1$, otherwise we can simply compute independently the number of completions of the~$R$ and~$S$ tables as in warm-up example 1 and multiply the two numbers. We claim the following:
	\begin{claim}
		Let~$D'$ be a complete database over~$R,S$ with constants in~$\dom$, and let~$I_R \defeq \{a \in \dom \mid R(a) \in D', S(a) \notin D'\}$, $I_S \defeq \{a \in \dom \mid S(a) \in D', R(a) \notin D'\}$, and $I_{RS} \defeq \{a \in \dom \mid \{R(a),S(a)\} \subseteq D'\}$. Then~$D'$ is a completion of~$D$ if and only if
		we have~$|I_R| \leq n_R$, $|I_S| \leq n_S$, and~$1 \leq |I_{RS}| \leq \min(n_{RS} + n_R - |I_R|, n_{RS} + n_S - |I_S|)$.
	\end{claim}
	\begin{proof}
		The only if part is easy to check. Suppose then that these conditions hold on~$D'$. We can assign the first~$|I_R|$ nulls that are only in~$R$ so that they span~$I_R$, and assign the first~$|I_S|$ nulls that are only in~$S$ so that they span~$I_S$. If~$|I_{RS}| \leq n_{RS}$ then we can assign the nulls that are in both~$R$ and~$S$ so that they span~$I_{RS}$, and then we can assign the remaining nulls (that appear only in~$R$ or only in~$S$) to a value that has already been assigned, and we indeed obtain~$D'$ as a completion. If~$|I_{RS}| > n_{RS}$, we assign~$n_{RS}$ nulls so that they span~$n_{RS}$ elements of~$I_{RS}$ (say they span a set~$I'_{RS} \subseteq I_{RS}$), then, because we have~$|I_{RS}| \leq \min(n_{RS} + n_R - |I_R|, n_{RS} + n_S - |I_S|)$, we can use the remaining nulls that occur only in~$R$ or in~$S$ to span~$I_{RS} \setminus I'_{RS}$. Again we obtain~$D'$ as a completion.
	\end{proof}

	But then this implies that the result can be expressed as \[\sum_{\substack{I_R \subseteq \dom\\ |I_R|\leq n_R}} \sum_{\substack{I_S \subseteq \dom \setminus I_R\\|I_S|\leq n_S}} \sum_{\substack{I_{RS} \subseteq \dom \setminus (I_R \cup I_S)\\ 1 \leq |I_{RS}| \leq \min(n_{RS} + n_R - |I_R|, n_{RS} + n_S - |I_S|)}} 1.\]
	We cannot compute this expression as-is because we are summing over subsets of~$\dom$. However, since each subsum depends only on the sizes of the sets introduced before it, we can simplify this expression to
	\[\sum_{0 \leq i_R \leq n_R} \binom{d}{i_R} \sum_{0 \leq i_S \leq n_S} \binom{d-i_R}{i_S} \sum_{1 \leq i_{RS} \leq  \min(n_{RS} + n_R - i_R, n_{RS} + n_S - i_S)} \binom{d-i_R-i_S}{i_{RS}}\]
	which we can compute in polynomial time.
	The result can also be expressed as the following expression which work for all~$n_R,n_S,n_{RS} \in \mathbb{N}$:
	\begin{equation}
		\label{eq:RxSy}
	\sum_{0 \leq i_R,i_S,i_{RS} \leq d} \binom{d}{i_R} \binom{d-i_R}{i_S} \binom{d-i_R-i_S}{i_{RS}} \times \checkp(i_R,i_S,i_{RS}) \end{equation}
	where~$\checkp(i_R,i_S,i_{RS}) \defeq \begin{cases}
			    \text{ if } i_R>n_R \text{ then } 0\\
			    \text{ if } i_S>n_S \text{ then } 0\\
		\text{ if } n_{RS} \geq 1 \text{ and } i_{RS}=0 \text{ then } 0\\
			    \text{ if } i_R=0, n_R \geq 1 \text{ and } n_{RS} = 0 \text{ then } 0\\
			    \text{ if } i_S=0, n_S \geq 1 \text{ and } n_{RS} = 0 \text{ then } 0\\
			    \text{ if } i_{RS} >  \min(n_{RS} + n_R - i_R, n_{RS} + n_S - i_S) \text{ then } 0\\
			    \text{ otherwise } 1
	\end{cases}$.

	\subsubsection{Warm-up example 4: $\ucountcompls(R(x)\land S(x))$ without constants}
	We use again Equation~\ref{eq:RxSy}, but to ensure that the query is satisfied we add that~$\checkp(i_R,i_S,i_{RS})$ becomes~$0$ when~$i_{RS} = 0$.

	\subsubsection{Warm-up example 5: $\ucountcompls(R(x)\land S(y))$ with constants}
	This example is much more involved, and we will mimick all the steps of the general proof.
	Let~$C_{RS},C_R,C_S$ (resp, $N_{RS},N_R,N_S$) be the sets of constants (resp., nulls) that occur respectively: in~$R$ and in~$S$, only in~$R$, only in~$S$, and denote~$c_{RS},c_R,c_S$ (resp., $n_{RS},n_R,n_S$) their sizes.
	For the same reason as in warm-up example 2, we can assume wlog that~$C \defeq C_{RS} \cup C_R \cup C_S \subseteq \dom$.
	Let~$c = c_{RS} + c_R + c_S$.
	We claim the following:
	\begin{claim}
		\label{claim:expl-disjoint}
		For a triplet~$(I_R,I_S,I_{RS})$ of subsets of~$\dom$ satisfying the conditions~($\star$) $I_R \subseteq \dom \setminus C$, $I_S \subseteq \dom \setminus (C \cup I_R)$, and $I_{RS} \subseteq \dom \setminus (C_{RS} \cup I_R \cup I_S)$,
		let us define~$P(I_R,I_S,I_{RS})$ to be the complete database consisting of the following facts:
		\begin{enumerate}
			\item $R(a)$ and~$S(a)$ for~$a \in C_{RS} \cup I_{RS}$;
			\item $R(a)$ for~$a \in I_R \cup (C_R \setminus I_{RS})$;
			\item $S(a)$ for~$a \in I_S \cup (C_S \setminus I_{RS})$;
		\end{enumerate}
		Then, for any two such triplets of sets~$(I_R,I_S,I_{RS})$ and~$(I'_R,I'_S,I'_{RS})$ that are different, the complete databases~$P(I_R,I_S,I_{RS})$ and~$P(I'_R,I'_S,I'_{RS})$ are distinct.
	\end{claim}
	\begin{proof}
		To help the reader, we have drawn in Figure~\ref{fig:drawing} how the sets can intersect.
		If we have~$I_{RS} \neq I'_{RS}$ with~$a \in I_{RS}$ and~$a
		\notin I'_{RS}$, then one can check
		that~$P(I_R,I_S,I_{RS})$ contains both
		facts~$R(a)$ and~$S(a)$,
		while~$P(I'_R,I'_S,I'_{RS})$ does not. So let us
		assume now that~$I_{RS} = I'_{RS}$. If we have~$I_R \neq I'_R$
		with~$a \in I_R$ and~$a\notin I'_R$ then one can check
		that~$P(I_R,I_S,I_{RS})$ contains the
		fact~$R(a)$ while~$P(I'_R,I'_S,I'_{RS})$ does
		not. Hence let us assume that~$I_R = I'_R$. Using the same
		reasoning we obtain that~$I_S = I'_S$, thus completing the
		proof.
	\end{proof}

	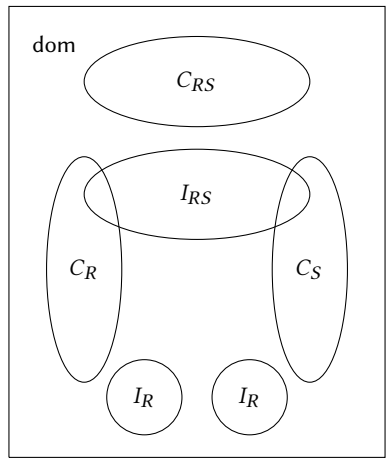
\begin{figure}
		\centering
	\begin{tikzpicture}
		\draw (0,0) -- (5,0) -- (5,6) -- (0,6) -- (0,0);
		\draw \boundellipse{2.5,5}{1.5}{.6};
		\draw \boundellipse{2.5,3.5}{1.5}{.6};
		\draw \boundellipse{1,2.5}{.5}{1.5};
		\draw \boundellipse{4,2.5}{.5}{1.5};
		\draw \boundellipse{1.8,.8}{.5}{.5};
		\draw \boundellipse{3.2,.8}{.5}{.5};
		\node at (.6,5.5) {$\dom$};
		\node at (2.5,5) {$C_{RS}$};
		\node at (2.5,3.5) {$I_{RS}$};
		\node at (1,2.5) {$C_R$};
		\node at (4,2.5) {$C_S$};
		\node at (1.8,.8) {$I_R$};
		\node at (3.2,.8) {$I_R$};
	\end{tikzpicture}
		\caption{How the sets $\dom, I_R,I_S,I_{RS},C_{RS},C_R$ and $C_S$ from Claim~\ref{claim:expl-disjoint} are allowed to intersect when they satisfy~($\star$). The sets themselves and the intersections can be empty.}
		\label{fig:drawing}

	\end{figure}

	Our next step is to show that every completion of~$D$ is of the form~$P(I_R,I_S,I_{RS})$ for some triplet~$(I_R,I_S,I_{RS})$ satisfying~($\star$):
	\begin{claim}
		\label{claim:expl-have-all}
		For every completion~$D'$ of~$D$, there exist a triplet~$(I_R,I_S,I_{RS})$ satisfying~($\star$) such that~$D' = P(I_R,I_S,I_{RS})$.
	\end{claim}
	\begin{proof}
		We define:
		\begin{itemize}
			\item $I_R \defeq D'(R) \setminus (C_R \cup D'(S))$; where we see~$D'(R)$ as the set of constants occurring in relation~$R$ of~$D'$.
			\item $I_R \defeq D'(S) \setminus (C_S \cup D'(R))$;
			\item $I_{RS} \defeq (D'(R) \cap D'(S)) \setminus C_{RS}$.
		\end{itemize}
		Then one can easily check that we have~$D' = P(I_R,I_S,I_{RS})$.
	\end{proof}

	By combining these two claims, we have that the result that we wish to compute is equal to
	\[\sum_{I_R \subseteq \dom \setminus C} \sum_{I_S \subseteq \dom \setminus (C \cup I_R)} \sum_{I_{RS} \subseteq \dom \setminus (C_{RS} \cup I_R \cup I_S)} \checkp(I_R,I_S,I_{RS}) \]
	where~$\checkp(I_R,I_S,I_{RS}) \defeq \begin{cases}
		1 \text{ if } P(I_R,I_S,I_{RS}) \text{ is a possible completion of } D \\
		0 \text{ otherwise}
	\end{cases}$.

	Next, we show that the value of~$\checkp(I_R,I_S,I_{RS})$ can be computed in polynomial time and actually only depends on the sizes of these sets. 
	In order to show this, we will use the following:
	\begin{claim}
		\label{claim:expl-conditions}
		We have~$\checkp(I_R,I_S,I_{RS}) = 1$ if and only if the following conditions hold:
		\begin{enumerate}
			\item if~$n_R \geq 1$ and $|C_R \cup C_{RS} \cup I_{RS}| = 0$, then we have~$|I_R| \neq 0$. Intuitively, this means that the value of a null in~$N_R$ cannot be absorbed by~$C_R \cup C_{RS} \cup I_{RS}$.
			\item if~$n_S \geq 1$ and $|C_S \cup C_{RS} \cup I_{RS}| = 0$, then we have~$|I_S| \neq 0$.
			\item if~$n_{RS} \geq 1$ and $|C_{RS}| = 0$, then we have~$|I_{RS}| \neq 0$.
			\item the following system of equations, whose variables are natural numbers between~$0$ and~$d$, has a solution:
				\begin{align*}
					z_{N_R}^{\{N_R\}} + z_{N_R}^{\{N_R,C_S\}} +  z_{N_R}^{\{N_R,N_S\}} &\leq n_R  \\
					z_{N_S}^{\{N_S\}} + z_{N_S}^{\{N_S,C_R\}} +  z_{N_S}^{\{N_S,N_R\}} &\leq n_R  \\
					z_{C_R}^{\{C_R,N_S\}} &\leq c_R  \\
					z_{C_S}^{\{C_S,N_R\}} &\leq c_S  \\
					z_{N_R}^{\{N_R\}} &\geq |I_R| \\
					z_{N_S}^{\{N_S\}} &\geq |I_S| \\
					n_{RS} + \min(z_{N_R}^{\{N_R,C_S\}}, z_{C_S}^{\{C_S,N_R\}}) + 
					\min(z_{N_R}^{\{N_R,N_S\}}, z_{N_S}^{\{N_S,N_R\}}) + 
					\min(z_{N_S}^{\{N_S,C_R\}}, z_{C_R}^{\{C_R,N_S\}})
					&\geq |I_{RS}|
				\end{align*}
		\end{enumerate}
	\end{claim}
		\begin{proof}
			We prove the claim informally by explaining the main ideas, because a formal proof would be too long and not that interesting.
			Conditions~(1-3) are easily checked to be necessary. We now explain why condition~(4) is also necessary.
			Suppose that~$P(I_R,I_S,I_{RS})$ is a completion of~$D$.
			Observe that~($\dagger$) to obtain the constants in~$I_{RS}$, we had to use some or all of the following:
		\begin{itemize}
			\item the nulls in~$N_{RS}$; or
			\item the nulls in~$N_R$ together with those in~$N_S$; or
			\item the nulls in~$N_R$ together with the constants in~$C_S$; or
			\item the nulls in~$N_S$ together with the constants in~$C_R$.
		\end{itemize}

			But then, to obtain~$P(I_R,I_S,I_{RS})$ as a completion, we must have used three disjoint (possibly empty) sets $Z_{N_R}^{\{N_R\}}, Z_{N_R}^{\{N_R,C_S\}},  Z_{N_R}^{\{N_R,N_S\}}$ of the nulls in~$N_R$ of sizes~$0\leq z_{N_R}^{\{N_R\}}, z_{N_R}^{\{N_R,C_S\}},  z_{N_R}^{\{N_R,N_S\}} \leq d$, we have done the same for the nulls in~$N_S$ and we also used a subset of the constants of~$C_R$ (and~$C_S$) in such a way that, according to~($\dagger$):
		\begin{itemize}
			\item the nulls in~$Z_{N_R}^{\{N_R\}}$ have been used to obtain the set~$I_R$ (which, we recall, is the set of constants that occur only in~$R$ and that are not in~$C_R$). Note that only the nulls in~$N_R$ could have been used to obtain constants in~$I_R$. This is what the fifth equation expresses.
			\item the nulls in~$Z_{N_R}^{\{N_R,C_S\}}$ have values in~$Z_{C_S}^{\{C_S,N_R\}}$, which gives us constants in~$I_{RS}$. Observe that at maximum we could obtain $\min(z_{N_R}^{\{N_R,C_S\}}, z_{C_S}^{\{C_S,N_R\}})$ constants in this manner.
			\item the nulls in~$Z_{N_R}^{\{N_R,N_S\}}$ and those in~$Z_{N_R}^{\{N_R,N_S\}}$ have common values, which gives us constants in~$I_{RS}$.
				Again, observe that we can get at most~$\min(z_{N_R}^{\{N_R,N_S\}}, z_{N_S}^{\{N_S,N_R\}})$ constants using these.
		\end{itemize}
			The first~$4$ equations express the partitioning process, and
			the last equation then expresses that by combining all these constants we indeed obtained the whole set~$I_{RS}$.

			We now explain why conditions~(2-4) are sufficient. If~$|I_R|,|I_S|$ and~$I_{RS}$ are all~$\geq 1$ then condition~(4) is sufficient, because we can use the nulls and constants as explained above, and we have enough of them to obtain the sets~$I_R,I_S,I_{RS}$. We explain what happens when~$I_R = \emptyset$ for instance. In that case, condition~(1) ensures us that we have either~$n_R=0$ or $C_R \cup C_{RS} \cup I_{RS} \neq \emptyset$. If we have~$n_R=0$ then it is fine, since the only nulls that could be used to fill~$I_R$ are those in~$N_R$. If we have~$n_R \geq 1$ and~$C_R \cup C_{RS} \cup I_{RS} \neq \emptyset$ then we can use these to absorb the values of the nulls in~$N_R$, and we are fine (i.e., we will be able to obtain~$I_R = \emptyset$). We leave it to the reader to complete the small gaps in this proof.

		\end{proof}

	Using this, we have that the value of~$\checkp(I_R,I_S,I_{RS})$ only depends on the sizes of~$I_R,I_S,I_{RS}$, and moreover can be computed in polynomial time
	\begin{claim}
		The value of~$\checkp(I_R,I_S,I_{RS})$ only depends on~$|I_R|,|I_S|,|I_{RS}|,n_R,n_S,n_{RS},c_R,c_S,c_{RS}$, and can be computed in \ptime.
	\end{claim}
	\begin{proof}
		The fact that this value only depends on the sizes of these sets is simply by inspection of the conditions in Claim~\ref{claim:expl-conditions}. Conditions~(1-3) can obviously be checked in \ptime. The fact that condition~(4) can be checked in \ptime\ is because we can test all possible assignments between~$0$ and~$d$ for all these variables and see if there is one assignment that satisfies the equations (note that the number of variables is fixed).
	\end{proof}

	But then, we can express the result as follows
	\[\sum_{0 \leq i_R,i_S,i_{RS} \leq d} \binom{d-c}{i_R} \binom{d-c-i_R}{i_S} \binom{d-c_{RS} - i_R - i_S}{i_{RS}}  \times \checkp(i_R,i_S,i_{RS}) \]
	and we can evaluate this expression as-is in FP because computing~$\checkp(I_R,I_S,I_{RS}) \in \{0,1\}$ is in \ptime\ by the last claim. This concludes this example.

	\subsubsection{Proof of Theorem~\ref{thm:ucc-naive-fp}}
	We now present the general proof.
	Let~$\sigma = \{R_1,\ldots,R_l\}$ be the set of relation symbols, and~$D$ be an incomplete database over these relations.
	For every~$\mathbf{s} \subseteq \sigma$, $\mathbf{s} \neq \emptyset$, let:
	 \begin{itemize}
		\item $C_\mathbf{s}$ be the set of constants that occur in all relations of~$\mathbf{s}$ and in none of the others;~$c_\mathbf{s}$ be its size;
		\item~$N_\mathbf{s}$ be the set of nulls that occur in all relations of~$\mathbf{s}$ and in none of the others;~$n_\mathbf{s}$ be its size.
	 \end{itemize}
	We also define~$c$ as $\sum_{\emptyset \neq \mathbf{s} \subseteq \sigma} c_\mathbf{s}$.
	We can assume wlog that~$C_\mathbf{s} \subseteq \dom$ for all~$\emptyset \neq \mathbf{s} \subseteq \sigma$, otherwise we can simply remove from~$D$ the corresponding facts. 
	Let~$L \defeq 2^l-1$, and let~$\mathbf{s}_1,\ldots,\mathbf{s}_L$ be an arbitrary linear order of~$\{\mathbf{s} \subseteq \sigma \mid \mathbf{s} \neq \emptyset\}$ (for instance, by non-decreasing size).
	We will follow the same steps as in the last example.
	The following lemma is the generalization of Claim~\ref{claim:expl-disjoint}, and explains how we can guide the computation so that we do not count the same completion twice:

	\begin{lemma}
		\label{lem:disjoint}
		For a tuple~$(I_{\mathbf{s}_1},\ldots,I_{\mathbf{s}_L})$ of subsets of~$\dom$ satisfying~($\star$)
		\[I_{\mathbf{s}} \subseteq (\dom \setminus (C \cup \bigcup_{\substack{\emptyset \neq \mathbf{s'}\subseteq \sigma \\ \mathbf{s'} \neq \mathbf{s}}} I_{\mathbf{s'}})) \cup \bigcup_{\emptyset \neq \mathbf{s'}\subsetneq \mathbf{s}} C_\mathbf{s'}\]
		for every~$\mathbf{s} \in (\mathbf{s}_1,\ldots,\mathbf{s}_L)$ (in other words, all the sets~$I_\mathbf{s}$ are mutually disjoint subsets of~$\dom$, and a set~$I_\mathbf{s}$ can only contain a constant~$b\in C$ if~$b$ is in one of the sets~$C_\mathbf{s'}$ for which~$\mathbf{s'}$ is striclty included in~$\mathbf{s}$),
		let us define~$P(I_{\mathbf{s}_1},\ldots,I_{\mathbf{s}_L})$ to be the complete database consisting of the following facts, for every~$\emptyset \neq \mathbf{s} \subseteq \sigma$:
		\begin{itemize}
			\item $R(a)$ for every~$R \in \mathbf{s}$ and~$a \in I_\mathbf{s}$ and~$a \in C_\mathbf{s} \setminus \bigcup_{\mathbf{s} \supsetneq \mathbf{s'}} I_\mathbf{s'}$
		\end{itemize}
		Then, for every two such tuples~$(I_{\mathbf{s}_1},\ldots,I_{\mathbf{s}_L})$ and~$(I'_{\mathbf{s}_1},\ldots,I'_{\mathbf{s}_L})$ satisfying ($\star$) and that are distinct, we have that~$P(I_{\mathbf{s}_1},\ldots,I_{\mathbf{s}_L}) \neq P(I'_{\mathbf{s}_1},\ldots,I'_{\mathbf{s}_L})$.
	\end{lemma}
	\begin{proof}
		Let us write~$P = P(I_{\mathbf{s}_1},\ldots,I_{\mathbf{s}_L})$
		and~$P' = P(I'_{\mathbf{s}_1},\ldots,I'_{\mathbf{s}_L})$.
		Assume that~$P=P'$, and let us show
		that~$(I_{\mathbf{s}_1},\ldots,I_{\mathbf{s}_L}) =
		(I'_{\mathbf{s}_1},\ldots,I'_{\mathbf{s}_L})$.  Assume by way
		of contradiction that for some~$\emptyset \neq \mathbf{s}
		\subseteq \sigma$ we have~$I_\mathbf{s} \neq I'_\mathbf{s}$. Then
		(wlog) there exists~$a \in I_\mathbf{s} \setminus
		I'_\mathbf{s}$. By the definition of~$P$, we have that~$P$
		contains all the facts~$R(a)$ for~$R \in \mathbf{s}$. Let us
		show that~$P$ does not contain any fact~$R(a)$
		for~$R\notin\mathbf{s}$. Otherwise, assume that~$P$
		contains~$R(a)$ with~$R \notin \mathbf{s}$. Then there
		exists~$\mathbf{s'} \subseteq \sigma$ such that~$R \in
		\mathbf{s'}$ and such that~$a \in I_{\mathbf{s'}} \cup (
		C_\mathbf{s'} \setminus \bigcup_{\mathbf{s'} \supsetneq
		\mathbf{s''}} I_\mathbf{s''})$.
		Since~$\mathbf{s}$ does not contain~$R$ while~$\mathbf{s'}$
		does, we have~$\mathbf{s'} \not\subseteq \mathbf{s}$. But then
		by~($\star$) we have that~$I_\mathbf{s}$ and~$I_{\mathbf{s'}}
		\cup  C_\mathbf{s'}$ are disjoint, which is a contradiction
		because~$a$ is supposed to be in both~$I_\mathbf{s}$
		and~$I_{\mathbf{s'}} \cup ( C_\mathbf{s'} \setminus
		\bigcup_{\mathbf{s'} \supsetneq \mathbf{s''}}
		I_\mathbf{s''})$.  Therefore, it is indeed the case that~$P$
		does not contain any fact~$R(a)$ for~$R\notin\mathbf{s}$.  Now,
		if~$P'$ contains a fact~$R(a)$ for some~$R \notin \sigma$ then
		we are done since this would imply~$P \neq P'$, a
		contradiction. Hence we can assume that~$P'$ does not contain any
		fact~$R(a)$ for~$R \notin \sigma$.  We will now prove
		that~$P'$ does not contain all the facts~$R(a)$ for~$R \in
		\sigma$, thus establishing a contradiction (because~$P$ does, so we would have~$P
		\neq P'$) and concluding this proof.  Assume by contradiction
		that~$P'$ contains all the facts~$R(a)$ for~$R \in \mathbf{s}$.
		First of all, observe that we have~$a \notin C_\mathbf{s}$
		because by~($\star$) we have that~$I_\mathbf{s}$
		and~$C_\mathbf{s}$ are disjoint, and we know that~$a \in
		I_\mathbf{s}$. Hence, the only way in which~$P'$ could contain
		all the facts~$R(a)$ for~$R \in \mathbf{s}$ is if there
		exist~$\mathbf{s'}_1,\ldots,\mathbf{s'}_k$ with~$k \geq 1$
		and~$\mathbf{s'}_j \subsetneq \mathbf{s}$ for~$1 \leq j\leq k$
		such that~$\bigcup_{1 \leq j \leq k} \mathbf{s'}_j =
		\mathbf{s}$ and such that for every~$1 \leq j \leq k$ we have
		that~(i) $a \in I_{\mathbf{s'}_j} \cup ( C_{\mathbf{s'}_j}
		\setminus \bigcup_{\mathbf{s'}_j \supsetneq \mathbf{s''}} I_\mathbf{s''})$.  Observe that there must
		exist~$1 \leq j_1,j_2 \leq k$ such that~$\mathbf{s'}_{j_1}$
		and~$\mathbf{s'}_{j_2}$ are incomparable by inclusion
		(otherwise, since all~$\mathbf{s}_j$ are strictly included
		in~$\mathbf{s}$, their union could not be equal
		to~$\mathbf{s}$). Also observe that by~($\star$) we have that
		the sets~$I_{\mathbf{s'}_{j_1}} \cup C_{\mathbf{s'}_{j_1}}$ and
		$I_{\mathbf{s'}_{j_2}} \cup C_{\mathbf{s'}_{j_2}}$ must be
		disjoint. But then~(i) applied to~$j_1$ and~$j_2$ gives a
		contradiction (namely, these two sets are not disjoint since
		they both contain~$a$).  This finishes the proof.
	\end{proof}

	This next Lemma generalizes Claim~\ref{claim:expl-have-all} and tells us that by summing over all such tuples~$(I_{\mathbf{s}_1},\ldots,I_{\mathbf{s}_L})$ we cannot miss a completion of~$D$:
	\begin{lemma}
		\label{lem:have-all}
		Let~$D'$ be a completion of~$D$. Then there exists a tuple~$(I_{\mathbf{s}_1},\ldots,I_{\mathbf{s}_L})$ of subsets of~$\dom$ satisfying~($\star$) such that~$D' = P(I_{\mathbf{s}_1},\ldots,I_{\mathbf{s}_L})$.
	\end{lemma}
	\begin{proof}
		For~$\emptyset \neq \mathbf{s} \subseteq \sigma$, let us define~$D_\mathbf{s}$ to be the set of constants that occur in all relation of~$\mathbf{s}$ and in none of the others. Define the set~$I_\mathbf{s}$ for~$\emptyset \neq \mathbf{s} \subseteq \sigma$ as follows: $I_\mathbf{s} \defeq D_\mathbf{s} \setminus C_\mathbf{s}$.
		It is then routine to check that~$(I_{\mathbf{s}_1},\ldots,I_{\mathbf{s}_L})$ satisfies~($\star$) and is such that~$D' = P(I_{\mathbf{s}_1},\ldots,I_{\mathbf{s}_L})$.
	\end{proof}

	Lemma~\ref{lem:disjoint} and~\ref{lem:have-all} allows us to express the result as

	\begin{equation}
		\sum_{I_{\mathbf{s}_1} \subseteq \dom \setminus C} \ldots 
		\sum_{I_{\mathbf{s}_j} \subseteq (\dom \setminus (C \cup \bigcup_{1 \leq k < j} I_{\mathbf{s}_k})) \cup \bigcup_{\emptyset \neq \mathbf{s'}\subsetneq \mathbf{s}} C_\mathbf{s'}} \ldots
		\sum_{I_{\mathbf{s}_L} \subseteq \dom \setminus (C_{\mathbf{s}_L} \cup \bigcup_{1 \leq k < L} I_{\mathbf{s}_k}) }
		\checkp(I_{\mathbf{s}_1},\ldots,I_{\mathbf{s}_L})
	\end{equation}

	where~$\checkp(I_{\mathbf{s}_1},\ldots,I_{\mathbf{s}_L}) \in \{0,1\}$ is defined by~$\checkp(I_{\mathbf{s}_1},\ldots,I_{\mathbf{s}_L})\defeq \begin{cases}
		1 \text{ if } P(I_{\mathbf{s}_1},\ldots,I_{\mathbf{s}_L}) \text{ is a completion of } D \text{ that satisfies } q \\
		0 \text{ otherwise}
	\end{cases}$.
	
	As such we cannot evaluate this expression in \ptime.
	The next step is to show that the value of~$\checkp(I_{\mathbf{s}_1},\ldots,I_{\mathbf{s}_L})$ only depends on~$(|I_{\mathbf{s}_1}|,\ldots,|I_{\mathbf{s}_L}|)$, which would allow us to rewrite the result as
	\begin{equation}
		\label{eq:final}
		\sum_{0 \leq i_{\mathbf{s}_1},\ldots,i_{\mathbf{s}_L} \leq d} \prod_{1 \leq j \leq L} 
		\binom{d - c - \sum_{1\leq k < j} i_{\mathbf{s}_k} + \sum_{\emptyset \neq \mathbf{s'}\subsetneq \mathbf{s}} C_\mathbf{s'}}{i_{\mathbf{s}_j}}
		\times \checkp(i_{\mathbf{s}_1},\ldots,i_{\mathbf{s}_L})
	\end{equation}

	We give here the necessary and sufficient conditions for~$P(I_{\mathbf{s}_1},\ldots,I_{\mathbf{s}_L})$ to be a completion of~$D$ that satisfies~$q$.

	\begin{lemma}
	We have~$\checkp(I_{\mathbf{s}_1},\ldots,I_{\mathbf{s}_L}) = 1$ if and only if the following conditions hold:
		\begin{enumerate}
			\item for every basic singleton query~$C_i(x)$ of~$q$, letting~$\mathbf{s}$ be its sets of relation symbols,
				there exists~$\mathbf{s} \subseteq \mathbf{s'} \subseteq \sigma$ such that we have~$|I_\mathbf{s'}|\geq 1$ or~$c_\mathbf{s'}\geq 1$.
			\item for every~$\emptyset \neq \mathbf{s} \subseteq \sigma$, if~$n_\mathbf{s} \geq 1$ and
				$|\bigcup_{\mathbf{s'} \supseteq \mathbf{s}} C_\mathbf{s'} \cup \bigcup_{\mathbf{s'} \supsetneq \mathbf{s}} I_\mathbf{s'}| = 0$ then~$|I_\mathbf{s}| \neq 0$.
			\item consider the following system of equations, with integer variables between~$0$ and~$d$:
		\begin{itemize}
			\item for every two sets~$A,A'$ of subsets of~$\{\emptyset \neq \mathbf{s} \subseteq \sigma \}$, we have a variable~$z_{N_\mathbf{s}}^{A,A'}$ for every~$\mathbf{s} \in A$ and a variable~$z_{C_\mathbf{s}}^{A,A'}$ for every~$\mathbf{s} \in A'$.
				For instance if~$\sigma = \{R,S,T,U\}$ and if~$A = \{\{R,S\},\{S,T\}\}$ and~$A' = \{\{U\}\}$ we have the variables~$z^{\{\{R,S\},\{S,T\}\},\{\{U\}\}}_{N_{\{R,S\}}}$ and $z^{\{\{R,S\},\{S,T\}\},\{\{U\}\}}_{N_{\{S,T\}}}$ and $z^{\{\{R,S\},\{S,T\}\},\{\{U\}\}}_{C_{\{U\}}}$. The intuition is that we will use $z^{\{\{R,S\},\{S,T\}\},\{\{U\}\}}_{N_{\{R,S\}}}$ of the nulls in~$N_{\{R,S\}}$ and combine them with $z^{\{\{R,S\},\{S,T\}\},\{\{U\}\}}_{N_{\{S,T\}}}$ of the nulls in~$N_{\{S,T\}}$ and with $z^{\{\{R,S\},\{S,T\}\},\{\{U\}\}}_{C_{\{U\}}}$ of the constants in~$C_{\{U\}}$ in order to obtain constants in~$I_{\{R,S,T,U\}}$. Let us write~$V$ this set of variables.
				(we note here that we are using sligthly different notation than for the last warm-up example; this is for readability reasons only.)
			\item Now, for every~$\emptyset \neq \mathbf{s} \subseteq \sigma$ we have the constraint
				\[\sum_{z_{N_\mathbf{s}}^{A,A'} \in V} z_{N_\mathbf{s}}^{A,A'} \leq n_\mathbf{s}\]
				as well as the constraint
				\[\sum_{z_{C_\mathbf{s}}^{A,A'} \in V} z_{N_\mathbf{s}}^{A,A'} \leq c_\mathbf{s}\]
				intuitively expressing that we do not use more nulls and constants than there are available.
			\item for every~$\emptyset \neq \mathbf{s} \subseteq \sigma$ we have a constraint

				\[\sum_{\substack{A,A' \subseteq \{\emptyset \neq \mathbf{s} \subseteq \sigma\}\\A\cup A' = \mathbf{s}}} \min_{z_{*}^{A,A'} \in V} z_{*}^{A,A'}
				\geq I_\mathbf{s}
				\]
								intuitively meaning that we have allocated the groups of nulls and constants in a way that allows us to fill the set~$I_\mathbf{s}$.
		\end{itemize}
		Then this system of equations must have a solution.
		\end{enumerate}
	\end{lemma}
	\begin{proof}
		The idea is the same as in Claim~\ref{claim:expl-conditions}. The only difference is that we added condition (1), which ensures that the guessed completion indeed satisfies the query.
	\end{proof}

	As in the last warm-up example, this implies that the value of~$\checkp(I_{\mathbf{s}_1},\ldots,I_{\mathbf{s}_L})$ only depends on~$(|I_{\mathbf{s}_1}|,\ldots,|I_{\mathbf{s}_L}|)$ and can be computed in FP (by testing all assignments of the~$z_*^*$ variables; keep in mind that the schema is fixed so there are only a fixed number of such variables).
	But then we can compute the result in FP by evaluating the expression~\ref{eq:final}, which finishes the proof. 
\end{toappendix}

The proof of the tractability part of this theorem is combinatorial and very technical, and we present it in Appendix~\ref{apx:ucc-naive-fp}, where we also give several examples to provide the main intuitions.
Note that, as in the last section, we do not claim membership in~\shp. However, 
and also as in the last section, 
we can show that these problems are in \shp\ for Codd tables, which allows us to obtain our last dichotomy for exact counting.

\begin{toappendix}
	\subsection{Proof of Theorem~\ref{thm:ucc-codd-fp}}
\end{toappendix}
\begin{theoremrep}[Dichotomy]
	\label{thm:ucc-codd-fp}
	Let~$q$ be an \sjfbcq.
	If $R(x,x)$ or~$R(x,y)$ is a pattern of $q$, then 
$\cucountcompls(q)$ is \shp-complete. Otherwise, this problem is in~\fp.
\end{theoremrep}
\begin{proof}
	Hardness follows from Theorem~\ref{thm:ucc-naive-fp}, while membership in \shp\ follows from the result proven in Appendix~\ref{apx:countcompls-codd-sharp-p}.
\end{proof}

\section{Approximating the numbers of valuations and completions}
\label{sec:approx} 

As we saw in the previous sections, counting valuations and completions of an incomplete database are usually intractable problems. 
However, this does not necessarily rule out the existence of efficient approximation algorithms for such counting problems, in particular if some source of randomization is allowed. 
In this section, we investigate this question by focusing on the well-known notion of Fully Polynomial-time Randomized Approximation Scheme (FPRAS) for counting problems~\cite{jerrum1986random}.
Formally, 
let~$\Sigma$ be a finite alphabet and $f : \Sigma^* \to \mathbb{N}$ be a counting problem. 
Then~$f$
is said to have an FPRAS if there is a randomized algorithm~$\mathcal{A} : \Sigma^* \times (0,1) \to \mathbb{N}$ and a polynomial~$p(u,v)$ such that, given~$x \in \Sigma^*$ and~$\epsilon \in (0,1)$,
algorithm~$\mathcal{A}$ runs in time~$p(|x|,\nicefrac{1}{\epsilon})$ and satisfies the following condition:
\begin{eqnarray*}
\Pr\big(|f(x) - \mathcal{A}(x,\epsilon)| \, \leq \, \epsilon f(x)\big) & \geq &\frac{3}{4}.
\end{eqnarray*}
Observe that the property of having an FPRAS is closed under polynomial-time parsimonious reductions, that is, if we have an FPRAS for a counting problem~$A$ and for counting problem~$B$ we have that~$B \pr A$, then we also have an FPRAS for~$B$.

In the following sections, we investigate the existence of FPRAS for the problems of counting valuations and completions of an incomplete database. 
We first deal with counting valuations in Section~\ref{sec:approx-countvals}, where we show a general condition under which
this problem has an FPRAS (which will apply, in particular, to all Boolean conjunctive queries). Then, in Section~\ref{sec:approx-countcompls}, we show that the situation is quite different for counting completions, as in most cases this problem does not admit an FPRAS.

\subsection{Approximating the number of valuations}
\label{sec:approx-countvals}

To prove the main result of this section, we need to consider the 
counting complexity class \spanl~\cite{alvarez1993very}. Given a finite alphabet~$\Sigma$, an \nl-transducer~$M$ over~$\Sigma$ is a nondeterministic Turing Machine with input and output alphabet~$\Sigma$, a read-only input tape, a write-only output tape (where the head is always moved to the right once a symbol in~$\Sigma$ is written in it, so that the output cannot be read by~$M$), and a work-tape of which, on input~$x$, only the first~$c \cdot \log(|x|)$ cells can be used for a fixed constant~$c > 0$ (so that the space used by~$M$ is logarithmic). Moreover,~$y \in \Sigma^*$ is said to be an output of~$M$ with input~$x$, if there exists an accepting run of~$M$ with input~$x$ such that~$y$ is the string in the output tape when~$M$ halts. Then a function $f: \Sigma^* \to \mathbb{N}$ is said to be in \spanl\ if there exists an \nl-transducer~$M$ over~$\Sigma$ such that for every~$x \in \Sigma^*$, the value~$f(x)$ is equal to the number of distinct outputs of~$M$ with input~$x$. 
In \cite{alvarez1993very}, it was proved that $\spanl \subseteq \shp$, and also that this inclusion is strict unless~$\nl = \np$.

Very recently, the authors of~\cite{arenas2019efficient} have shown 
that every problem in \spanl\ has an~FPRAS.
\begin{theorem}[{\cite[Corollary 3]{arenas2019efficient}}]
\label{thm:spanl-fpras}
	Every problem in \spanl\ has an FPRAS.
\end{theorem}
By using this result, we can give a general condition on a Boolean query~$q$ under which~$\countvals(q)$ has an FPRAS, as this condition ensures that~$\countvals(q)$ is in \spanl. More precisely,
a Boolean query~$q$ is said to be \emph{monotone} if for every pair of (complete) databases~$D$, $D'$ such that~$D \subseteq D'$, if $D \models q$, then $D' \models q$. Moreover,~$q$ is said to have \emph{bounded minimal models} if there exists a constant~$C_q$ (that depends only on~$q$) satisfying that for every (complete) database~$D$, if~$D \models q$, then there exists
$D'\subseteq D$ such that~$D' \models q$ and the number of facts in~$D'$ is at most~$C_q$.
Finally, the \emph{model checking problem} for~$q$, denoted by~$\mc(q)$, is the problem of deciding, given a (complete) database~$D$, whether~$D \models q$. Then~$q$ is said to have a model checking in a complexity class~$\CC$ if $\mc(q) \in \CC$. With this terminology, we can state the main result of this section. 
	\begin{toappendix}
		\subsection{Proof of Proposition~\ref{prp:countvals-in-spanl}}
	\end{toappendix}
\begin{propositionrep}
\label{prp:countvals-in-spanl}
\begin{sloppypar}
	Assume that a Boolean query~$q$ is monotone, has model checking in nondeterministic linear space, and has bounded minimal models.
	Then~$\countvals(q)$
	is in~\spanl.
\end{sloppypar}
\end{propositionrep}
\begin{proof}
	Let~$D$ be the input incomplete database, with the domains for each null.
	First, the machine guesses a subset~$D'\subseteq D$ of size~$\leq C_q$, such that each fact of~$D'$ is over a relation symbol that appears in~$q$.
	Observe that~$D'$ contains at most~$|D'| \times \arity(q) \leq C_q \times \arity(q)$ distinct nulls, and that this is a constant.
	The machine then guesses and remembers a valuation~$\nu$ of~$D'$ and computes~$\nu(D')$.
	The encoding size~$||\nu(D')||$ of~$\nu(D')$ is~$O(\log |D|)$, so the machine can check in nondeterministic linear space
	whether~$\nu(D') \models Q$, and stops and rejects in the branches that fail the test.
	Then, the machine reads the input tape left to right and for every occurrence of a null~$\bot$ (appearing in~$D$) that it finds, it does the following:
	\begin{itemize}
		\item It checks whether~$\bot$ appears before on the input tape and if so it simply continues;
		\item Else if~$\bot$ does not appear before on the input tape but appears in~$D'$ then the machine writes~$\nu(\bot)$ on its output tape;
		\item Else if~$\bot$ does not appear before on the input tape and does not appear in~$D'$ then it guesses a
			value for it and writes that value on the output tape (but it does not remember that value).
	\end{itemize}
	It is easy to see that this procedure can be carried out by a logspace nondeterministic transducer, so we only need to show that the distinct outputs of the machine
	correspond exactly to the distinct valuations~$\nu$ of~$D$ such that~$\nu(D) \models Q$.
	Since the machine writes values for nulls in order of first appearance on the input tape, it is clear that every valuation is outputted exactly once.
	Let~$\nu$ be a valuation that is outputted, and let~$D'$ be the subdatabase such that~$\nu(D') \models Q$. Since~$\nu(D') \subseteq \nu(D)$ and~$q$
	is monotone, we have~$\nu(D)\models Q$. Inversely, let~$\nu$ be a valuation of~$D$ such that~$\nu(D) \models Q$, and let us show that it must be outputted.
	Since~$\nu(D) \models Q$ and~$q$ has bounded minimal models, there exists~$D_\nu \subseteq \nu(D)$ of size~$\leq C_q$ such that~$D_\nu \models Q$.
	But~$D_\nu$ is~$\nu(D')$ for some~$D' \subseteq D$ of size~$\leq C_q$.
	Then it is clear that one of the branches of the machine has guessed~$D'$
	and then~$\nu_{|D'}$ and then has written~$\nu$ on the output tape.
\end{proof}

\begin{toappendix}
We note here that the same proof does not work for counting completions.
Informally, there are two complications that arise if we try to modify its proof to make the machine write a
completion on the output tape:
\begin{itemize}
	\item First, in order to write a completion on the output tape, the machine could need to remember the values of a nonconstant number of
		nulls, which it obviously cannot do in logspace;
	\item Second, even if we considered Codd tables in order to avoid the previous complication, there does not seem to be a way to ensure
		that when we write a completion~$\nu(D)$ to the output tape, this completion has not been written already in another branch of the computation but with 
		the tuples written in a different order (in which case this completion would be counted more than once).
\end{itemize}
\end{toappendix}

In particular, given that a union of Boolean of conjunctive queries satisfies the three properties of the previous proposition, we conclude from 
Theorem~\ref{thm:spanl-fpras}
that $\countvals(q)$ can be efficiently approximated by using a randomized algorithm if~$q$ is a union of BCQs.\footnote{As a matter of fact, this holds even for the larger class of unions of BCQs with {\em inequalities} (that is, atoms of the form $x \neq y$), as 
	such queries also satisfy the aforementioned three properties.}

\begin{corollary}
\label{cor:countvals-has-fpras}
	If~$q$ is a union of BCQs, then~$\countvals(q)$ has an FPRAS.
Hence, for such a query $q$, we have that each one of the problems $\ucountvals(q)$, $\ccountvals(q)$, $\cucountvals(q)$ admits an FPRAS.
\end{corollary}

We prove in the next section that the good properties stated in 
Proposition~\ref{prp:countvals-in-spanl} do not hold for counting completions.

\subsection{Approximating the number of completions}
\label{sec:approx-countcompls}

In this section, we prove that the problem of counting completions of an incomplete database is much harder in terms of approximation than the problem of counting valuations. In this investigation, two randomized complexity classes play a fundamental role. Recall that~$\rp$ is the class of decision problems~$L$ for which there exists a polynomial-time probabilistic Turing Machine~$M$ such that: (a) if~$x \in L$, then~$M$ accepts with probability at least~$\nicefrac{3}{4}$; and (b) if~$x \not\in L$, then~$M$ does not accept~$x$. Moreover,~$\bpp$ is defined exactly as~$\rp$ but with condition (b) replaced by: (b') if~$x \not\in L$, then~$M$ accepts with probability at most~$\nicefrac{1}{4}$. Thus,~$\bpp$ is defined as~$\rp$ but allowing errors for both the elements that are and are not in~$L$. It is easy to see that~$\rp \subseteq \bpp$. Besides, it is known that~$\rp \subseteq \np$, and this inclusion is widely believed to be strict. Finally, it is not known whether~$\bpp \subseteq \np$ or~$\np \subseteq \bpp$, but it is widely believed that~$\np$ is not included in~$\bpp$.

\paragraph{{\bf The non-uniform case}}
Recall that \sIS\ is the problem of counting the number of independent sets of a graph. This problem
will play a fundamental role when showing non-approximability of counting completions in the non-uniform case. More precisely, the following is known about the approximability of \sIS.
\begin{theorem}[{\cite[Theorem 3.1]{dyer2002counting}}]
\label{thm:IS-no-fpras}
	The problem \sIS\ does not admit an FPRAS unless $\np=\rp$. 
\end{theorem}
In the proof of Proposition~\ref{prp:countcompls}, we considered the problem $\sVC$ of counting the number of vertex covers of a graph $G=(V,E)$, and showed that $\sVC \pr \ccountcompls(R(x))$. By observing that $S \subseteq V$ is an independent set of~$G$ if and only if~$V \setminus S$ is a vertex cover of~$G$, we can conclude that $\sIS(G) = \sVC(G)$ and, thus, the same reduction from the proof of Proposition~\ref{prp:countcompls} establishes that $\sIS \pr \ccountcompls(R(x))$.
Therefore, from the fact that 
the reduction in Lemma~\ref{lem:pattern-parsimonious-compls} is also parsimonious and preserves the property of being a Codd table, and the fact that the existence of an FPRAS is closed under polynomial-time parsimonious reductions, we obtain the following result from Theorem~\ref{thm:IS-no-fpras}.

\begin{theorem}[Dichotomy]
\label{thm:countcompls-sjfcq}
	For every \sjfbcq~$q$, it holds that $\ccountcompls(q)$ does not admit an FPRAS unless $\np = \rp$ $($and, hence, the same holds for~$\countcompls(q)$$)$.
\end{theorem}

\paragraph{{\bf The uniform case}}
Recall that from Theorem~\ref{thm:ucc-naive-fp}, we know that if an \sjfbcq~$q$ contains neither $R(x,x)$ nor $R(x,y)$ as a pattern, then $\ucountcompls(q)$ is in \fp. Thus, the question to answer in this section is whether~$\ucountcompls(q)$ and~$\cucountcompls(q)$ can be efficiently approximated if~$q$ contains any of these two patterns. 
For the case of naïve tables, we will give a negative answer 
to this question.
Notice that, this time, our reduction from~\sIS\ in Proposition~\ref{prp:Rxx-Rxy-hard-compls} is not parsimonious, so we cannot use Theorem~\ref{thm:IS-no-fpras} as we did for the non-uniform case.
Instead, we will rely on the following well-known fact: if there exists a \bpp\ algorithm for a problem that is \np-complete,
then 
$\np \subseteq \bpp$, which implies that $\np = \rp$~\cite{K82}.

\begin{proposition}
	\label{prp:ucountcompls-Rxy-no-frpas}
	Neither~$\ucountcompls(R(x,x))$ nor $\ucountcompls(R(x,y))$ admits an FPRAS unless $\np=\rp$. This holds even in the restricted setting 
		in which all nulls are interpreted over the same fixed domain~$\{1,2,3\}$.
\end{proposition}

\begin{proof}
	Let~$G=(V,E)$ be a graph. First, we explain how to construct an incomplete database~$D$ containing a single binary relation~$R$, with uniform domain~$\{1,2,3\}$, and such that (a) all completions of~$D$ satisfy both queries; (b) if~$G$ is~$3$-colorable then~$D$ has~$8$ completions; and (c) if~$G$ is not~$3$-colorable then~$D$ has~$7$ completions. For every node~$u\in V$ we have a null~$\bot_u$.
	The database~$D$ consists of the following three disjoint sets of facts:
		\begin{itemize}
			\item For every edge~$\{u,v\} \in E$, we have the two facts~$R(\bot_u,\bot_v)$ and~$R(\bot_v,\bot_u)$; we call these the \emph{encoding facts}.
			\item We have the facts~$R(1,2),R(2,1),R(2,3),R(3,2),R(1,3)$, and $R(3,1)$; we call these the \emph{triangle facts};
			\item We have six fresh nulls~$\bot_1,\bot'_1,\bot_2,\bot'_2,\bot_3,\bot'_3$ and the facts $R(\bot_i,\bot'_i)$ and $R(\bot'_i,\bot_i)$ for~$1 \leq i \leq 3$; we call these the \emph{auxiliary facts};
			\item Last, we have a fact~$R(c,c)$, where~$c$ is a fresh constant.
		\end{itemize}
		It is clear that all the completions of~$D$ satisfy both queries, so we only need to prove~(b) and~(c).
		Observe that a candidate completion of~$D$ can be equivalently seen as an undirected graph, \emph{possibly with self-loops}, over the nodes~$\{1,2,3\}$ (we omit the fact~$R(c,c)$ since it is in every completion) and that contains the triangle. Thanks to the auxiliary facts, it is easy to show that all such graphs with at least one self-loop can be obtained as a completion of~$D$. For instance, the completion that is triangle with a self-loop only on~$1$ can be obtained by assigning~$1$ to all the nulls in the coding facts, assigning~$1$ to~$\bot_1$, $\bot'_1$,~$\bot_2$ and~$\bot_3$ and assigning~$2$ to~$\bot'_2$ and~$\bot'_3$.
		There are~$7$ such completions. Then, the completion whose graph is the triangle with no self-loops is obtainable if and only if~$G$ is~$3$-colorable (we assign a~$3$-coloring to the nulls in the coding facts, and assign~$1$ to~$\bot_i$ and~$2$ to~$\bot'_i$ for every~$i \in \{1,2, 3\}$). This indeed proves~(b) and~(c).
		Next, we show that any FRPAS with~$\epsilon=\nicefrac{1}{16}$ for counting the number of completions of~$D$ would yield a \bpp\ algorithm to solve $3$-colorability, thus implying~$\np=\rp$ since 
		$3$-colorability
		is an \np-complete problem.

	Let~$\mathcal{A}$ be an FPRAS for~$\ucountcompls(q)$, with~$q$ being $R(x,x)$ or~$R(x,y)$, and let us define a \bpp\ algorithm~$\mathcal{B}$ for~$3$-colorability using~$\mathcal{A}$. On input graph~$G$, algorithm $\mathcal{B}$ does the following. First, it computes in polynomial time the naïve table~$D$ as described above.
	Then~$\mathcal{B}$ calls~$\mathcal{A}$ with input~$(D,\nicefrac{1}{16})$, and if~$\mathcal{A}(D,\nicefrac{1}{16}) \geq 7.5$, then~$\mathcal{B}$ accepts, otherwise $\mathcal{B}$ rejects.
	We now prove that~$\mathcal{B}$ is indeed a \bpp\ algorithm for~$3$-colorability. Assume first that~$G$ is~$3$-colorable. Then by~(b) and by definition of what is an FPRAS, we have that~$\Pr\big(|8 - \mathcal{A}(D,\nicefrac{1}{16})| \leq \nicefrac{8}{16} \big) \geq \frac{3}{4}$. This implies in particular that
	$\Pr\big(\mathcal{A}(D,\nicefrac{1}{16}) \geq 8 -\nicefrac{8}{16} \big) \geq \frac{3}{4}$. Since~$8-\nicefrac{8}{16} = 7.5$ we conclude that if~$G$ is~$3$-colorable, then~$\mathcal{B}$ accepts with probability at least~$\nicefrac{3}{4}$.
	Next, assume that~$G$ is not~$3$-colorable. Then by~(c) we have that $\Pr\big(|7 - \mathcal{A}(D,\nicefrac{1}{16})| \leq \nicefrac{7}{16} \big) \geq \frac{3}{4}$. This implies in particular that~$\Pr\big(\mathcal{A}(D,\nicefrac{1}{16}) \leq 7 +\nicefrac{7}{16} \big) \geq \frac{3}{4}$. Since~$7+\nicefrac{7}{16} < 7.5$, this implies in particular that~$\Pr\big(\mathcal{A}(D,\nicefrac{1}{16}) < 7.5 \big) \geq \frac{3}{4}$. From this, we conclude that if~$G$ is not~$3$-colorable, then~$\mathcal{B}$ rejects with probability at least~$\nicefrac{3}{4}$. This concludes the proof of the~proposition.
\end{proof}

By observing again that the reduction in Lemma~\ref{lem:pattern-parsimonious-compls} is parsimonious, and that the existence of an FPRAS is closed under parsimonious reductions, we obtain that $\ucountcompls(q)$ cannot be efficiently approximated if $q$ contains $R(x,x)$ or $R(x,y)$ as a pattern.

\begin{theorem}[Dichotomy]
	Let $q$ be an \sjfbcq. If $q$ has $R(x,x)$ or $R(x,y)$ as a pattern, then $\ucountcompls(q)$ does not admit an FPRAS unless $\np=\rp$. Otherwise, 
	this problem 
	is in~\fp\ $($by Theorem~\ref{thm:ucc-naive-fp}$)$.
\end{theorem}

Up until now, we do not know if this result still holds for Codd tables, or if it is possible to design an FPRAS in this setting.
We leave this question open for future research.

\section{On the general landscape: beyond~\shp}
\label{sec:misc}

Recall that, when studying the complexity of counting completions for \sjfbcqs\ in Section~\ref{sec:countcompls-sjfcqs}, we did not show that these problems are in~\shp\ for naïve tables. The goal of this section is then twofold. First, we want to give formal evidence that we indeed could not show membership in \shp\ in 
Section~\ref{sec:countcompls-sjfcqs}. Second, we want to identify a counting complexity class that is more appropriate to describe the complexity of $\countcompls(q)$, but in a more general setting which is not based on some syntactic restrictions imposed on $q$. More precisely, for the second goal, we consider the complexity of the model checking problem $\mc(q)$ for $q$, which is defined in Section~\ref{sec:approx-countvals} as the problem of deciding, given a (complete) database $D$, whether~$D \models q$. 
In this section, all upper bounds will be proved for the most general scenario 
of non-uniform naïve tables,
while all lower bounds will 
be proved for the most restricted scenario of uniform naïve tables with 
the fixed domain~$\{0,1\}$.

To meet our first goal, we need to define the complexity class \spp\ introduced in \cite{gupta1991power,OH93,FFK94}. Given a nondeterministic Turing Machine $M$ and a string $x$, let $\acc_M(x)$ (resp., $\rej_M(x)$) be the number of accepting (resp., rejecting) runs of $M$ with input $x$, and let $\gap_M(x) = \acc_M(x) - \rej_M(x)$.
Then a language $L$ is said to be in \spp~\cite{FFK94} if there exists a polynomial-time nondeterministic Turing Machine $M$ such that: (a) if $x \in L$, then $\gap_M(x) = 1$; and (b) if $x \not\in L$, then $\gap_M(x) = 0$. In this way, \spp\ is the smallest class that can be defined in terms of the gap function $\gap_M$. It is conjectured that $\np \not\subseteq \spp$ as, for example, for every known polynomial-time nondeterministic Turing Machine $M$ accepting an \np-complete problem, the function $\gap_M$ is not bounded. In the following proposition, we show how this conjecture helps us to reach our first goal. 

\begin{toappendix}
\subsection{Proof of Proposition~\ref{prp:not-in-shp}}
\end{toappendix}

\begin{propositionrep}
	\label{prp:not-in-shp}
	There exists an \sjfbcq~$q$ such that $\ucountcompls(q)$ is not in \shp\ unless $\np \subseteq \spp$.
\end{propositionrep}

\begin{toappendix}
	The proof of this result relies on the proof of Theorem~\ref{thm:general-spanp-complete} (we presented the results in this order in the main text for narrative purposes).
	Let $q$ be the \sjfbcq\ defined in Equation~\eqref{eq:q} in the proof of Theorem~\ref{thm:general-spanp-complete}.
	Its schema $\sigma = \{S\} \cup \{C_{abc} \mid (a,b,c) \in \{0,1\}^3\}$ consists of~$10$ relation symbols, with~$S$ being binary and each~$C_{abc}$ being ternary.
	Let us denote by~$\ucountcompls(\sigma)$ the problem that takes as input an incomplete database over schema~$\sigma$ and outputs its number of completions.
	The first part of our proof is to reduce~$\ucountcompls(\sigma)$ to~$\ucountcompls(q)$:

	\begin{lemma}
		\label{lem:q3-to-sigma}
		We have that~$\ucountcompls(\sigma) \pr \ucountcompls(q)$.
	\end{lemma}
	\begin{proof}
		Let~$D$ be an incomplete database over schema~$\sigma$, that is an input of~$\ucountcompls(\sigma)$. 
		We construct in polynomial time an incomplete database~$D'$ over the same schema such that~$\ucountcompls(\sigma)(D) = \ucountcompls(q)(D')$, thus establishing the parsimonious reduction.	
		Let~$f$ be a fresh constant that does occurs neither in~$D$ nor in the domain of some null.
		Then the relation~$D'(S)$ is the same as the relation~$D(S)$, plus a fact~$S(f,f)$.
		Moreover, for every~$(a,b,c) \in \{0,1\}^3$, the relation~$D'(C_{abc})$ consits of all the facts in~$D(C_{abc})$, plus a fact~$C_{abc}(f,f,f)$.
		It is easy to see that~$D$ and~$D'$ have the same number of completions. Moreover, thanks to the facts that use the constant~$f$,
		we have that every completion of~$D'$ satisfies~$q$. Therefore, we indeed have that~$\ucountcompls(\sigma)(D) = \ucountcompls(q)(D')$.
	\end{proof}

For the second part of the proof, we need to introduce the complexity class \gapp. This class consists of function problems that can be expressed as the difference of two functions in $\shp$~\cite{FFK94,G95}. It is known that if the inclusion $\spanp \subseteq \gapp$ holds, then we have that $\np \subseteq \spp$~\cite{MTV94}.\footnote{In fact, the class \gapspanp\ is defined in \cite{MTV94}, where it is proved that a function $f$ is in \gapspanp\ if and only if $f = g - h$, where $h,g$ are functions in \spanp. Then it is shown in \cite[Corollary 3.5]{MTV94} that the inclusion $\gapspanp \subseteq \gapp$ implies that $\np \subseteq \spp$. But if we have that $\spanp \subseteq \gapp$, then we also have that $\gapspanp \subseteq \gapp$ as $\gapp$ is closed under subtraction and, therefore, we conclude that $\np \subseteq \spp$ as desired.} With this, we are able to prove the proposition.

\begin{proof}[Proof of Proposition~\ref{prp:not-in-shp}]
	Assume that~$\ucountcompls(q)$ is in $\shp$. Then, by Lemma~\ref{lem:q3-to-sigma} we have that~$\ucountcompls(\sigma) \in \shp$ as well (because~$\shp$ is closed under polynomial-time parsimonious reductions).
	Now, observe that for every incomplete database~$D$ over~$\sigma$, the following holds:
	\begin{eqnarray*}
	\ucountcompls(\lnot q)(D) & = & \ucountcompls(\sigma)(D) - \ucountcompls(q)(D).
	\end{eqnarray*}
	But then this means that~$\ucountcompls(\lnot q)$ is in \gapp\ (since both problems in the right hand side are in~$\shp$).
	Since~$\ucountcompls(\lnot q)$ is \spanp-complete by Theorem~\ref{thm:general-spanp-complete} under polynomial-time parsimonious reductions, and since \gapp\ is closed under polynomial-time parsimonious reductions, this would indeed imply that $\spanp \subseteq \gapp$ and, hence, that $\np \subseteq \spp$.
\end{proof}

\end{toappendix}

To meet our second goal, we need to introduce one last counting complexity class. The class \spanp~\cite{kobler1989counting} is defined exactly as the class \spanl\ introduced in Section~\ref{sec:approx-countvals}, but considering \emph{polynomial-time} nondeterministic Turing machines with output, instead of \emph{logarithmic-space} nondeterministic Turing machines with output. It is straightforward to prove that $\shp \subseteq \spanp$. Besides, it is known that $\shp = \spanp$ if and only if $\np = \up$~\cite{kobler1989counting}.\footnote{Recall that \up\ is the class Unambiguous Polynomial-Time introduced in \cite{V76}, and that $L \in \up$ if and only if there exists a polynomial-time nondeterministic Turing Machine $M$ such that if $x \in L$, then $\acc_M(x) = 1$, and if $x \not\in L$, then $\acc_M(x) = 0$.} Therefore, it is widely believed that $\shp$ is properly included in $\spanp$.
The following easy observation can be seen as a first hint that $\spanp$ is a good alternative to describe the complexity of counting completions.

\begin{observation}
\label{obs:general-upper-bound}
If $q$ is a Boolean query such that $\mc(q)$ is in \ptime, then $\countcompls(q)$ is in $\spanp$.
\end{observation}

Notice that this result applies to all~\sjfbcqs\ and, more generally, to all FO Boolean queries. In fact, this results applies to even more expressive query languages such as Datalog~\cite{abiteboul1995foundations}. More surprisingly, in the following theorem we show that $\ucountcompls(q)$ can be $\spanp$-complete for an FO query $q$ and, in fact, already for the negation of an~\sjfbcq.

\begin{toappendix}
\subsection{Proof of Theorem~\ref{thm:general-spanp-complete}}
\end{toappendix}

\begin{theoremrep}
\label{thm:general-spanp-complete}
\begin{sloppypar}
	There exists an \sjfbcq\ $q$ such that $\ucountcompls(\neg q)$ is \spanp-complete under polynomial-time parsimonious reductions. 
\end{sloppypar}
\end{theoremrep}

\begin{toappendix}
\begin{proof}
Notice that we only need to show that $\ucountcompls(\neg q)$ is \spanp-hard under parsimonious reductions, for a fixed \sjfbcq\ $q$.
In this proof, we reduce from counting the number of satisfying assignments of a 3-CNF formula that are distinct in the first~$k$ variables, that we denoted by~$\ksat$:

\begin{definition}
	The problem~$\ksat$ takes as input a 3-CNF formula~$F$ on variables~$\{x_1,\ldots,x_n\}$ and an integer
	$1 \leq k \leq n$, and outputs the number of assignments of the first~$k$ variables that can be extended
	to a satisfying assignment of~$F$.
\end{definition}
	\begin{proposition}[{\cite[Section 6]{kobler1989counting}}]
		$\ksat$ is \spanp\ complete (under polynomial-time parsimonious reductions).
	\end{proposition}
We reduce from $\ksat$ to $\ucountcompls(\neg q)$, for a fixed \sjfbcq\ $q$ to be defined.
	Let~$F$ be a 3-CNF on variables~$\{x_1,\ldots,x_n\}$, and~$1 \leq k \leq n$.
	We first explain how we build the incomplete database~$D$, and we will define the sjfBCQ~$q$ after.
	For every variable~$x_i$, $1 \leq i \leq n$, we have a null~$\bot_{x_i}$, and the (uniform) domain
	is~$\{0,1\}$.
	For~$(a,b,c) \in \{0,1\}^3$, we have a relation~$C_{abc}$ of arity~$3$, and we fill
	it with every tuple of the form~$C_{abc}(a',b',c')$ with~$(a',b',c') \in \{0,1\}^3$
	such that~$a=a' \lor b=b' \lor c=c'$ holds; hence for every~$(a,b,c) \in \{0,1\}^3$ there are exactly~$7$ facts of this form.
	For every clause~$K=l_1 \lor l_2 \lor l_3$ of~$F$ with~$l_1,l_2,l_3$ being
	literals over variables~$y_1,y_2,y_3$,
	letting~$(a_1,a_2,a_3) \in \{0,1\}^3$ be the unique tuple such that~$a_i=1$ iff~$l_i$ is
	a positive literal, we add to~$C_{a_1 a_2 a_3}$ the
	fact~$C_{a_1 a_2 a_3}(\bot_{y_1},\bot_{y_2},\bot_{y_3})$.
	Last, we have a binary relation~$S$ that we fill with the tuples~$S(i,\bot_{x_i})$
	for~$1 \leq i \leq k$.
	The sjfBCQ~$q$ then simply says that there exists a tuple that appears in all the relations~$C_{abc}$:
	\begin{eqnarray}\label{eq:q}
	q & = & \exists x \exists y\, S(x,y) \land \exists x \exists y \exists z\, \bigg(\bigwedge_{(a,b,c)\in \{0,1\}^3} C_{abc}(x,y,z)\bigg)
	\end{eqnarray}
	Note that we added the seemingly useless query $\exists x \exists y\, S(x,y)$ to~$q$ because the set of relations in $D$ has to be equal to the set of relations occurring in $q$.
	We now show that the number of completions of~$D$ that do not satisfy~$q$ is equal to the number of assignments of the first~$k$ variables that can be extended
	to a satisfying assignment of~$F$, thus establishing that~$\ucountcompls(\lnot q)$ is \spanp-hard (under polynomial-time parsimonious reductions).
	First, observe that the assignments of the variables are in bijection with the valuations of the nulls of~$D$.
	One can then readily observe the following:
	 \begin{itemize}
		 \item If~$q$ is falsified in a completion of~$D$, it can only be because there does not exist a tuple that occurs in all the relations; this is because the query~$\exists x,y\, S(x,y)$ is always satisfied by any completion of~$D$.
		 \item For every assignment of the variables, letting~$\nu$ be the corresponding valuation of the nulls, there exists a tuple that is in all relations~$C_{abc}$ of~$\nu(D)$ if and only if that assignment is not satisfying for~$F$. Indeed, one can check that this happens if and only if there exists a relation~$C_{abc}$ such that~$\nu(D)(C_{abc})$ contains exactly~$8$ facts.
		 \item For every two valuations~$\nu,\nu'$ such that the corresponding assignments are not satisfying the query, we have that~$\nu(D) \neq \nu'(D)$ if and only if~$\nu$ and~$\nu'$ differ on the first~$k$ variables. This is because, by the previous item, each relation~$C_{abc}$ contains exactly the~$7$ ground tuples that we initially put in~$D$.
	 \end{itemize}
	By putting it all together, we obtain that the reduction works as expected.
\end{proof}

\end{toappendix}

This theorem 
gives evidence that $\spanp$ is the right class to describe the complexity of counting completions for FO queries (and even for queries with model checking in polynomial time). 
It is important to notice that \spanp-hardness is proved in Theorem~\ref{thm:general-spanp-complete} by considering parsimonious reductions. This is a delicate issue because from the main result in \cite{TW92}, it is possible to conclude that every counting problem that is \shp-hard (even under polynomial-time parsimonious reductions) is also \spanp-hard under polynomial-time Turing reductions, so a more restrictive notion of reduction has to be used when proving that a counting problem is \spanp-hard~\cite{kobler1989counting}.

We conclude this section by considering an even more general scenario where queries have model checking in \np. Interestingly, in this case \spanp\ is again the right class to describe the complexity not only of counting completions, but also of counting valuations. 

\begin{toappendix}
\subsection{Proof of Theorem~\ref{theo-np-spanp}}
\end{toappendix}

\begin{theoremrep}\label{theo-np-spanp}
	If~$q$ is a Boolean query with~$\mc(q) \in \np$, then both $\countvals(q)$ and $\countcompls(q)$ are in \spanp.
	Moreover, there exists such a Boolean query~$q$ for which $\ucountvals(q)$ is~\spanp-complete under polynomial-time parsimonious~reductions $($and for~$\ucountcompls(q)$, we can even take~$q$ to be the negation of an \sjfbcq, hence with model checking in~\ptime, as given by Theorem~\ref{thm:general-spanp-complete}$)$.
\end{theoremrep}

\begin{toappendix}
\begin{proof}
	It is straightforward to prove that these problems are in \spanp. The part in between parenthesis has been shown in theorem~\ref{thm:general-spanp-complete}. Thus, we need to prove that $\ucountvals(q)$ is \spanp-hard for a fixed Boolean query $q$ such that $\mc(q) \in \np$, under polynomial-time parsimonious reductions. To do this, we will reduce from the \spanp-complete problem $\hamsubgraphs$, defined as follows:
	
	\begin{definition}
	\label{def:hamsubgraphs}
		Let~$G=(V,E)$ be a undirected graph, and let~$S \subseteq V$.
		The \emph{subgraph of~$G$ induced by~$S$}, denoted by~$G[S]$, is the graph with set of
		nodes~$S$ and set of edges $\{\{u,v\} \in E \mid u,v\in S\}$,
		We recall that a graph~$G$ is \emph{Hamiltonian} when there exists a cycle in~$G$ that visits every node of~$G$ exactly once.
		The problem~$\hamsubgraphs$ takes as input a simple graph~$G=(V,E)$ and an integer~$k$, and outputs the number of induced subgraphs~$G[S]$ with~$|S| = k$
		such that~$G[S]$ is Hamiltonian.
	\end{definition}
	
	\begin{proposition}[{\cite[Section 6]{kobler1989counting}}]
	\label{prp:hamsubgraphs-hard}
	$\hamsubgraphs$ is \spanp-complete (under polynomial-time parsimonious reductions).
	\end{proposition}

Next we show that $\hamsubgraphs \pr \ucountvals(q)$, for a fixed Boolean query $q$ (to be defined). 
Let~$G=(V,E)$ be an undirected graph. We first explain how we construct the incomplete database~$D$, and we will then define the query~$q$.
		The schema contains two binary relation symbols~$R,T$ and one unary relation symbol~$K$. Fix a linear order~$a_1,\ldots,a_n$ of the nodes of~$G$.
		For every edge~$\{u,v\} \in E$ we have the facts~$R(u,v)$ and~$R(v,u)$.
		For~$1 \leq i \leq n$ we have a fact~$T(a_i,\bot_i)$, and the domain of the nulls is~$\{0,1\}$.
		For~$1 \leq j \leq k$ we have a fact~$K(j)$. Observe that~$D$ is a Codd table.
		We now define the Boolean query~$q$, which will be a sentence in existential second-order logic ($\exists$SO) over relational signature~$R,T,K$.
		Before doing so, we explain the main idea: intuitively,~$q$ will check that there are exactly~$k$ facts of the form~$T(a_i,1)$ in the relation~$T$ and that, letting~$S$ be the set of nodes~$v$ such that~$T(v,1)$ is in relation~$T$, the induced subgraph~$G[S]$ is Hamiltonian. This will indeed ensure
		that we have $\ucountvals(q)(D) = \hamsubgraphs(G,k)$, thus completing this reduction, which is parsimonious and can be performed in polynomial-time.
		The query is
		\[q \ = \ \exists S\, \psi_1(S) \land \psi_2(S)\]
		where~$S$ is a unary second order variable and the formula $\psi_1(S)$ states that (a) the elements~$s$ of~$S$ are exactly all the elements such that~$T(s,1)$ holds,
		and that (b) there are exactly the same number of elements in~$S$ as there are elements~$j$ for which~$K(j)$ holds.
		It is clear that (a) can be expressed in FO.
		Moreover, (b) can be expressed in $\exists$SO by asserting the existence of a binary second-order relation~$U$
		that represents a bijective function from~$S$ to the elements in~$K$.
		Then~$\psi_2(S)$ is a formula that asserts that~$G[S]$ is Hamiltonian. Since this is a property in NP,~$\psi_2(S)$ can be expressed in~$\exists$SO by Fagin's theorem (see, e.g.,~\cite{immerman2012descriptive}).
		This shows that the reduction is correct.
		Finally, the fact that~$\mc(q)$ is in NP again follows from Fagin's theorem.
		This concludes the proof.

\end{proof}
\end{toappendix}

\section{Related Work}
\label{sec:related}
There are two main lines of work that must be compared to what we do in this paper. In both cases the goal is to go beyond the traditional notion of \emph{certain answers} that so far had been used almost exclusively to deal with query answering over uncertain data.
We discuss them here, explain how they relate to our problems and what are the fundamental differences.

\paragraph{{\bf Best answers and 0-1 laws for incomplete databases}}
Libkin has recently introduced a framework that can be used to measure the certainty with which a Boolean query holds on an incomplete database, and also to compare query answers (for a non-Boolean query)~\cite{libkin2018certain}.  
For a Boolean query~$q$, incomplete database~$D$, and integer~$k$, he defines the quantity~$\mu^k(q,D)$ as~$\frac{|\mathrm{Supp}^k(q,D)|}{|V^k(D)|}$, where~$V^k(D)$ denotes the set of valuations of~$D$ with domain~$\{1,\ldots,k\}$, and $\mathrm{Supp}^k(q,D)$ denotes the set of valuations~$\nu \in V^k(D)$ such that $\nu(D) \models q$; hence,~$\mu^k(q,D)$ represents the relative frequency of valuations~$\nu$ in~$\{1,\ldots,k\}$ for which the query is satisfied. He then shows that, for a very large class of queries (namely, \emph{generic queries}), the value~$\mu^k(q,d)$ always tends to~$0$ or~$1$ as~$k$ tends to infinity (and the same results holds when considering completions instead of valuations). This means that, intuitively,  over an infinite domain the query~$q$ is either almost certainly true or almost certainly false. 

He also studies the complexity of finding best answers for a non Boolean query~$q$. As mentioned in the introduction, 
a tuple~$\overline{a}$ is a better answer than another tuple~$\overline{b}$ when for every valuation~$\nu$ of~$D$, if we have~$\overline{b} \in q(\nu(d))$ then we also have~$\overline{b} \in q(\nu(d))$. A best answer is then an answer such that there is no other answer strictly better than it (under inclusion of the sets of satisfying valuations). He studies the complexity of comparing answers under this semantics, and that of computing the set of best answers (see also \cite{GS19}). 

There are several crucial differences between this previous work and ours. First, Libkin does not study the complexity of computing~$\mu^k(q,d)$. We do this under the name~$\ucountvals(q)$; moreover, we also study the setting in which the domains are not uniform.
Second, knowing that a tuple is the best answer might not tell us anything about the size of its ``support'', i.e., the number of valuations that support it.  In particular, a best answer is not necessarily an answer which has the biggest support. 
Finally, under the semantics of better answers it does not matter if we look at the completions or at the valuations (i.e., a tuple is a best answer with respect to inclusion of valuations iff it is the best answer with respect to completions); while we have shown that it does matter for counting problems.

\paragraph{{\bf Counting problems for probabilistic databases and consistent query answering.}}
Remarkably, 
counting problems have received considerable attention in other database 
 scenarios where uncertainty issues appear.
 As mentioned in the introduction, this includes the settings of probabilistic databases and inconsistent databases. 
 In the former case, uncertainty is represented as a probability distribution on the possible states of the data~\cite{suciu2011probabilistic,dalvi2012dichotomy}.
 There, query answering 
 amounts to computing 
 a weighted sum of the probabilities of the possible states of the data that satisfy a query~$q$. We call this problem {\sf Prob}$(q)$.  
 In the case of 
inconsistent databases, 
 we are given a set~$\Sigma$ of constraints and a 
database~$D$ that does not necessarily satisfy~$\Sigma$; cf.~\cite{ABC99,2011Bertossi,Bertossi19}. Then the task is to reason about the set of all {\em repairs} of~$D$ with respect to~$\Sigma$ \cite{ABC99}.
In our context, this means that one wants to count the number of repairs of~$D$ with respect to~$\Sigma$ that satisfy a given query~$q$. 
When~$q$ and~$\Sigma$ are fixed, we call this problem~$\#${\sf Repairs}$(q,\Sigma)$. 
 
 Both {\sf Prob}$(q)$ and~$\#${\sf Repairs}$(q,\Sigma)$
have been intensively studied already.
To start with, counting complexity dichotomies have been obtained for the problem $\#${\sf Repairs}$(q,\Sigma)$; e.g., \cite{maslowski2013dichotomy} gives a dichotomy for this problem when~$q$ is an \sjfbcq\ and~$\sigma$ consists of primary keys, and~\cite{maslowski2014counting} extends this result to CQs with self-joins but only for unary keys constraints. We also mention~\cite{calautti2019counting}, where the problem of counting repairs such
that a particular input tuple is in the result of the query on the repair is studied. 
A seemingly close counting problem for probabilistic databases is the problem~{\sf Prob}$(q)$ over \emph{block independent disjoint} (BIDs) databases. We do not define it formally here, but counting repairs under primary keys can be seen as a special case of this problem,
where the tuples in a “block” all have the same probability, and where the sum of the probabilities sum to $1$ (and in BIDs this sum is allowed to be $<1$, meaning that a block
can be completely erased). Dichotomies for this problem have been obtained in~\cite{dalvi2011queries} for \sjfbcqs.
Counting complexity dichotomies for other models of probabilistic databases also exist; e.g., for {\em tuple-independent} 
probabilistic databases in which each fact is assigned an independent probability of being part of the actual dataset. 
Interestingly, dichotomies in this case hold for arbitrary unions of BCQs, and thus not just for \sjfbcqs~\cite{dalvi2012dichotomy}. 

In some cases, one can use a problem of the form~$\#${\sf Repairs}$(q,\Sigma)$ (or {\sf Prob}$(q)$) to show the hardness of a problem of the form~$\countvals(q')$. For instance, we explain in Appendix~\ref{apx:alternative_using_repairs} how the \shp-hardness of $\ccountvals(R(x) \land S(x))$ can be deduced from that of $\#${\sf Repairs}$(R'(\underbar{y},x) \land S'(\underbar{z},x))$.\footnote{Where we write, for instance, $R'(\underbar{y},x)$, to mean that the first attribute of~$R$ is a primary key.}
In general however, the problems $\#${\sf Repairs}$(q,\Sigma)$\linebreak and~{\sf Prob}$(q)$ seem to be unrelated to our problems, for the following reasons.
		First, in our setting the nulls can appear anywhere, so there is no notion of primary keys here; hence it seems unlikely that one can design a generic reduction from the problem of counting valuations/completions to the problem of counting repairs.
In fact, it would perfectly make sense to study our counting problems where we add constraints such a functional dependencies.
		Second, in the BID and counting repairs problems, each “valuation” (repair) gives a different complete
database, while in our case we have seen that this is not necessarily the case.
In particular, problems of the form $\countcompls(q)$ have no analogues in these settings, whereas we have seen that they behave very differently in our setting.

	\begin{toappendix}
		\subsection{Alternative proof of~Proposition~\ref{prp:RxSx-hard} using~$\#${\sf Repairs}}
            \label{apx:alternative_using_repairs}
        In this section we explain how Proposition~\ref{prp:RxSx-hard}, i.e., the \shp-hardness of $\ccountvals(R(x) \land S(x))$, can be deduced from the hardness of the problem $\#${\sf Repairs}$(R'(\underbar{y},x) \land S'(\underbar{z},x))$, which was established in~\cite{maslowski2013dichotomy}.\footnote{To see that~\cite{maslowski2013dichotomy} establishes the hardness of~$q=R'(\underbar{y},x) \land S'(\underbar{z},x)$, first apply their rewrite rule R7 (from Fig.~6) to obtain~$q'=R'(\underbar{y},x) \land S'(\underbar{x},x)$, then apply rewrite rule R10 to obtain~$q''=R'(\underbar{y},x) \land S'(\underbar{x},a)$. Then,~$q''$ is \shp-hard by Lemma 19, and so is~$q$ by Lemma 7.}
Let~$D'$ be a database with binary relation~$R',S'$. We construct an incomplete Codd database~$D$ with unary relations~$R,S$ as follows. For every constant~$a$ that appears in the first attribute of~$R'$, we have a tuple~$R(\bot)$ in~$D'$, where~$\bot$ is a fresh null, and we set~$\dom(\bot) = \{b \mid R'(a,b) \in D'\}$.
For every constant~$a$ that appears in the first attribute of~$S'$, we have a tuple~$S(\bot)$ in~$D'$, where~$\bot$ is a fresh null, and we set~$\dom(\bot) = \{b \mid S'(a,b) \in D'\}$. It is then clear that the number of repairs of~$D'$ that satisfy~$R'(\underbar{y},x) \land S'(\underbar{z},x)$ is equal to the number of valuations of~$D$ that satisfy~$R(x) \land S(y)$.
	\end{toappendix}

Concerning approximation results, it is known that the problems $\#${\sf Repairs}$(q,\Sigma)$ and {\sf Prob}$(q)$
admit an FPRAS in some important settings. In particular, when~$q$ is a union of BCQs, this 
holds for $\#${\sf Repairs}$(q,\Sigma)$ when $\Sigma$ is a set of primary keys \cite{calautti2019counting}, and 
for {\sf Prob}$(q)$ over BID and tuple-independent probabilistic databases~\cite{dalvi2011queries}.
We observe here that this is reminescent of our Corollary~\ref{cor:countvals-has-fpras}, which shows that problems of the form~$\countvals(q)$ have an FPRAS for every union of BCQs.

\section{Final Remarks}
\label{sec:conclusion}

Our work aims to be a first step in the study of counting problems over incomplete databases. 
The main conclusion behind our results is that the counting problems studied in this paper 
are particularly hard from a computational 
point of view, especially 
when compared to more positive results obtained in other uncertainty scenarios; e.g., over probabilistic and inconsistent databases. 
As we have shown, a particularly difficult problem in our context is that of counting completions, even in the uniform setting where all nulls have the same domain. 
In fact, Proposition~\ref{prp:Rxx-Rxy-hard-compls} shows that this problem is \shp-hard even in very restricted scenarios, 
and Proposition~\ref{prp:ucountcompls-Rxy-no-frpas}
 that it cannot be approximated by an FPRAS. It seems then that the only way in which one could try to tackle 
 this problem is by developing suitable tractable heuristics, without provable quantitative guarantees, 
 but that work sufficiently well in practical scenarios. An example of this could 
 be developing algorithms that compute ``under-approximations'' for the number of completions of a na\"ive table satisfying a certain 
 \sjfbcq\ $q$. Notice that a related approach has been proposed by Console et al. for constructing under-approximations of the set of certain answers by applying methods based on many-valued logics \cite{CGL16}. 

We plan to continue working on several interesting problems that are left open in this paper. First of all, we would like to pinpoint 
 the  complexity 
of $\countcompls(q)$ when $q$ is an \sjfbcq; in particular, whether this problem is~\spanp-complete for at least one such 
a query. We also want to study whether the non-existence of FPRAS for $\ucountcompls(q)$ established in 
Proposition~\ref{prp:ucountcompls-Rxy-no-frpas}
 continues to hold over Codd tables.
We would also like to develop a more thorough understanding of the role of fixed domains in our dichotomies. 
In several cases, that we have explicitly stated, our lower bounds hold even if nulls in tables are interpreted over a fixed domain. 
Still, in some cases we do not know whether this holds. These include, e.g., Proposition~\ref{prp:RxSxyTy-hard-codd}, Proposition~\ref{prp:countcompls}, and Proposition~\ref{prp:Rxx-Rxy-hard-compls} for the case of Codd tables.
Finally, it would also be interesting to 
 study these counting problems under bag semantics (instead of the set semantics used in this paper), study counting problems for non-Boolean queries as in~\cite{libkin2018certain,GS19}, and consider arbitrary BCQs as opposed to only \sjfbcqs.

\begin{acks}
The third author would like to thank Antoine Amarilli for reminding him of the paper~\cite{cai2012holographic} in~\cite{cstheory_a_variant} and for suggesting to use \#BIS in the proof of Proposition~\ref{prp:RxSxyTy-hard-codd}.
This work was partly funded by the Millennium Institute for Foundational Research on Data (IMFD).
\end{acks}

\bibliographystyle{ACM-Reference-Format}
\bibliography{main}

\appendix

\onecolumn

\end{document}